\documentclass[10pt,a4paper]{article}
\usepackage[utf8]{inputenc}

\usepackage{geometry}
\usepackage{amsthm}
\newtheorem{theorem}{Theorem}
\newtheorem{lemma}{Lemma}

\newtheorem{corollary}{Corollary}
\newtheorem{definition}{Definition}

\newtheorem{claim}{Claim}[theorem]
\theoremstyle{plain}

\usepackage{xcolor}
\usepackage[colorlinks=true,citecolor=black,linkcolor=black,urlcolor=black,bookmarksopen=true]{hyperref}

\usepackage{amssymb}

\newtheorem{assumption}{Assumption}

\usepackage{graphicx}
\usepackage[linesnumbered,ruled,vlined]{algorithm2e}
\usepackage{xifthen}
\usepackage{paralist}
\usepackage{amsmath}
\makeatletter
\newcommand{\xRightarrow}[2][]{\ext@arrow 0359\Rightarrowfill@{#1}{#2}}
\makeatother

\usepackage{bbold}

\usepackage{multirow}

\newcommand{\mc}{\mathcal}

\newcommand{\eg}{e.g., }
\newcommand{\ie}{i.e., }
\newcommand{\wrt}{w.r.t. }
\newcommand{\etal}{\textit{et al.}}
\newcommand{\etc}{etc}
\newcommand{\st}{s.t. }

\newcommand{\secparam}{\kappa}

\newcommand{\adversary}{\mc{A}}

\newcommand{\env}{\mc{Z}}
\newcommand{\totalParties}{n}
\newcommand{\party}{\mc{P}}
\newcommand{\partySet}{\mathbb{P}}
\newcommand{\observer}{\Omega}

\newcommand{\attacker}{\mc{M}}

\newcommand{\miningpower}{\mu}

\newcommand{\hash}{\mathsf{H}}

\newcommand{\chain}{\mc{C}}
\newcommand{\block}{\mc{B}}
\newcommand{\genesis}{\block_G}
\newcommand{\oracle}{\mc{O}}

\newcommand{\mesg}{m}

\newcommand{\msgOutputSet}{M}

\newcommand{\execution}{\mc{E}}

\newcommand{\executionTrace}{\Im}
\newcommand{\proto}{\Pi}
\newcommand{\strategy}{S}
\newcommand{\strategySet}{\mathbb{S}}
\newcommand{\profile}{\sigma}

\newcommand{\messageValidityPredicate}{\mc{V}}
\newcommand{\infractionPredicate}{\mc{X}}

\newcommand{\compliant}{\infractionPredicate\mathsf{-compliant}}

\newcommand{\slot}{r}

\newcommand{\epoch}{e}
\newcommand{\epochLength}{l_\epoch}

\newcommand{\utility}{U}
\newcommand{\reward}{R}
\newcommand{\totalReward}{\mc{R}}
\newcommand{\rewardVal}{\rho}
\newcommand{\proportionalRewardFunc}{\varrho}
\newcommand{\universalRewardFunc}{\xi}
\newcommand{\cost}{C}
\newcommand{\costVal}{c}

\newcommand{\networkLossProb}{d}

\newcommand{\difficulty}{\delta}
\newcommand{\powQueryNum}{q}
\newcommand{\queryCost}{\lambda}

\newcommand{\attestation}{\alpha}
\newcommand{\validator}{V}

\newcommand{\exchangeRate}{X}
\newcommand{\utilityBoost}{B}

\newcommand{\exchangeRateVal}{x}
\newcommand{\utilityBoostVal}{b}
\newcommand{\deposit}{g}

\title{
    Blockchain Nash Dynamics and the Pursuit of Compliance
}

\author{
    Dimitris Karakostas \\ University of Edinburgh \\ d.karakostas@ed.ac.uk
    \and
    Aggelos Kiayias \\ University of Edinburgh and IOHK \\ akiayias@inf.ed.ac.uk
    \and
    Thomas Zacharias \\ University of Edinburgh \\ tzachari@inf.ed.ac.uk
}

\begin{document}

\maketitle

\begin{abstract}
    We study Nash-dynamics in the context of blockchain protocols.
    We introduce a formal model, within which one can assess whether the Nash
    dynamics can lead utility-maximizing participants to defect from the ``honest''
    protocol operation, towards variations that exhibit one or more undesirable
    {\em infractions}, such as abstaining from participation and producing
    conflicting protocol histories. Blockchain protocols that do not lead to
    such infraction states are said to be \emph{compliant}. Armed with this model, we
    evaluate the compliance of various Proof-of-Work (PoW) and Proof-of-Stake
    (PoS) protocol families, with respect to different utility functions and
    reward schemes, leading to the following results:
    \begin{inparaenum}[i)]
        \item PoS ledgers under resource-proportional rewards can be compliant if costs are negligible, but non-compliant if costs are significant;
        \item PoW and PoS under block-proportional rewards exhibit different compliance behavior, depending on the lossiness of the network;
        \item PoS ledgers can be compliant \wrt one infraction, \ie producing
            conflicting messages, but non-compliant (and non-equilibria) \wrt
            abstaining or an attack we call selfish signing;
        \item taking externalities, such as exchange rate fluctuations, into
            account, we quantify the benefit of economic penalties, in the
            context of PoS protocols, in disincentivizing particular infractions.
    \end{inparaenum}
\end{abstract}

\section{Introduction}\label{sec:introduction}

The advent of
Bitcoin~\cite{nakamoto2008bitcoin} brought the economic aspects of  consensus
protocols to the forefront. While classical literature in consensus
primarily dealt with fail-stop or Byzantine ``error models''~\cite{DBLP:journals/jacm/PeaseSL80},
 the pressing question post-Bitcoin is whether the participants' incentives
align with what the consensus protocol asks them
to do.
Motivated by this, a line of work investigated if Bitcoin
is an equilibrium under certain conditions
\cite{KrollDaveyFeltenWEIS2013,kiayias16EC}.
Another pinpointed deviations
that can be more profitable for some players, assuming others follow
the protocol~\cite{FC:EyaSir14,FC:SapSomZoh16,FCW:JLGVM14,CCS:CKWN16}. The research body also includes tweaks towards improving the
blockchain protocol in various settings~\cite{FC:FKORVW19,koutsoupias19www}, game-theoretic studies of pooling behavior~\cite{lewenberg15,CCS:CKWN16,ITCS:ArnWei19},
and equilibria that involve abstaining from the protocol~\cite{DBLP:conf/ec/FiatKKP19} in high cost scenarios.
Going beyond consensus, economic mechanisms
have also been considered in the context of multi-party computation~\cite{CCS:KumMorBen15,FC:DavDowLar19,FC:DavDowLar18}, to
disincentivize ``cheating''.
Finally, various works optimized particular attacks, \eg:
\begin{inparaenum}[i)]
    \item optimal selfish mining strategies~\cite{FC:SapSomZoh16};
    \item a framework~\cite{CCS:GKWGRC16} for quantitatively
    evaluating blockchain parameters and identifying optimal strategies for
    selfish mining and double-spending, taking into account network delays;
    \item alternative strategies~\cite{EPRINT:NKMS15}, that are more
        profitable than selfish mining.
\end{inparaenum}

Though these works provide glimpses on these protocols'
behavior in a game-theoretic perspective, they offer little
guidance on how to design and parameterize new consensus protocols. This problem is
of high importance, given the negative light shed on Bitcoin's
perceived energy inefficiency and carbon footprint~\cite{martin2021energy},
that necessitates alternative designs.
Proof-of-Stake (PoS) is currently the most prominent
alternative to Bitcoin's Proof-of-Work (PoW) mechanism. PoW requires computational
effort to produce valid messages, \ie blocks acceptable by the
protocol. PoS relies on each party's stake, \ie assets they own, so
blocks are created at (virtually) no cost beyond transaction processing.
Interestingly, while it is proven that PoS protocols are Byzantine resilient~\cite{C:KRDO17,EPRINT:CGMV18,EPRINT:GHMVZ17}
and are even equilibriums under
certain conditions~\cite{C:KRDO17},
their security is heavily contested by PoW protocols proponents
via an economic argument termed the \emph{nothing-at-stake} attack~\cite{li2017securing,ethereumFaq,nothing-at-stake-1}. This argument
asserts that maintainers of PoS ledgers can maximize their expected rewards
by producing conflicting blocks when possible.

What merit do these criticisms have?
Participating in a blockchain protocol is a voluntary action that involves
a participant downloading the software and committing resources to run it.
Given the open source nature of these protocols,
nothing prevents the participant from modifying the behaviour of the software
in some way and engage with the other parties following a modified strategy.
There are a number of undesirable adjustments that a participant can do, \eg
i) run the protocol intermittently instead of continuously;
ii) not extend the most recent ledger of transactions they are aware of;
iii) extend simultaneously more than one ledger of transactions.
One can consider the above as fundamental {\em infractions} to the protocol
rules and they may have serious security implications, both in terms of
the consistency and the liveness of the underlying ledger.

To address these issues, many blockchain systems
introduce additional mechanisms on top of
incentives, frequently with only rudimentary game theoretic analysis.
These include: i) rewards for ``uncle blocks'' (Ethereum);
ii) stake delegation (EOS, Polkadot, Cardano~\cite{SCN:KarKiaLar20}), where
users assign their participation rights to delegates or
stake pools;
iii) penalties for
misbehavior, also referred to as \emph{``slashing''} (Ethereum 2.0~\cite{buterin2017casper,casper-incentives}).
Unfortunately, the lack of thorough
analysis of these mechanisms is, naturally, a serious impediment to wider adoption.
For instance, in the case of penalties
employing multiple replicas for redundancy, \ie to increase crash-fault
tolerance, may produce conflicting blocks due to a faulty configuration, if two replicas come alive simultaneously.
However, if a party employs no failover mechanism and experiences network connectivity
issues, it may fail to participate. Furthermore, software or hardware bugs can
always compromise an -- otherwise safe and secure -- configuration.
This highlights the flip side of such penalty mechanisms:
participants may choose to not engage,
(\eg to avoid the risk of forfeiting funds, or because they do not
own sufficient funds to make a deposit), or, if they do engage, they may steer clear
of fault-tolerant sysadmin practices,  which could pose quality of service concerns and hurt the system in the long run.

The above considerations put forth the fundamental question that motivates
our work: {\em How effective are blockchain protocol designs in disincentivizing
 particularly adverse protocol infractions?} In more detail, the question we
ask is whether selfish behavior can lead to specific types of deviations,
taking a blockchain protocol as the initial point of reference of honest ---
compliant --- behavior.

\paragraph{Our Contributions and Roadmap.}\label{sec:contributions}
Our main question relates to the Nash dynamics of blockchain protocols.
In the classical Nash dynamics problem~\cite{rosenthal73},
the question is whether allowing selfish
players to perform step-wise payoff-improving moves leads the system
to an equilibrium, and in how many steps this may happen; \eg
\cite{DBLP:conf/stoc/FabrikantPT04} considers the case of congestion games.
In this perspective, the action space can be seen as a directed graph,
where vertices represent vectors of player strategies and edges correspond
to player moves. Notably, deciding whether the Nash dynamics converge to a
(Nash or sink) equilibrium is particularly difficult, often being a
NP-hard or PSPACE-complete problem~\cite{sink2009}.

This work adapts Nash dynamics to the setting of blockchain protocols,
with a particular focus on studying specific undesirable protocol infractions.
Importantly, instead of asking for convergence, we ask whether the ``cone''
in the directed graph positioned at the protocol contains strategies
from a given infraction set $\infractionPredicate$
(Figure~\ref{fig:cone}). If the cone is free of infractions,
the protocol is deemed $\infractionPredicate$-compliant.
In turn, we also consider
$\epsilon$-Nash-dynamics~\cite{DBLP:journals/geb/ChienS11}, \ie considering
only steps in the graph which represent best responses and improve the participant's payoff more than
$\epsilon$. Armed with this model, we investigate various protocols from a
compliance perspective.

\begin{figure}[h]
    \begin{center}
        \includegraphics[width=0.6\columnwidth]{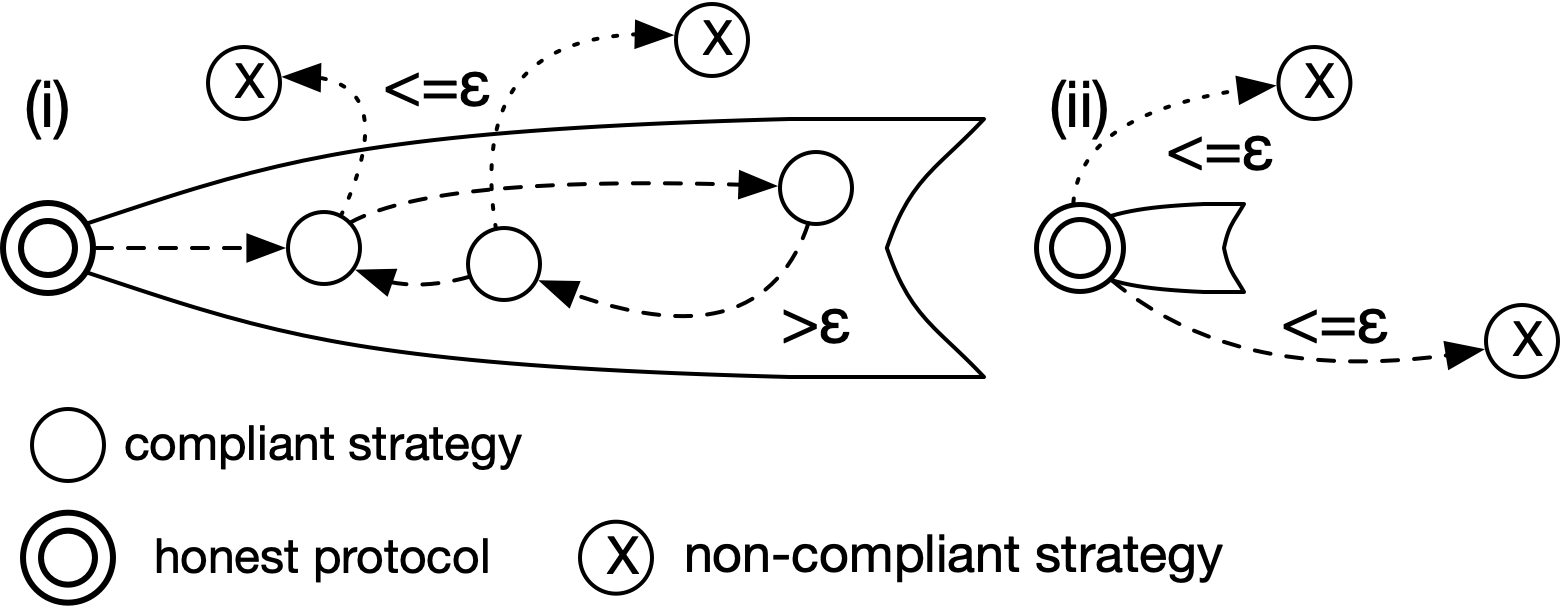}
    \end{center}
    \caption{
        Illustration of a compliant protocol that does not exhibit an
        equilibrium (i), vs a protocol which is an approximate Nash equilibrium
        (ii).
    }
    \label{fig:cone}
\end{figure}

A core motivation of our work is that
$\infractionPredicate$-compliance enables
to validate the incentive structure of a blockchain
protocol \wrt specific disruptive behaviors (as captured by $\infractionPredicate$), while abstracting away any deviations
that avoid the infractions. In this sense, $\infractionPredicate$-compliance of a protocol is a weaker notion compared to a Nash equilibrium, allowing a variety of possible protocol deviations as long as they do not fall into $\infractionPredicate$. This also enables a two-tiered analysis where  compliance analysis rules out crucial deviations, while the set of all compliant behaviors can be analyzed in, say, worst-case fashion.
Moreover, negative results in terms of compliance are immediately informative as they identify one or more specific infractions. This is helpful from the point of view of mechanism parameterisation for blockchain protocols, where various penalties (\eg reward reduction or slashing of funds) are typically employed to mitigate specific deviant behaviors. So far, there exists no framework that enables a formal argument as to whether a specific penalty is sufficient to mitigate a certain behavior. Our work provides such framework and we illustrate its applicability in this setting, by analyzing an array of Nakamoto longest chain protocol families.

In detail, our paper is organized as follows.
Section~\ref{sec:model} describes our model of
\emph{compliant} strategies and protocols. A strategy is compliant if a party
that employs it never violates a predicate $\infractionPredicate$, which
captures well-defined types of deviant behavior. Accordingly, a protocol is
compliant if, assuming a starting point where no party deviates, no party will
eventually employ a non-compliant strategy, assuming sequential unilateral
defections. Section~\ref{sec:blockchains} specifies compliance for blockchain
protocols, under an infraction predicate that captures abstaining and producing
conflicting blocks, and two types of utility, absolute rewards and profit.
We then explore different reward schemes and protocol families.
Section~\ref{sec:universal} shows that \emph{resource-proportional} rewards, \ie which depend
only on a party's mining or staking power, result in compliance \wrt
rewards alone (i.e., when costs are negligible), but non-compliance \wrt profit (rewards minus costs). Next, we explore \emph{block-proportional}
rewards, using as baseline the total blocks adopted by an impartial observer of the
system. Section~\ref{subsec:bitcoin} shows that PoW systems are compliant \wrt
rewards alone. Section~\ref{subsec:single-leader-pos} shows
that PoS systems, which enforce that a single party participates at a time, are
compliant, under a synchronous network, but non-compliant under a lossy
network (contrary to PoW). Section~\ref{subsec:multi-leader-pos} shows that PoS systems, which
allow multiple parties to produce blocks for the same time slot, are not
compliant. Notably, our negative results show that a party can gain a
\emph{non-negligible} reward by being non-compliant under a certain network routing assumption, in this way also highlighting the way the network interacts with protocol incentives.
Section~\ref{sec:compliant-non-equilibrium} highlights the distinction between
compliance and Nash equilibria, by showcasing a protocol under two types of
rewards that, in both cases, is compliant \wrt an infraction predicate that captures realistic deviations and non-compliant \wrt another realistic predicate (and hence, not a Nash equilibrium).
Finally, we evaluate
compliance under various externalities, specifically the varying exchange rate of the platform's underlying token, which models
real-world prices, and external rewards, which come as a result of successful
attacks. We show that
applying a penalty would be necessary if a certain deviant behavior is synergistic to mounting the attacks, and we provide estimations for such penalties \wrt the ledger's parameters and the market's expected behavior.

\section{Compliance Model}\label{sec:model}

We assume a distributed protocol $\proto$, which is executed by a set of
parties $\partySet$ over a number of time slots.
Every party $\party \in \partySet$ is activated on each time slot, following a
schedule set by an environment $\env$, which also provides the parties with
inputs.
Each party $\party \in \partySet$ is associated with a number
$\miningpower_{\party} \in [0, 1]$. $\miningpower_{\party}$ identifies
$\party$'s percentage of participation power in the protocol, \eg its votes,
hashing power, staking power, \etc; consequently, $\sum_{\party \in \partySet}
\miningpower_{\party} = 1$.
$\secparam$ denotes $\proto$'s security parameter, $\mathsf{negl}(\cdot)$ denotes that a function is negligible, i.e., asymptotically smaller than the inverse of any polynomial, $[n]$ denotes the set $\{1, \ldots, n\}$, and $E[X]$ denotes the expectation of random variable $X$.

\subsection{Preliminaries}\label{sec:preliminaries}

We assume a peer-to-peer network, \ie parties do not communicate via
point-to-point connections. Instead, they use the following variant of a \emph{diffuse}
functionality (cf.~\cite{EC:GarKiaLeo15}).

\paragraph{Router.}
We consider a special party called \emph{router} $\adversary$. On each time slot,
$\adversary$ retrieves all created messages and decides their order and time of
delivery. In essence, $\adversary$ models the underlying communication network.
In this work, we consider the following three routers, which are of
interest:
\begin{inparaenum}[i)]
    \item \emph{synchronous}: all messages are delivered at the end of the
        round during which they were created;
    \item \emph{lossy}: a message is omitted by $\adversary$, \ie it is never
        delivered to any recipient, with probability
        $\networkLossProb$;\footnote{This router aims to model the
        setting where a network with stochastic delays is used by an
        application, where users reject messages delivered with delay above a
        (protocol-specific) limit. For example, various protocols, like
        Bitcoin~\cite{nakamoto2008bitcoin}, resolve message conflicts based on
        the order of delivery; thus, delaying a message for long enough, such
        that a competing message is delivered beforehand, is equivalent to
        dropping the message altogether.}
    \item \emph{uniform}: the order of message delivery is uniformly
        randomized.
\end{inparaenum}

\paragraph{Diffuse Functionality.}\label{par:diffuse}
The functionality, parameterized by a router $\adversary$, initializes a variable $\mathit{slot}$ to $1$, which is
readable from all parties. In addition, it maintains a string
$\textsc{Receive}_{\party}()$ for each party $\party$. Each party $\party$ is
allowed to fetch the contents of $\textsc{Receive}_{\party}()$ at the beginning
of each time slot. To diffuse a (possibly empty) message $\mesg$, $\party$
sends to the functionality $\mesg$, which records it. On each slot, every party
completes its activity by sending a special $\textsc{Complete}$ message to the
functionality.  When all parties submit $\textsc{Complete}$, the functionality
delivers the messages, which are diffused during this slot, as follows. First,
it sends all messages to $\adversary$. Following,
$\adversary$ responds with a list of tuples $\langle \party, l_\party \rangle$,
where $\party \in \mathbb{P}$ and $l_\party$ is an ordered list of messages.
Subsequently, the functionality includes all messages in $l_\party$, following
its specified order, in the $\textsc{Receive}_{\party}()$ string of $\party$.
Hence, the received messages contain no information on each message's creator.
Finally, the functionality increases the value of $\mathit{slot}$ by $1$.

\paragraph{Approximate Nash Equilibrium.}\label{sec:equilibrium}
An approximate Nash equilibrium is a common tool for expressing a solution to a
non-cooperative game involving $\totalParties$ parties $\party_1,\ldots,\party_n$. Each party $\party_i$
employs a strategy $\strategy_i$. The strategy is a set of rules and actions
the party makes, depending on what has happened up to any point in the game,
\ie it defines the part of the entire distributed protocol $\proto$ performed
by $\party_i$. There exists an ``honest'' strategy, defined by $\proto$, which
parties may employ; for ease of notation,
$\proto$ denotes both the distributed protocol and the honest strategy.
A \emph{strategy profile} is a vector of all players' strategies.
Each party $\party_i$ has a game \emph{utility} $\utility_i$, which is a real
function that takes as input a strategy profile. A strategy profile is an
$\epsilon$-Nash equilibrium when no party can increase its utility more than
$\epsilon$ by \emph{unilaterally} changing its strategy
(Definition~\ref{def:equilibrium}).

\begin{definition}\label{def:equilibrium}
    Let:
    \begin{inparaenum}[i)]
        \item $\epsilon$ be a non-negative real number;
        \item $\strategySet$ be the set of strategies a party may employ;
        \item $\profile^* = (\strategy^*_i, \strategy^*_{-i})$ be a strategy
            profile of $\partySet$, where $\strategy^*_i$ is the strategy
            followed by $\party_i$;
        \item $\strategy^*_{-i}$ denote the $\totalParties - 1$ strategies
            employed by all parties except $\party_i$.
    \end{inparaenum}
    We say that $\profile^*$ is an \emph{$\epsilon$-Nash equilibrium} \wrt a
    utility vector $\bar{\utility} = \langle \utility_1, \ldots,
    \utility_\totalParties \rangle$ if:
    $\forall \party_i \in \partySet \; \forall \strategy_i \in \strategySet \setminus \{ \strategy^*_i \} : \utility_i(\strategy^*_i, \strategy^*_{-i}) \geq \utility_i(\strategy_i, \strategy^*_{-i}) - \epsilon$.
\end{definition}

For simplicity, when all parties have the same utility $U$, we say that the profile $\profile^*$ is an
$\epsilon$-Nash equilibrium \wrt $U$. We also say that $\proto$ is an $\epsilon$-Nash equilibrium \wrt $U$ when the strategy profile $\profile_\proto=\langle\proto,\ldots,\proto\rangle$ where all parties follow the honest strategy is an $\epsilon$-Nash equilibrium \wrt $U$.

\subsection{Basic Notions}\label{subsec:basic}
A protocol's execution $\execution_{\env, \adversary, \profile, \slot}$ until a given time slot
$\slot$ is probabilistic and parameterized by:
\begin{inparaenum}[i)]
    \item the environment $\env$;
    \item a router $\adversary$;
    \item the strategy profile $\profile$ of the participating parties.
\end{inparaenum}
As discussed, $\env$
provides the parties with inputs and schedules their activation. For notation
simplicity, when $\slot$ is omitted, $\execution_{\env, \adversary, \profile}$ refers to the
end of the execution, which occurs after polynomially many time slots.

An \emph{execution trace} $\executionTrace_{\env, \adversary, \profile, \slot}$ until a time slot
$\slot$ is the value that the random variable $\execution_{\env, \adversary, \profile, \slot}$ takes for a fixed
environment $\env$, router $\adversary$, and strategy profile $\profile$, and for fixed random coins of
$\env$, each party $\party \in \partySet$, and every
protocol-specific oracle (see below).
A party $\party$'s view of an
execution trace $\executionTrace_{\env, \adversary, \profile, \slot}^{\party}$ consists of
the messages that $\party$ has sent and received until slot $\slot$. For
notation simplicity, we omit the subscripts $\{ \env, \adversary,
\profile, \slot \}$ from both $\execution$ and $\executionTrace$, unless
required for clarity.\smallskip

The protocol $\proto$ defines two components, which are related to our
analysis: (1) the oracle $\oracle_\proto$, and (2) the ``infraction'' predicate
$\infractionPredicate$. We present them below.

\paragraph{The Oracle $\oracle_\proto$.}
The oracle $\oracle_\proto$ provides the parties with the core functionality
needed to participate in $\proto$. For example, in a Proof-of-Work (PoW) system,
$\oracle_\proto$ is the random or hashing oracle, whereas in an authenticated
Byzantine Agreement protocol, $\oracle_\proto$ is a signing oracle. On each time
slot, a party can perform at most a polynomial number of queries to
$\oracle_\proto$; in the simplest case, each party can submit a single query
per slot. Finally, $\oracle_\proto$ is \emph{stateless}, \ie its random coins
are decided upon the beginning of the execution and its responses do not depend
on the order of the queries.

\paragraph{The Infraction Predicate $\infractionPredicate$.}
The infraction predicate $\infractionPredicate$ abstracts the
deviant behavior that the analysis aims to capture.
Given the execution trace and a party $\party$, $\infractionPredicate$
responds with $1$ only if $\party$ deviates from the protocol in some
well-defined manner. Definition~\ref{def:infraction-predicate} provides the
core generic property of $\infractionPredicate$, \ie that honest
parties never deviate.
With hindsight, our analysis will
focus on infraction predicates that capture either producing conflicting
messages or abstaining.

\begin{definition}[Infraction Predicate Property]\label{def:infraction-predicate}
    The infraction predicate $\infractionPredicate$ has the property that, for
    every execution trace $\executionTrace$ and for every party $\party \in
    \partySet$, if $\party$ employs the (honest) strategy $\proto$ then
    $\infractionPredicate(\executionTrace, \party) = 0$.
\end{definition}

We stress that Definition~\ref{def:infraction-predicate} implies that
$\infractionPredicate$ being $0$ is a necessary but not sufficient
condition for honesty. Specifically, for all honest parties
$\infractionPredicate$ is always $0$, but $\infractionPredicate$ might
also be $0$ for a party that deviates from $\proto$, in a way not
captured by $\infractionPredicate$. In that case, we say that the party employs
an $\compliant$ strategy (Definition~\ref{def:compliant-strategy}). A strategy
profile is $\compliant$ if all its strategies are $\compliant$, so
the ``all honest'' profile $\profile_\proto$, where all parties employ
$\proto$, is $\compliant$.

\begin{definition}[Compliant Strategy]\label{def:compliant-strategy}
    Let $\infractionPredicate$ be an infraction predicate. A strategy $\strategy$ is
    \emph{$\mathcal{X}$-compliant} if and only if $\infractionPredicate(\executionTrace, \party) =
    0$ for every party $\party$ and for every trace $\executionTrace$ where $\party$ employs $\strategy$.
\end{definition}

\paragraph{The observer $\observer$.}
We assume a special party $\observer$, the \emph{(passive)
observer}. This party does not actively participate in the execution, but it
runs $\proto$ and observes the protocol's execution. Notably, $\observer$ is
\emph{always online}, \ie it bootstraps at the beginning of the execution and
is activated on every slot, in order to receive diffused messages. Therefore,
the observer models a user of the system, who frequently uses the system but
does not actively participate in its maintenance. Additionally, at the last
round of the execution, the environment $\env$ activates only $\observer$, in
order to receive the diffused messages of the penultimate round and have a
complete point of view.

\subsection{Compliant Protocols}\label{subsec:compliant}

To define the notion of an $(\epsilon, \infractionPredicate)$-compliant protocol
$\proto$, we require two parameters:
\begin{inparaenum}[(i)]
    \item the associated infraction predicate $\infractionPredicate$ and
    \item a non-negative real number $\epsilon$.
\end{inparaenum}
Following Definition~\ref{def:compliant-strategy}, $\infractionPredicate$
determines the set of compliant strategies that the parties may follow in
$\proto$. Intuitively, $\epsilon$ specifies the sufficient gain threshold after which a party
switches strategies. In particular, $\epsilon$ is
used to define when a strategy profile $\profile'$ is \emph{directly reachable}
from a strategy profile $\profile$, in the sense that $\profile'$ results from the
unilateral deviation of a party $\party_i$ from $\profile$ and,
by this
deviation, the utility of $\party_i$ increases more than $\epsilon$
while $\profile'$ sets a \emph{best response} for $\party_i$.
Generally, $\profile'$ is \emph{reachable} from $\profile$, if $\profile'$ results
from a ``path'' of strategy profiles, starting from $\profile$, which are
sequentially related via direct reachability. Finally, we define the \emph{cone}
of a profile $\profile$ as the set of all strategies that are
reachable from $\profile$, including $\profile$ itself.

Given the above definitions, we say that $\proto$ is
$(\epsilon,\mathcal{X})$-compliant if the cone of the ``all honest'' strategy
profile $\profile_\proto$ contains only profiles that consist of $\compliant$
strategies. Thus, if a protocol is compliant, then the parties may
(unilaterally) deviate from the honest strategy only in a compliant manner, as
dictated by $\infractionPredicate$.
Formally, first we define ``reachability'' between two strategy profiles, as
well as the notion of a ``cone'' of a strategy profile \wrt the reachability
relation. Then, we define a compliant protocol \wrt its associated
infraction predicate.

\begin{definition}\label{def:reach}
    Let:
    \begin{inparaenum}[i)]
        \item $\epsilon$ be a non-negative real number;
        \item $\proto$ be a protocol run by parties $\party_1, \ldots, \party_\totalParties$;
        \item $\bar{\utility}=\langle \utility_1, \ldots, \utility_\totalParties \rangle$ be a utility vector, where $\utility_i$ is the utility of $\party_i$;
        \item $\strategySet$ be the set of all strategies a party may employ.
    \end{inparaenum}
    We provide the following definitions.
    \begin{enumerate}
        \item Let $\profile, \profile' \in \strategySet^\totalParties$ be two strategy profiles where $\profile = \langle \strategy_1, \ldots, \strategy_\totalParties \rangle$ and $\profile'= \langle \strategy'_1, \ldots, \strategy'_\totalParties \rangle$. We say that $\profile'$ is \emph{directly $\epsilon$-reachable from $\profile$ \wrt $\bar{\utility}$}, if there exists $i \in [\totalParties]$ \st (i) $\forall j \in [\totalParties] \setminus \{i\}: \strategy'_j = \strategy_j$, (ii) $\utility_i(\profile') > \utility_i(\profile) + \epsilon$, and (iii) for every strategy profile $\profile''=\langle \strategy''_1, \ldots, \strategy''_\totalParties \rangle$ \st $\forall j \in [\totalParties] \setminus \{i\}: \strategy''_j = \strategy_j$, it holds that $\utility_i(\profile'') \leq \utility_i(\profile')$. (i.e., $\profile'$ sets a best response for $\party_i$)
        \item Let $\profile, \profile' \in \strategySet^\totalParties$ be two distinct strategy profiles. We say that $\profile'$ is \emph{$\epsilon$-reachable from $\profile$ \wrt $\bar{\utility}$}, if there exist profiles $\profile_1, \ldots, \profile_k$ such that (i) $\profile_1 = \profile$, (ii) $\profile_k = \profile'$, and (iii) $\forall j \in [2, k]$ it holds that $\profile_j$ is directly $\epsilon$-reachable from $\profile_{j-1}$ \wrt $\bar{\utility}$.
        \item For every strategy profile $\profile \in \strategySet^\totalParties$ we define the \emph{$(\epsilon, \bar{\utility})$-cone of $\profile$} as the set:
        $\mathsf{Cone}_{\epsilon, \bar{\utility}}(\profile) := \{\profile' \in \strategySet^\totalParties\;|\;(\profile' = \profile) \lor (\profile'\mbox{ is $\epsilon$-reachable from }\profile\mbox{ \wrt }\bar{\utility})\}$.
    \end{enumerate}
\end{definition}

\begin{definition}\label{def:compliant}
    Let:
    \begin{inparaenum}[i)]
        \item $\epsilon$ be a non-negative real number;
        \item $\proto$ be a protocol run by the parties $\party_1, \ldots, \party_\totalParties$;
        \item $\infractionPredicate$ be an infraction predicate;
        \item $\bar{\utility}=\langle \utility_1, \ldots, \utility_\totalParties \rangle$ be a utility vector, where $\utility_i$ is the utility of party $\party_i$;
        \item $\strategySet$ be the set of all strategies a party may employ;
        \item $\strategySet_{\infractionPredicate}$ be the set of $\compliant$ strategies.
    \end{inparaenum}
   \par A strategy profile $\profile \in \strategySet^\totalParties$ is $\compliant$ if $\profile \in (\strategySet_{\infractionPredicate})^\totalParties$.
   \par The \emph{$(\epsilon,\bar{\utility})$-cone of $\proto$}, denoted by $\mathsf{Cone}_{\epsilon,\bar{U}}(\proto)$, is the set $\mathsf{Cone}_{\epsilon, \bar{\utility}}(\profile_\proto)$, \ie the set of all strategies that are $\epsilon$-reachable from the ``all honest'' strategy profile $\profile_\proto = \langle \proto, \ldots, \proto \rangle$ \wrt $\bar{\utility}$, including $\profile_\proto$.
   \par $\proto$ is \emph{$(\epsilon,\infractionPredicate)$-compliant \wrt $\bar{\utility}$} if $\mathsf{Cone}_{\epsilon,\bar{U}}(\proto) \subseteq (\strategySet_{\infractionPredicate})^\totalParties$, \ie all strategy profiles in the $(\epsilon,\bar{\utility})$-cone of $\proto$ are $\compliant$.
\end{definition}

\subsection{Compliance and Approximate Nash Equilibria}\label{subsec:char_app_Nash}

In this subsection, we show that a protocol is an $\epsilon$-Nash equilibrium \wrt some utility if and only if it is $(\epsilon,\infractionPredicate)$-compliant \wrt the same utility, for any associated infraction predicate $\infractionPredicate$. We begin by proving a useful lemma stating that a protocol is an approximate Nash equilibrium if and only if the cone of the protocol includes only the all-honest strategy profile.

\begin{lemma}\label{lem:app_Nash_singleton}
    Let:
    \begin{inparaenum}[i)]
        \item $\epsilon$ be a non-negative real number;
        \item $\proto$ be a protocol run by the parties $\party_1, \ldots, \party_\totalParties$;
        \item $\bar{\utility}=\langle \utility_1, \ldots, \utility_\totalParties \rangle$ be a utility vector, with $\utility_i$ the utility of $\party_i$.
    \end{inparaenum}
    Then, $\proto$ is an $\epsilon$-Nash equilibrium \wrt $\bar{\utility}$ (\ie $\profile_\proto=\langle \proto, \ldots, \proto \rangle$ is an $\epsilon$-Nash equilibrium \wrt $\bar{\utility}$) if and only if the $(\epsilon,\bar{\utility})$-cone of $\proto$, $\mathsf{Cone}_{\epsilon,\bar{U}}(\proto)$, is the singleton $\{\profile_\proto\}$.
\end{lemma}

\begin{proof}$(\Rightarrow)$: Assume that $\profile_\proto$ is an $\epsilon$-Nash equilibrium \wrt $\bar{\utility}$ and let $\profile=\langle\strategy_1, \ldots, \strategy_\totalParties\rangle\neq\profile_\proto$ be a strategy profile s.t. there exists $i \in [\totalParties]$ s.t. $\forall j \in [\totalParties] \setminus\{i\}: \strategy_j = \proto$. We will show that $\profile$ is not directly $\epsilon$-reachable from $\profile_\proto$ \wrt $\bar{\utility}$.
Since $\profile_\proto$ is an $\epsilon$-Nash equilibrium \wrt $\bar{\utility}$, it holds that $\utility_i(\profile) \leq \utility_i(\profile_\proto) + \epsilon$. Therefore, Definition~\ref{def:reach} is not satisfied and $\profile$ is not directly $\epsilon$-reachable from $\profile_\proto$ \wrt $\bar{\utility}$.
Since no strategy profiles are directly $\epsilon$-reachable from $\profile_\proto$ \wrt $\bar{\utility}$, it is straightforward that there are no $\epsilon$-reachable strategy profiles from $\profile_\proto$ \wrt $\bar{\utility}$. The latter implies that the $(\epsilon, \bar{\utility})$-cone of $\proto$ contains only $\profile_\proto$, \ie $\mathsf{Cone}_{\epsilon,\bar{\utility}}(\proto) = \{ \profile_\proto\}$. \\[2pt]

$(\Leftarrow)$: Assume that $\profile_\proto$ is not an $\epsilon$-Nash equilibrium \wrt $\bar{\utility}$. This means that there exists a strategy profile $\profile^*=\langle\strategy^*_1, \ldots, \strategy^*_\totalParties\rangle$ s.t. there exists $i \in [\totalParties]$ s.t. $\forall j \in [\totalParties] \setminus\{i\}: \strategy^*_j = \proto$ and it holds that $\utility_i(\profile^*) > \utility_i(\profile_\proto) + \epsilon$. Among all unilateral deviations of $\party_i$, consider the strategy profile $\profile^{**}$ that sets the best response for $\party_i$ (that may not necessarily be $\profile^*)$. We directly get that $\utility_i(\profile^{**})\geq\utility_i(\profile^*) > \utility_i(\profile_\proto) + \epsilon$.
Thus, by Definition~\ref{def:reach}, $\profile^{**}$ is (directly) $\epsilon$-reachable from $\profile_\proto$ \wrt $\bar{\utility}$, i.e. $\profile^{**}\in\mathsf{Cone}_{\epsilon,\bar{\utility}}(\proto)$, which implies that $\{ \profile_\proto\}\subsetneq\mathsf{Cone}_{\epsilon,\bar{\utility}}(\proto)$.

\end{proof}

The statement in Lemma~\ref{lem:app_Nash_singleton} resembles the well-known statement that a pure Nash equilibrium is a sink equilibrium that contains a single strategy profile~\cite{sink}. Nonetheless, there are differences between the notions of a sink equilibrium and a cone. Recall that a sink equilibrium is a strongly connected component of the strategy profile graph that has no outgoing edges. On the other hand, according to Definition~\ref{def:reach}, the subgraph induced by the nodes of a cone of a strategy profile $\profile$ may not even be strongly connected (e.g., the cone could be a subtree rooted at $\profile$).

By applying Lemma~\ref{lem:app_Nash_singleton}, we prove the main result of this subsection.

\begin{theorem}\label{thm:eq_comp}
    Let:
    \begin{inparaenum}[i)]
        \item $\epsilon$ be a non-negative real number;
        \item $\proto$ be a protocol run by the parties $\party_1, \ldots, \party_\totalParties$;
        \item $\bar{\utility}=\langle \utility_1, \ldots, \utility_\totalParties \rangle$ be a utility vector, with $\utility_i$ the utility of $\party_i$.
    \end{inparaenum}
    Then, $\proto$ is an $\epsilon$-Nash equilibrium \wrt $\bar{\utility}$ if and only if $\proto$ is $(\epsilon, \infractionPredicate)$-compliant \wrt $\bar{\utility}$ for any associated infraction predicate~$\infractionPredicate$.
\end{theorem}
\begin{proof}$(\Rightarrow)$: Assume that $\profile_\proto$ is an $\epsilon$-Nash equilibrium \wrt $\bar{\utility}$. By Lemma~\ref{lem:app_Nash_singleton}, we have that $\mathsf{Cone}_{\epsilon,\bar{\utility}}(\proto) = \{ \profile_\proto\}$. By Definitions~\ref{def:infraction-predicate} and~\ref{def:compliant-strategy}, we have that for every infraction predicate $\infractionPredicate$ $\profile_\proto$ is $\compliant$, so we deduce that $\mathsf{Cone}_{\epsilon,\bar{\utility}}(\proto) = \{ \profile_\proto\}\subseteq (\strategySet_{\infractionPredicate})^\totalParties$, i.e., the protocol $\proto$ is $(\epsilon,\infractionPredicate)$-compliant w.r.t. $\bar{\utility}$.\\[2pt]
$(\Leftarrow)$: Assume that $\proto$ is $(\epsilon, \infractionPredicate)$-compliant \wrt $\bar{\utility}$ for every associated infraction predicate $\infractionPredicate$. We define the ``all honest'' infraction predicate, $\infractionPredicate_\mathsf{hon}$ as follows: for every party $\party$ and every trace $\executionTrace$,
\begin{equation*}
\infractionPredicate_\mathsf{hon}(\executionTrace,\party):=
 \left\{\begin{array}{ll}
  0,&\mbox{if $\party$ follows $\proto$ in }\executionTrace\\
  1,&\mbox{otherwise}
\end{array}
\right.\;.
\end{equation*}

Clearly, $\infractionPredicate_\mathsf{hon}$ satisfies Definition~\ref{def:infraction-predicate}, i.e., it is indeed an infraction predicate. By the definition of $\infractionPredicate_\mathsf{hon}$, for every strategy $\strategy\neq\proto$, it holds that $\infractionPredicate_\mathsf{hon}(\executionTrace,\party)=1$ for every party $\party$ and every trace $\executionTrace$ where $\party$ follows $\strategy$ in $\executionTrace$. The latter implies that for every strategy $\strategy\neq\proto$ it holds that $S\notin\strategySet_{\infractionPredicate_\mathsf{hon}}$ (i.e., $S$ is not $\infractionPredicate_\mathsf{hon}$-compliant), which in turn implies that for every strategy profile $\profile\neq \profile_\proto=\langle\proto,\ldots,\proto\rangle$ it holds that $\profile\notin(\strategySet_{\infractionPredicate_\mathsf{hon}})^\totalParties$. On the other hand, it is straightforward that $\profile_\proto\in(\strategySet_{\infractionPredicate_\mathsf{hon}})^\totalParties$.

By the above, we have that $(\strategySet_{\infractionPredicate_\mathsf{hon}})^\totalParties=\{\profile_\proto\}$. By assumption, we have that $\proto$ is $(\epsilon, \infractionPredicate_\mathsf{hon})$-compliant \wrt $\bar{\utility}$, i.e. $\mathsf{Cone}_{\epsilon,\bar{\utility}}(\proto)\subseteq (\strategySet_{\infractionPredicate_\mathsf{hon}})^\totalParties=\{\profile_\proto\}$. Besides, by Definition~\ref{def:reach}, we know that $\{\profile_\proto\}\subseteq\mathsf{Cone}_{\epsilon,\bar{\utility}}(\proto)$, thus $\mathsf{Cone}_{\epsilon,\bar{\utility}}(\proto)=\{\profile_\proto\}$. By Lemma~\ref{lem:app_Nash_singleton}, we conclude that $\proto$ is an $\epsilon$-Nash equilibrium \wrt $\bar{\utility}$.

\end{proof}

According to the equivalence proven in Theorem~\ref{thm:eq_comp}, the property that ``a protocol is an approximate Nash equilibrium'' can be interpreted as a composition of all possible
statements that ``it is not in any party's interest to be non-compliant'', however compliance is specified by the associated infraction predicate.
\paragraph{Remark.}
It is easy to see that $\infractionPredicate$-compliance is a \emph{strict relaxation} of the approximate Nash equilibrium notion. For example, consider a protocol $\proto^*$ that is not an $\epsilon^*$-Nash equilibrium \wrt some utility $\bar{\utility}^*$ (cf. Theorem~\ref{thm:universal-profit} for such a counterexample). Now set $\infractionPredicate^*$ to be the predicate that always returns $0$. Clearly, by Definitions~\ref{def:compliant-strategy} and~\ref{def:compliant}, the protocol $\proto^*$ is $(\epsilon^*, \infractionPredicate^*)$-compliant \wrt $\bar{\utility^*}$.

\section{Blockchain Protocols}\label{sec:blockchains}

In this work, we focus on blockchain-based distributed ledger protocols.  In
the general case, a ledger defines a global state, which is distributed across
multiple parties and is maintained via a consensus protocol. The distributed
ledger protocol defines the validity rules which allow a party to extract the
final ledger from its view. A
blockchain is a distributed database, where each message $\mesg$ is a block
$\block$ of transactions and each transaction updates the system's global
state. Therefore, at any point of the execution, a party $\party$ holds some view of
the global state, which comprises of the blocks that $\party$ has
adopted. We note that, if at least one valid block is diffused (\wrt the
validity rules of the protocol), then every honest party can
extract a final ledger from its execution view.

\subsection{The Setting}\label{subsec:setting}

Every blockchain protocol $\proto$ defines a \emph{message validity} predicate
$\messageValidityPredicate$.  Party $\party$ accepts block $\block$,
received during a time slot $\slot$, if
$\messageValidityPredicate(\executionTrace^{\party}_\slot, \block) = 1$. For example, in Proof-of-WorK
(PoW) systems like Bitcoin, a block is valid if its
hash is below a certain threshold; in Proof-of-Stake (PoS) protocols like
Ouroboros~\cite{C:KRDO17}, a block is valid if it was created by a
specific party, given a known leader schedule. In all cases,
$\block$ is valid if its creator submits at least one query for
$\block$ to $\oracle_\proto$.

Each block $\block$ is associated with the following metadata:
\begin{inparaenum}[i)]
    \item an index $\mathit{index}(\block)$;
    \item the party $\mathit{creator}(\block)$ that created $\block$;
    \item a set $\mathit{ancestors}(\block) \subseteq
    \executionTrace^{\mathit{creator}(\block)}$, \ie blocks in the view of
    $\mathit{creator}(\block)$ (at the time of $\block$'s creation)
    referenced by $\block$.
\end{inparaenum}
Message references are implemented as hash pointers, given a hash function
$\hash$ employed by the protocol. Specifically, each block $\block$ contains the
hash of all blocks in the referenced blocks $\mathit{ancestors}(\block)$.
Blockchain systems
are typically bootstrapped via a global common reference string, \ie a
``genesis'' block $\genesis$.
Therefore, the blocks form a hash tree, stemming from $\genesis$ and
$\mathit{index}(\block)$ is the height of $\block$ in the hash tree. If
$\block$ references multiple messages, \ie belongs to multiple tree branches,
$\mathit{index}(\block)$ is the height of the longest one.

The protocol also defines the \emph{message equivalency operator}, $\equiv$. Specifically, two messages are equivalent if
their hashes match, \ie $\mesg_1 \equiv \mesg_2 \Leftrightarrow \hash(\mesg_1) =
\hash(\mesg_2)$. At a high level, two equivalent messages are
interchangeable by the protocol.

\paragraph{Infraction Predicate.}\label{sec:blockchain-infraction-predicate}
In our analysis of blockchain systems, we will consider two types of deviant
behavior (Definition~\ref{def:blockchain-infraction}):
\begin{inparaenum}[i)]
    \item creating conflicting valid messages of same origin, and
    \item abstaining.
\end{inparaenum}
We choose these predicates because they may lead to non-compliance in
interesting use cases. The former refers to the widely discussed topic in
blockchain systems of one participant extending two conflicting transaction
histories. The latter deals with the issue of participants who intermittently
engage in the system's maintenance, thus potentially hurting the safety of the
deployed system; in particular, the more users participate in maintenance, the
higher the level of resources that an adversary needs to reach to break a
system's security. Other infraction predicates are of course also possible to
define --- see Section~\ref{sec:conclusion}.

\begin{definition}[Blockchain Infraction Predicate]\label{def:blockchain-infraction}
    Given a party $\party$ and an execution trace $\executionTrace$,
    we define the following infraction predicates:
    \begin{enumerate}
        \item
            \emph{conflicting predicate}: $\infractionPredicate_\mathrm{conf}(\executionTrace, \party) = 1$ if
            there exist blocks $\block, \block'
            \in \executionTrace$ such that $\mathit{creator}(\block) = \mathit{creator}(\block') = \party\land\messageValidityPredicate(\executionTrace^{\party}, \block)=\messageValidityPredicate(\executionTrace^{\party}, \block')=1\land
            \mathit{index}(\block) = \mathit{index}(\block') \land \block \not \equiv
            \block';$
        \item
            \emph{abstaining predicate}: $\infractionPredicate_\mathrm{abs}(\executionTrace, \party) = 1$ if
            there exists a time slot $\slot$ such that $\party$ makes \emph{no} queries to oracle $\oracle_\proto$
            during $\slot$;
        \item
            \emph{blockchain predicate}: $\infractionPredicate_\mathrm{bc}(\executionTrace, \party) = 1$ if
            $(\infractionPredicate_\mathrm{conf}(\executionTrace, \party) = 1) \lor (\infractionPredicate_\mathrm{abs}(\executionTrace,\party) = 1)$.
    \end{enumerate}
\end{definition}

We note that preventing conflicting messages is not the same
as resilience against Sybil attacks~\cite{douceur2002sybil}. The latter
restricts an attacker from creating multiple identities.
Instead, our infraction predicate ensures that a user does not increase their
utility by creating conflicting messages with one of its identities.
Thus, a system may be compliant but not Sybil resilient, \eg if a party
participates via multiple identities without increasing its utility via
conflicting messages.

Finally, at the end of the execution, the observer $\observer$ outputs a chain
$\chain_{\observer, \executionTrace}$. Typically, this is the longest valid chain, \ie the longest
branch of the tree that stems from genesis $\genesis$.\footnote{We assume that the longest chain (in
blocks) contains the most hashing power, which is the metric used in PoW
systems.} In case multiple longest chains exist, a choice is made
either at random or following a chronological ordering of messages. The number
of messages in $\chain_{\observer, \executionTrace}$ that are created by a party $\party$ is
denoted by $\msgOutputSet_{\party, \executionTrace}$.

\subsection{Utility: Rewards and Costs}\label{sec:blockchain-utility}

For each execution, the blockchain protocol defines a number of total rewards,
which are distributed among the participating parties. For each party $\party$,
these rewards are expressed via the \emph{reward random variable}
$\reward_{\party, \execution_{\env, \adversary, \profile}}$.  For a specific trace $\executionTrace_{\env, \adversary, \profile}$, the
random variable takes a non-negative real value, denoted by $\reward_{\party,
\executionTrace_{\env, \adversary, \profile}}$. Intuitively, $\reward_{\party, \executionTrace_{\env, \adversary, \profile}}$ describes
the rewards that $\party$ receives from the protocol from the point of view of
the observer $\observer$, \ie \wrt the blocks output by $\observer$ at the end
of the execution.

Our analysis is restricted to systems where rewards are distributed to parties if
and only if the genesis block is extended by at least one block during the
execution, in which case at least one party receives a non-negative amount of
rewards (Assumption~\ref{ass:reward-distribution}).

\begin{assumption}\label{ass:reward-distribution}
    Let $\executionTrace$ be an execution trace. If no block is produced during $\executionTrace$, then it
    holds that $\forall \party \in \partySet: \reward_{\party, \executionTrace} =
    0$.
    If at least one block is produced during $\executionTrace$, then it
    holds that $\exists \party \in \partySet: \reward_{\party, \executionTrace} \neq
    0$.
\end{assumption}

In addition to rewards, a party's utility is affected by cost. Specifically,
the \emph{cost random variable} $\cost_{\party, \execution_{\env, \adversary, \profile}}$ expresses the
operational cost of $\party$ during an execution $\execution_{\env, \adversary, \profile}$. For a fixed
trace $\executionTrace_{\env, \adversary, \profile}$, $\cost_{\party, \executionTrace_{\env, \adversary, \profile}}$ is a non-negative
real value. Our analysis is restricted to cost schemes which are \emph{linearly
monotonically increasing} in the number of queries that a party makes to the
oracle $\oracle_\proto$, with no queries incurring zero cost
(Assumption~\ref{ass:zero-cost}). Intuitively, this assumption considers the
electricity cost of participation, while the cost of equipment and other
operations, such as parsing or publishing messages, is zero.

\begin{assumption}\label{ass:zero-cost}
    For every execution trace $\executionTrace$, a party $\party$'s cost is
    $\cost_{\party, \executionTrace} = 0$ if and only if it performs no queries
    to $\oracle_\proto$ in every time slot. Else, if during $\executionTrace$ a party $\party$ performs
    $t$ queries, then its cost is $\cost_{\party, \executionTrace} = t \cdot \queryCost$, for some fixed parameter $\queryCost$.
\end{assumption}

We define two types of utility. First is \emph{Reward}, \ie the expected
rewards that a party receives when the cost is $0$. Second is \emph{Profit},
\ie rewards minus participation cost.

\begin{definition}\label{def:utility}
    Let $\profile$ be a strategy profile and $\execution_{\env, \adversary, \profile}$ be an
    execution during which parties follow $\profile$. We define two types of
    blockchain utility $\utility_{\party}$ of a party $\party$ for $\profile$:
    \begin{enumerate}
        \item \emph{Reward}: $\utility_{\party}(\profile) = E[\reward_{\party, \execution_{\env, \adversary, \profile}}]$
        \item \emph{Profit}: $\utility_{\party}(\profile) = E[\reward_{\party, \execution_{\env, \adversary, \profile}}] - E[\cost_{\party, \execution_{\env, \adversary, \profile}}]$
    \end{enumerate}
\end{definition}

For the computation of $\utility_{\party}$, the environment $\env$ and the
router $\adversary$ are fixed. Therefore, the expectation of the random
variables $\reward_{\party, \execution_{\env, \adversary, \profile}}$ and
$\cost_{\party, \execution_{\env, \adversary, \profile}}$ is computed over the
random coins of $\env$, $\adversary$, $\oracle_\proto$, and every party $\party
\in \partySet$. Intuitively, a party's utility depends on both their strategy
choice and the underlying network (expressed via the router). As such,
different routers may yield different optimal strategies for parties to employ
and possibly different equilibria ceteris paribus.
Following, we evaluate the compliance of various Proof-of-Work (PoW) and Proof-of-Stake (PoS)
blockchain protocols \wrt two types of rewards, \emph{resource-proportional} and
\emph{block-proportional}.

\section{Resource-Proportional Rewards}\label{sec:universal}

As described in Section~\ref{sec:model}, a party $\party$ controls a percentage
$\miningpower_{\party}$ of the system's participating power. Although
this is set at the beginning of the execution, it is not always
public. For instance, $\party$ could obscure its amount of hashing power by
refraining from performing some queries. In some cases, each
party's power is published on the ledger and, for all executions,
can be extracted from the observer's chain. This is the case in
non-anonymous PoS ledgers, where each party's power, denoted by
its assets, is logged in real time on the ledger.

These systems, where power distribution is public, can employ a special
type of rewards, \emph{resource-proportional rewards}.
Specifically, the system
defines a fixed, total number of rewards $\totalReward > 0$. At the end
of an execution, if at least one block is created, each party $\party$ receives
a percentage $\universalRewardFunc(\miningpower_{\party})$ of
$\totalReward$, where $\universalRewardFunc(\cdot): [0, 1] \rightarrow [0, 1]$;
in the real world, $\universalRewardFunc$ is usually the identity
function. If no blocks are created during the execution, then every party gets
$0$ rewards.

Intuitively, resource-proportional rewards (Definition~\ref{def:universal-rewards}) compensate users for investing in the
system. Unless no block is created (which typically happens with negligible
probability when the parties follow the protocol), the
reward level depends \emph{solely} on a party's power, instead of the messages
diffused in the execution.

\begin{definition}[Resource-proportional Rewards]\label{def:universal-rewards}
    For a total number of rewards $\totalReward \in \mathbb{R}_{> 0}$ and $\universalRewardFunc: [0, 1] \rightarrow [0, 1]$ such that $\sum_{\party \in \partySet} \universalRewardFunc(\miningpower_{\party}) = 1$, a
    \emph{resource-proportional reward} random variable $\reward_{\party, \execution}$ satisfies
    the following:
    \begin{equation*}\forall \executionTrace \; \forall \party \in \partySet:
        \reward_{\party, \executionTrace} =
 \left\{\begin{array}{ll}
  \universalRewardFunc(\miningpower_{\party}) \cdot \totalReward,&\mbox{if there is at least one valid block in }\executionTrace\\
  0,&\mbox{otherwise}
\end{array}
\right.\end{equation*}

\end{definition}

As shown in Theorem~\ref{thm:universal-reward}, blockchains with resource-proportional
rewards are $\epsilon$-Nash equilibria w.r.t. utility \emph{Reward} (Definition~\ref{def:utility}), with $\epsilon$ typically small. By Theorem~\ref{thm:eq_comp}, the latter implies that such protocols are also $(\epsilon,\infractionPredicate)$-compliant w.r.t. the same utility, where $\infractionPredicate$ is an arbitrary associated infraction predicate. Intuitively, a
party is rewarded the same amount regardless of their protocol-related actions,
so nobody can increase their rewards by deviating from the honest
strategy.

\begin{theorem}\label{thm:universal-reward}
    Let:
    \begin{inparaenum}[i)]
        \item $\proto$ be a blockchain protocol run by the parties $\party_1, \ldots, \party_\totalParties$;
        \item $\adversary$ be a synchronous router (cf. Section~\ref{sec:preliminaries});
        \item $\bar{\utility}=\langle \utility_1, \ldots, \utility_\totalParties \rangle$ be a utility vector, where $\utility_i$ is the utility \emph{Reward} of party $\party_i$;
        \item $\totalReward$ be the total rewards distributed by the protocol;
        \item $\universalRewardFunc: [0, 1] \rightarrow [0, 1]$ be a resource-proportional reward function;
        \item $\alpha$ be the probability that no blocks are produced when all parties follow the honest strategy.
    \end{inparaenum}
    Then, $\proto$ is an $\epsilon$-Nash equilibrium \wrt $\bar{\utility}$, for $\epsilon := \alpha \cdot \underset{j \in [\totalParties]}{\mathsf{max}}\{\universalRewardFunc(\miningpower_{\party_j}) \cdot \totalReward \}$.
\end{theorem}

\begin{proof}
    By Definition~\ref{def:universal-rewards} and the definition of $\alpha$, for the ``all honest'' strategy profile $\profile_\proto := \langle \proto, \ldots, \proto \rangle$, we have that $\Pr[\reward_{\party_i,\execution_{\profile_\proto}} = \universalRewardFunc(\miningpower_{\party_i}) \cdot \totalReward] = 1 - \alpha$ and $\Pr[\reward_{\party_i, \execution_{\profile_\proto}} = 0] = \alpha$, for every $i \in [\totalParties]$. Therefore, for every $i \in [\totalParties]$,
    $\utility_i(\profile_\proto) = E\big[\reward_{\party_i,\execution_{\profile_{\party}}}\big] = (1 - \alpha) \cdot \universalRewardFunc(\miningpower_{\party_i}) \cdot \totalReward$.

    Assume that for some $i \in [\totalParties]$, $\party_i$ unilaterally deviates
    from $\proto$ by employing a different strategy $\strategy_i$. In this case, we
    consider the strategy profile $\profile = \langle \strategy_1, \ldots, \strategy_n \rangle$
    where $\strategy_j = \proto$ for $j \in [\totalParties] \setminus \{i\}$. Since $\utility_i$ is the utility
    \emph{Reward} under fair rewards with $\totalReward, \universalRewardFunc(\cdot)$, we have that for all random coins of the execution $\execution_{\profile}$, the value of the reward random variable $\reward_{\party_i, \execution_{\profile}}$ is no more than $\universalRewardFunc(\miningpower_{\party_i}) \cdot \totalReward$.
    Consequently, $\utility_i(\profile) \leq \universalRewardFunc(\miningpower_{\party_i}) \cdot \totalReward$, and so we have that
    \[\utility_i(\profile) \leq \utility_i(\profile_\proto) + \alpha \cdot \universalRewardFunc(\miningpower_{\party_i}) \cdot \totalReward \leq \utility_i(\profile_\proto) + \alpha \cdot \underset{j \in [\totalParties]}{\mathsf{max}}\{\universalRewardFunc(\miningpower_{\party_j}) \cdot \totalReward \}\;.\]
    If $\epsilon := \alpha \cdot \underset{j \in [\totalParties]}{\mathsf{max}}\{\universalRewardFunc(\miningpower_{\party_j}) \cdot \totalReward \}$ and since $i$ and $\strategy_i$ are arbitrary, we conclude that $\proto$ is an $\epsilon$-Nash equilibrium
    \wrt $\bar{\utility}$.
\end{proof}

Theorem~\ref{thm:universal-reward} is consistent with the incentives' analysis
of~\cite{C:KRDO17} under resource-proportional rewards. However, when
introducing operational costs to analyze profit, a problem arises: a user can
simply abstain and be rewarded nonetheless. Such behavior results in a
``free-rider problem''~\cite{baumol2004welfare}, where a user reaps some
benefits while not under-paying them or not paying at all.
Theorem~\ref{thm:universal-profit} formalizes this argument and
shows that a blockchain protocol, associated with the abstaining infraction predicate $\infractionPredicate_\mathrm{abs}$ (cf. Definition~\ref{def:blockchain-infraction}), under resource-proportional rewards is \emph{not}
$(\epsilon, \infractionPredicate_\mathrm{abs})$-compliant \wrt utility \emph{Profit}, for reasonable values of $\epsilon$.

\begin{theorem}\label{thm:universal-profit}
Let:
    \begin{inparaenum}[i)]
        \item $\proto$ be a blockchain protocol run by the parties $\party_1, \ldots, \party_\totalParties$;
        \item $\adversary$ be a synchronous router (cf. Section~\ref{sec:preliminaries});
        \item $\bar{\utility}=\langle \utility_1, \ldots, \utility_\totalParties \rangle$ be a utility vector, where $\utility_i$ is the utility \emph{Profit} of party $\party_i$;
        \item $\totalReward$ be the total rewards distributed by the protocol;
        \item $\universalRewardFunc: [0, 1] \rightarrow [0, 1]$ be a resource-proportional reward function;
        \item $\alpha$ be the probability that no blocks are produced when all parties follow the honest strategy.
    \end{inparaenum}\\
    For $i \in [\totalParties]$, also let the following:
    \begin{inparaenum}[i)]
        \item $q$ be the maximum number of queries that a party can make to the oracle $\oracle_\proto$ in each time slot.
        \item $\cost$ be the cost of a single query to $\oracle_\proto$;
        \item $\cost_i$ be the expected cost of $\party_i$ when $\party_i$ employs $\proto$;
        \item $\beta_i$ be the probability that no blocks are produced when $\party_i$ abstains throughout the entire execution and all the other parties follow $\proto$.
    \end{inparaenum}\\
     Assume that for every $i\in[\totalParties]$, it holds that $\cost>\beta_{i}\cdot \universalRewardFunc(\miningpower_{\party_{i}}) \cdot \totalReward\cdot q$. Then, for every $\epsilon\geq0$ s.t.
     $\epsilon < \underset{i \in [\totalParties]}{\mathsf{max}}\{\cost_i - (\beta_i - \alpha) \cdot \universalRewardFunc(\miningpower_{\party_i}) \cdot \totalReward \}$,
     the protocol $\proto$ is \emph{not} $(\epsilon, \infractionPredicate_\mathrm{abs})$-compliant \wrt $\bar{\utility}$.
\end{theorem}

\begin{proof}
    By Definition~\ref{def:universal-rewards} and the definition of $\alpha$, for the ``all honest'' strategy profile $\profile_{\proto} := \langle \proto, \ldots, \proto \rangle$, we have that $\Pr[\reward_{\party_i, \execution_{\profile_{\proto}}} = \universalRewardFunc(\miningpower_{\party_i}) \cdot \totalReward] = 1 - \alpha$ and $\Pr[\reward_{\party_i, \execution_{\profile_{\proto}}} = 0] = \alpha$, for every $i \in [\totalParties]$.
    Since $\proto$ is an $\infractionPredicate_\mathrm{abs}$-compliant strategy, if $\party_i$ follows $\proto$ then it does not abstain, \ie it makes queries to $\oracle_\proto$. Therefore, by Assumption~\ref{ass:zero-cost}, the expected cost of $\party_i$, $C_i$, is greater than $0$ and for $\profile_{\proto}$, the utility \emph{Profit} $\utility_i(\profile_{\proto})$ is:
    $\utility_i(\profile_{\proto}) = E\big[\reward_{\party_i, \execution_{\profile_{\proto}}}\big] - E\big[\cost_{\party_i, \execution_{\profile_{\proto}}}\big] = (1 - \alpha) \cdot \universalRewardFunc(\miningpower_{\party_i}) \cdot \totalReward - \cost_i$.

    Now assume that $\party_i$ unilaterally deviates by following the ``always abstain" strategy, $\strategy_\mathsf{abs}$, which is of course not $\infractionPredicate_\mathrm{abs}$-compliant. Then, $\party_i$ makes no queries to $\oracle_\proto$ and, by Assumption~\ref{ass:zero-cost}, its cost is $0$. Let $\profile_i$ be the strategy profile where $\party_i$ follows $\strategy_\mathsf{abs}$ and every party $\party \neq \party_i$ follows $\proto$. By the definition of $\beta_i$, we have that $\Pr[\reward_{\party_i, \execution_{\profile_i}} = \universalRewardFunc(\miningpower_{\party_i}) \cdot \totalReward] = 1 - \beta_i$ and $\Pr[\reward_{\party_i, \execution_{\profile_i}} = 0] = \beta_i$.

    By the definition of $\utility_i$, it holds that:
    $\utility_i(\profile_i) = E\big[\reward_{\party_i, \execution_{\profile_i}}\big] - E\big[\cost_{\party_i, \execution_{\profile_i}}\big] = (1 - \beta_i) \cdot \universalRewardFunc(\miningpower_{\party_i}) \cdot \totalReward - 0$.
    So, for $\epsilon_i < \cost_i - (\beta_i - \alpha) \cdot \universalRewardFunc(\miningpower_{\party_i}) \cdot \totalReward$, we have that:
    \begin{equation*}
        \begin{split}
            \utility_i(\profile_i) & = (1 - \beta_i) \cdot \universalRewardFunc(\miningpower_{\party_i}) \cdot \totalReward = \\
            & = (1 - \alpha) \cdot \universalRewardFunc(\miningpower_{\party_i}) \cdot \totalReward + (\alpha - \beta_i) \cdot \universalRewardFunc(\miningpower_{\party_i}) \cdot \totalReward \geq \\
            & \geq \utility_i(\profile_{\proto}) + \cost_i - (\beta_i - \alpha) \cdot \universalRewardFunc(\miningpower_{\party_i}) \cdot \totalReward > \\
            & > \utility_i(\profile_{\proto}) + \epsilon_i
        \end{split}
    \end{equation*}

 Let $i^*\in[n]$ be such that $\cost_{i^*} - (\beta_{i^*} - \alpha) \cdot \universalRewardFunc(\miningpower_{\party_{i^*}}) \cdot \totalReward = \underset{i \in [\totalParties]}{\mathsf{max}}\{ \cost_i - (\beta_i - \alpha) \cdot \universalRewardFunc(\miningpower_{\party_i}) \cdot \totalReward \}$ and assume that $0<\cost_{i^*} - (\beta_{i^*} - \alpha) \cdot \universalRewardFunc(\miningpower_{\party_{i^*}}) \cdot \totalReward$. By the above, we have that

\begin{equation}\label{eq:abstain_wins}
\utility_{i^*}(\profile_{i^*})>\utility_{i^*}(\profile_{\proto}) + \epsilon, \mbox{ for every }0\leq\epsilon < \cost_{i^*} - (\beta_{i^*} - \alpha) \cdot \universalRewardFunc(\miningpower_{\party_{i^*}}) \cdot \totalReward \;.
\end{equation}

The following claim will imply that the best response for $\party_{i^*}$ must be non $\infractionPredicate_\mathrm{abs}$-compliant.

\begin{claim}\label{claim:abstain_wins}
Let $\profile$ be a $\infractionPredicate_\mathrm{abs}$-compliant profile strategy that is a unilateral deviation from $\profile_\proto$ for $\party_{i^*}$. Then, it holds that $\utility_{i^*}(\profile)<\utility_{i^*}(\profile_{i^*})$.
\end{claim}

\noindent\emph{Proof of Claim~\ref{claim:abstain_wins}}. Since $\profile$ is $\infractionPredicate_\mathrm{abs}$-compliant, it holds that $\party_{i^*}$ participates in every trace as if she were honest, with the difference that she makes at least $1$ query, whereas as an honest player she would make at most $q$ queries per time slot. Therefore, it holds that $E[\cost_{\party_{i^*}, \execution_\profile}]\geq\frac{1}{q}\cost_{i^*}$. Besides, the rewards of $\party_{i^*}$ \wrt $\profile$ are bounded by the maximum value $\universalRewardFunc(\miningpower_{\party_{i^*}}) \cdot \totalReward$. Thus,

\begin{equation}\label{eq:abstain_wins1}
\begin{split}
\utility_{i^*}(\profile)-\utility_{i^*}(\profile_\proto)&\leq\big(\universalRewardFunc(\miningpower_{\party_{i^*}}) \cdot \totalReward-E[\cost_{\party_{i^*}, \execution_\profile}]\big)-\big((1-\alpha)\cdot\universalRewardFunc(\miningpower_{\party_{i^*}}) \cdot \totalReward-\cost_{i^*}\big)=\\
&=\alpha\cdot\universalRewardFunc(\miningpower_{\party_{i^*}}) \cdot \totalReward+\big(\cost_{i^*}-E[\cost_{\party_{i^*}, \execution_\profile}]\big)
\leq\\
&\leq\alpha\cdot\universalRewardFunc(\miningpower_{\party_{i^*}}) \cdot \totalReward+\big(1-\frac{1}{q}\big)\cdot\cost_{i^*}\;.
\end{split}
\end{equation}
Recall that
\begin{equation}\label{eq:abstain_wins2}
\begin{split}
\utility_{i^*}(\profile_{i^*})-\utility_{i^*}(\profile_\proto)&= (1 - \beta_{i^*}) \cdot \universalRewardFunc(\miningpower_{\party_{i^*}}) \cdot \totalReward - \big((1 - \alpha) \cdot \universalRewardFunc(\miningpower_{\party_{i^*}}) \cdot \totalReward - \cost_{i^*}\big)=\\
&=(\alpha-\beta_{i^*}) \cdot \universalRewardFunc(\miningpower_{\party_{i^*}}) \cdot \totalReward+\cost_{i^*}\;.
\end{split}
\end{equation}
By Eq.~\eqref{eq:abstain_wins1} and~\eqref{eq:abstain_wins2} and given that $\cost>\beta_{i^*}\cdot \universalRewardFunc(\miningpower_{\party_{i^*}}) \cdot \totalReward\cdot q$, we have that

\begin{equation*}
\utility_{i^*}(\profile_{i^*})-\utility_{i^*}(\profile)\geq \frac{1}{q}\cdot\cost_{i^*}-\beta_{i^*} \cdot \universalRewardFunc(\miningpower_{\party_{i^*}}) \cdot \totalReward>0\;.
\end{equation*}
\hfill $\Diamond$\\

By Claim~\ref{claim:abstain_wins}, we have that the best response $\profile^*$ for $\party_{i^*}$ is non $\infractionPredicate_\mathrm{abs}$-compliant, as it holds that $\utility_{i^*}(\profile^*)\geq\utility_{i^*}(\profile_{i^*})$. In addition, by Eq.~\eqref{eq:abstain_wins}, we get that for every $0\leq\epsilon < \cost_{i^*} - (\beta_{i^*} - \alpha) \cdot \universalRewardFunc(\miningpower_{\party_{i^*}}) \cdot \totalReward $, it holds that $\utility_{i^*}(\profile^*)>\utility_{i^*}(\profile_{\profile_\proto})+\epsilon$, therefore $\profile^*$ is directly $\epsilon$-reachable from $\profile_\proto$ \wrt $\bar{\utility}$. Thus, $\profile^*$ is a non $\infractionPredicate_\mathrm{abs}$-compliant strategy profile that is in $\mathsf{Cone}_{\epsilon, \bar{\utility}}(\proto)$.

Consequently, for $0\leq\epsilon < \underset{i \in [\totalParties]}{\mathsf{max}}\{ \cost_i - (\beta_i - \alpha) \cdot \universalRewardFunc(\miningpower_{\party_i}) \cdot \totalReward \}$, it holds that $\mathsf{Cone}_{\epsilon, \bar{\utility}}(\proto) \subsetneq (\strategySet_{\infractionPredicate_\mathrm{abs}})^\totalParties$, \ie the protocol $\proto$ is not $(\epsilon, \infractionPredicate_\mathrm{abs})$-compliant \wrt $\bar{\utility}$.

\end{proof}

Before concluding, let us examine the variables of the bound
$\underset{i \in [\totalParties]}{\mathsf{max}}\{ \cost_i - (\beta_i - \alpha) \cdot \universalRewardFunc(\miningpower_{\party_i}) \cdot \totalReward \}$
of Theorem~\ref{thm:universal-profit}. We note that, in the context of
blockchain systems, a ``party'' is equivalent to a unit of power; therefore, a
party $\party$ that controls $\miningpower_{\party}$ of the total power, in
effect controls $\miningpower_{\party}$ of all ``parties'' that participate in
the blockchain protocol.

To discuss $\alpha$ and $\beta_i$, we first consider the liveness
property~\cite{RSA:GarKia20} of blockchain protocols. Briefly, if a protocol
guarantees liveness with parameter $u$, then a transaction which is diffused on
slot $\slot$ is part of the (finalized) ledger of every honest party on round
$\slot + u$. Therefore, assuming that the environment gives at least one
transaction to the parties, if a protocol $\proto$ guarantees liveness unless
with negligible probability $\mathsf{negl}(\secparam)$,\footnote{Recall that $\secparam$ is
$\proto$'s security parameter, while $\mathsf{negl}(\cdot)$ is a negligible function.}
then at least one block is created during the execution with overwhelming
probability (in $\secparam$).

Now, we consider $\alpha$ and  $\beta_i$.
The former is
negligible,
since consensus protocols typically guarantee liveness against a
number of crash (or Byzantine) faults, let alone if all parties are honest. The
latter, however, depends on $\party_i$'s percentage of power
$\miningpower_{\party_i}$. For instance, consider Ouroboros, which is secure if
a deviating party $\party_i$ controls less than $\frac{1}{2}$ of the staking power and all others employ $\proto$.
 Thus, if
$\miningpower_{\party_i} = \frac{2}{3}$ and $\party_i$ abstains, the protocol
cannot guarantee liveness, \ie it is probable that no blocks are created. However,
if $\miningpower_{\party_i} = \frac{1}{4}$, then
liveness is guaranteed with overwhelming probability; hence, even if $\party_i$
abstains, at least one block is typically created.
Corollary~\ref{col:universal-profit-non-compliance}
generalizes this argument, by showing
that, if enough parties participate, then at least one of them is small enough,
such that its abstaining does not result in a system halt, hence it is
incentivized to be non-compliant.

\begin{corollary}\label{col:universal-profit-non-compliance}
    Let $\proto$ be a blockchain protocol, with security parameter $\secparam$,
    which is run by $\totalParties$ parties, under the same considerations of
    Theorem~\ref{thm:universal-profit}.
    Additionally, assume that $\proto$ has liveness with security threshold $\frac{1}{x}$ in the following sense: for every strategy profile $\profile$, if $\underset{\party \in \partySet_{-\profile}}{\sum} \miningpower_{\party} < \frac{1}{x}$,
    where $\partySet_{-\profile}$ is the set of parties that deviate from $\proto$ when $\profile$ is followed, then $\proto$ guarantees liveness with overwhelming (i.e., $1-\mathsf{negl}(\secparam)$) probability.
    If $x < \totalParties$, then for (non-negligible) values
    $\epsilon < \underset{i \in [\totalParties]}{\mathsf{max}}\{\cost_i\} - \mathsf{negl}(\secparam)$,
    $\proto$ is \emph{not} $(\epsilon, \infractionPredicate_\mathrm{abs})$-compliant \wrt $\bar{\utility}$.
\end{corollary}

\begin{proof}
    First, since $\proto$ guarantees liveness even under some byzantine faults,
    $\alpha=\mathsf{negl}(\secparam)$.

    Second, if $x < \totalParties$, then there exists $i\in[n]$ such that
    $\miningpower_{\party_i} < \frac{1}{x}$. To prove this, if $\totalParties >
    x$ and $\forall j\in[n]: \miningpower_{\party_j} \geq \frac{1}{x}$ then
    $\underset{j \in [\totalParties]}{\sum} \miningpower_{\party_i} \geq \frac{n}{x} > 1$.
    This contradicts to the definition of the parties' participating power (cf. Section~\ref{sec:model}), where it holds that $\underset{j \in [\totalParties]}{\sum} \miningpower_{\party_j}=1$.

    Now consider the strategy profile $\profile$ where $\party_i$ abstains and all the other parties honestly follow $\proto$. Then, by definition, $\partySet_{-\profile}=\{\party_i\}$ and therefore,
    \[\underset{\party \in \partySet_{-\profile}}{\sum} \miningpower_{\party}=\miningpower_{\party_i} < \frac{1}{x}\;.\]
    Thus, by the assumption for $\proto$, we have that if the parties follow $\profile$, then $\proto$ guarantees liveness with $1-\mathsf{negl}(\secparam)$ probability. Hence, $\beta_i = \mathsf{negl}(\secparam)$.
    Finally, since $\universalRewardFunc(\miningpower_{\party_i}) \in [0, 1]$
    and $\totalReward$ is a finite value irrespective of the parties' strategy
    profile, the value $(\beta_i - \alpha) \cdot
    \universalRewardFunc(\miningpower_{\party_i}) \cdot \totalReward$ is also
    negligible in $\secparam$.
\end{proof}

The minimal cost $\cost^\bot_{\party_i}$ of (honest)
participation for party $\party_i$ depends on the blockchain system's mechanism. In PoW systems, where participation
consists of repeatedly performing computations, cost increases with the
percentage of mining power; for instance, controlling $51$\% of Bitcoin's
mining power for $1$ hour costs
\$$1,700,000$.\footnote{\url{https://www.crypto51.app} [February 2022]} In PoS
systems, cost is typically irrespective of staking power, since participation
consists only of monitoring the network and regularly signing messages; for
example, running a production-grade Cardano node costs \$$180$ per
month\footnote{\url{https://forum.cardano.org/t/realistic-cost-to-operate-stake-pool/40056} [January 2022]}. Therefore, considering
Corollary~\ref{col:universal-profit-non-compliance}, the upper bound
$\underset{i \in [\totalParties]}{\mathsf{max}}\{\cost_i\} - \mathsf{negl}(\secparam)$
of $\epsilon$ is typically rather large for PoS systems.

The free-rider hazard is manifested in
Algorand\footnote{\url{https://algorand.foundation}}, a cryptocurrency system
that follows the Algorand consensus protocol~\cite{EPRINT:CGMV18,EPRINT:GHMVZ17}
and employs resource-proportional rewards, as defined above. Its users own ``Algo'' tokens and
transact over a ledger maintained by ``participation nodes'', which run the
Algorand protocol and extend the ledger via blocks. Each user receives a fixed
reward\footnote{The weekly reward per owned Algo is $0.00012$
Algos. [\url{https://algoexplorer.io/rewards-calculator}, February 2022]} per
Algo token they own~\cite{algorand-rewards}, awarded with every new block.
Users may also run a participation node, but are not
rewarded~\cite{DBLP:conf/dsn/FooladgarMJR20} for doing so, and participation is
proportional to the amount of Algos that the user owns. Therefore, a party that
owns a few Algos will expectedly abstain from participation in the consensus
protocol.

\paragraph{Remark.} In summary, under resource-proportional rewards in PoS protocols, users may
choose to abstain. This can impact performance, \eg delaying block production
and transaction finalization while, in the extreme case, it could result in a
``tragedy of the commons'' situation~\cite{lloyd1833two}, where all users
abstain and the system grinds to a halt. Interestingly, this section
illustrates a difference between PoW and PoS. In PoS systems, each party's
power is registered on the ledger, without requiring any action from them. In
PoW, power becomes evident only after the party puts their hardware to work.
Therefore, the idea behind Theorem~\ref{thm:universal-profit}'s proof, which
relies on abstaining, does not necessarily hold in PoW systems, like
Fruitchains~\cite{PODC:PasShi17}, that define rewards (approximately)
proportional to each party's mining power, as identified by their hashing
operations.

\section{Block-Proportional Rewards}\label{sec:proportional}

The arguably most common type of rewards in blockchain systems is
\emph{block-proportional} rewards. Each party is rewarded
proportionally to the number of blocks it contributes to the final chain,
at the end of the execution.
Block-proportional rewards are a generalization of the \emph{proportional allocation
rule}, which, for example, is employed in Bitcoin. The proportional allocation rule
states that a party $\party$'s \emph{expected} rewards of a single block are $\miningpower_{\party}$.
As shown by Chen \etal~\cite{chen2019axiomatic}, this
is the unique allocation rule that satisfies a list of desirable properties,
namely:
\begin{inparaenum}[i)]
    \item non-negativity,
    \item budget-balance,
    \item symmetry,
    \item sybil-proofness, and
    \item collusion-proofness.
\end{inparaenum}

Our work expands the scope by considering proportional rewards \wrt blocks for
the entirety of the execution. Specifically,
Definition~\ref{def:proportional-rewards} describes block-proportional rewards, where
a party $\party$'s rewards are \emph{strictly monotonically increasing} on the number of blocks that $\party$
contributes to the chain output by the observer $\observer$. The definition considers a \emph{proportional reward function} $\proportionalRewardFunc(\cdot,\cdot)$ that takes as input the chain of $\observer$ and $\party$ and outputs a value in $[0,1]$.

\begin{definition}[Block-Proportional Rewards]\label{def:proportional-rewards}
    For an execution trace $\executionTrace$, let $\chain_{\observer,\executionTrace}$ be the chain output by $\observer$ and
    $\totalReward_{\observer,\executionTrace} \in \mathbb{R}_{\geq 0}$ be the
    total number of rewards which are distributed by the protocol, according to
    $\observer$. Let $\msgOutputSet_{\party,\executionTrace}$ be the number of
blocks in the chain output by $\observer$ which are produced by $\party$. A \emph{block-proportional reward} random variable $\reward_{\party, \execution}$ satisfies the
    following conditions:
    \begin{enumerate}
        \item $\forall \executionTrace \; \forall \party \in \partySet: \reward_{\party, \executionTrace} = \proportionalRewardFunc(\chain_{\observer,\executionTrace}, \party) \cdot \totalReward_{\observer,\executionTrace}$
        \item $\forall \executionTrace: \sum_{\party \in \partySet} \proportionalRewardFunc(\chain_{\observer,\executionTrace}, \party) = 1$
        \item $\forall \executionTrace \; \forall \party, \party' \in \partySet:\msgOutputSet_{\party,\executionTrace}>\msgOutputSet_{\party',\executionTrace}\Rightarrow \proportionalRewardFunc(\chain_{\observer,\executionTrace}, \party) > \proportionalRewardFunc(\chain_{\observer,\executionTrace}, \party')$
    \end{enumerate}
\end{definition}

\subsection{Bitcoin}\label{subsec:bitcoin}

First, we consider the Bitcoin~\cite{nakamoto2008bitcoin} blockchain
protocol. Bitcoin is a prime example of a family of protocols that links the
amount of valid blocks, that each party can produce per execution, with the
party's hardware capabilities, including:
\begin{inparaenum}[i)]
    \item Proof-of-Work-based protocols like Ethereum~\cite{wood2014ethereum},
        Bitcoin NG~\cite{eyal2016bitcoin}, Zerocash~\cite{SP:BCGGMT14};
    \item Proof-of-Space~\cite{C:DFKP15} and
        Proof-of-Space-Time~\cite{C:MorOrl19} protocols like
        SpaceMint~\cite{FC:PKFGAP18}, Chia~\cite{AC:AACKPR17,EC:CohPie18}.
\end{inparaenum}

\paragraph{Execution Details.}
Typically, protocols from the aforementioned family enforce that, when all parties follow the protocol honestly, the expected percentage of blocks
created by a party $\party$ is $\miningpower_{\party}$ of the total blocks
produced by all parties during the execution. Along the lines of the formulation in~\cite{EC:GarKiaLeo15,C:GarKiaLeo17}, in the Bitcoin
protocol, each party $\party$ can make at most $\miningpower_{\party}
\cdot \powQueryNum$ queries to the hashing oracle $\oracle_\proto$ per time slot, where $q$ is the total number of queries that all parties can make to $\oracle_\proto$ during a time slot. We note that when $\party$ follows the Bitcoin protocol, they perform exactly $\miningpower_{\party} \cdot\powQueryNum$ queries to the hashing oracle.
Each query can be seen as an independent block production trial and is
successful with probability $\difficulty$, which is a
protocol-specific ``mining difficulty'' parameter.

From the point of view of the observer $\observer$, a party $\party$ is rewarded a
fixed amount $\reward$ for each block they contribute to the chain output by
$\observer$.
Then, Bitcoin implements a special case of block-proportional rewards (cf.
Definition~\ref{def:proportional-rewards}), such that:
\begin{itemize}
    \item The total number of rewards for $\executionTrace$ is
    $\totalReward_{\observer,\executionTrace} = \big|\chain_{\observer,\executionTrace}\big| \cdot \reward=\Big(\sum_{\hat{\party} \in \partySet} \msgOutputSet_{\hat{\party},\executionTrace}\Big)\cdot \reward$,
    where $|\cdot|$ denotes the length of a chain in blocks.

    \item The proportional reward function $\proportionalRewardFunc(\cdot,\cdot)$ is defined as
    $\proportionalRewardFunc\big(\chain_{\observer,\executionTrace}, \party\big) =\frac{\msgOutputSet_{\party,\executionTrace}}{\big|\chain_{\observer,\executionTrace}\big|}= \frac{\msgOutputSet_{\party,\executionTrace}}{\sum_{\hat{\party} \in \partySet} \msgOutputSet_{\hat{\party},\executionTrace}}$.
\end{itemize}
\noindent Thus, by Definition~\ref{def:proportional-rewards}, we have that:
\begin{equation}\label{eq:bitcoin_reward}
\forall \executionTrace \; \forall \party \in \partySet: \reward_{\party, \executionTrace} = \proportionalRewardFunc(\chain_{\observer,\executionTrace}, \party) \cdot \totalReward_{\observer,\executionTrace}=\msgOutputSet_{\party,\executionTrace}\cdot \reward\;.
\end{equation}

In Bitcoin, on each time slot a party keeps a local chain,
which is the longest among all available chains. If multiple longest chains
exist, the party follows the chronological ordering of messages.

Following, we assume that none of the participating parties has complete control
over message delivery. Therefore, when two parties $\party, \party'$
produce blocks for the same index on the same time slot, it may be unclear
which is adopted by third parties that follow the protocol, \ie depending on which arrives first.

Furthermore, the \emph{index} of each block $\mesg$ is an integer that identifies
the distance of $\mesg$ from $\genesis$, \ie its height in the tree of
blocks. Blocks on the same height, but different branches, have the same
index but are non-equivalent (recall, that two messages are equivalent if their
hash is equal).

\paragraph{Bitcoin is an Approximate Nash Equilibrium \wrt Reward.}
We prove that under our model, the Bitcoin protocol is a $\Theta(\delta^2)$-Nash equilibrium \wrt the utility \emph{Reward} (hence, by Thoerem~\ref{thm:eq_comp}, it is also a $(\Theta(\delta^2), \infractionPredicate)$-compliant w.r.t. the same utility, where $\infractionPredicate$ is any associated infraction predicate) and any synchronous router. By Definition~\ref{def:utility} and Eq.~\eqref{eq:bitcoin_reward}, we have that when parties follow the strategy profile $\sigma$, the utility $U_{\party}$ of party $\party$ is
\begin{equation}\label{eq:bitcoin_utility}
    U_{\party}(\sigma)=E\big[\msgOutputSet_{\party,\execution_\sigma}\big]\cdot \reward\;,
\end{equation}
where $\msgOutputSet_{\party,\execution_\sigma}$ is the (random variable) number of blocks produced by $\party$ in the chain output by $\observer$ and $\reward$ is the fixed amount of rewards per block. Our analysis considers typical values of the success probability $\delta$, sufficiently small such that $\delta\cdot q < 1$ (recall that $q$ is the total number of oracle queries available to all parties per slot).
We say that party $\party$ is \emph{successful} during time slot $\slot$, if $\party$ manages to produce at least one block, \ie at least one oracle query submitted by $\party$ during $\slot$ was successful. The time slot $\slot$ is \emph{uniquely successful} for $\party$, if no other party than $\party$ manages to produce a block in $\slot$.
Theorem~\ref{thm:bitcoin_eq_approx} is the main result of this subsection, which uses in its proof the result of Lemma~\ref{lem:unique_succ}.

\begin{lemma}\label{lem:unique_succ}
Assume an execution trace $\executionTrace$ of the Bitcoin protocol where all parties follow the honest strategy under a synchronous router. Let $\block_1,\ldots,\block_k$ be a sequence of blocks produced by party $\party\in\mathbb{P}$ during a time slot $\slot$ that was uniquely successful for $\party$ in $\executionTrace$. Then, $\block_1,\ldots,\block_k$ will be part of the chain output by observer $\observer$.
\end{lemma}

\begin{proof}
Let $h$ be the height of $\block_1$. Then, for every $j\in\{1,\ldots,k\}$, the height of the block $\block_j$ is $h+j-1$.
Assume for the sake of contradiction that there is a $j^*\in[k]$ such that $\block_{j^*}$ is not in the observer's chain. Since each block contains the hash of the previous block in the chain of $\observer$, the latter implies that the subsequence $\block_{j^*},\ldots,\block_k$ is not in the observer's chain. There are two reasons that $\block_{j^*}$ is missing from the observer's chain.
\begin{enumerate}
\item The observer $\observer$ never received $\block_{j^*}$. However, after the end of time slot $\slot$, $\observer$ will be activated and fetch the messages included in its $\textsc{Receive}_\observer()$ string. Therefore, the case that $\observer$ never received $\block_{j^*}$ cannot happen.

\item The observer has another block $\block'_{j^*}$ included in its chain that has the same height, $h+j^*-1$, as $\block_{j^*}$. Since $\slot$ was uniquely successful for $\party$ in $\executionTrace$, the block $\block'_{j^*}$ must have been produced in a time slot $\slot'$ that is different than $\slot$. Assume that $\party$ produced the block sequence $\block'_{j^*},\ldots,\block'_{k'}$ during $\slot'$. We examine the following two cases:
\begin{enumerate}
\item $\slot>\slot'$: then $\block'_{j^*}$ was produced before $\block_{j^*}$, so in time slot $\slot'+1$ all parties received (at least) the sequence $\block'_{j^*},\ldots,\block'_{k'}$. All parties select the longest chain, so the chain that they will select will have at least $h+k'-1\geq h+j^*-1$ number of blocks in  $\slot'+1$. Thus, for time slot $\slot\geq \slot'+1$ the parties submit queries for producing blocks which height is at least $h+k'>h+j^*-1$. So at time slot $\slot$, the party $\party$ cannot have produced a block which height is $\leq h+j^*-1$.
\item $r<r'$: then $\block_{j^*}$ was produced before $\block'_{j^*}$. So in time slot $r+1$ all parties receive the sequence $\block_1,\ldots,\block_k$. Thus, they will adopt a chain with at least $h+k-1\geq h+j^*-1$ number of blocks. Thus, for time slot $r'\geq r+1$ the parties submit queries for producing blocks with height at least $h+k>h+j^*-1$. Therefore, $\party$ cannot have produced a block of height $\leq h+j^*-1$ during time slot $r'$.
\end{enumerate}
\end{enumerate}

By the above, $\block_{j^*}$ is a block with height $h+j^*-1$ received by $\observer$ and no other block with height $h+j^*-1$ is included in $\observer$'s chain. Since $\observer$ adopts the longest chain, there must be a block with height $h+j^*-1$ that is included in its chain. It is straightforward that this block will be $\block_{j^*}$, which leads to contradiction.

\end{proof}

\begin{theorem}\label{thm:bitcoin_eq_approx}
    Let:
    \begin{inparaenum}[i)]
        \item $N \geq \kappa^c$ be the number of time slots of the execution, where $\kappa$ is the security parameter and $c$ is a sufficiently large constant;
        \item $\bar{\utility}$ be the utility vector where each party employs the utility Reward;
        \item $\adversary$ be a synchronous router.
    \end{inparaenum}
    The Bitcoin protocol is an $\epsilon$-Nash equilibrium \wrt $\bar{\utility}$ under $\adversary$, for $\epsilon:=\frac{NRq^2}{2}\delta^2$ .
\end{theorem}

\begin{proof}
    Consider a protocol execution where all parties follow the honest strategy $\Pi$, with $\profile_\proto$ denoting the profile $\langle\Pi,\ldots,\Pi\rangle$. Let $\party$ be a party and $\mu_{\party}$ be its mining power.
For $\slot\in[N]$, let $X_{\party,\slot}^{\profile_\proto}$ be the random variable that is $1$ if the time slot $\slot$ is uniquely successful for $\party$ and $0$ otherwise. By protocol description, a party $\party'$ makes $\mu_{\party'} \cdot q$ oracle queries during $\slot$, each with success probability $\delta$. Thus:
\begin{equation}\label{eq:pr_unique}
\begin{split}
&\Pr[X_{\party,r}^{\profile_\proto}=1]=\\
=&\Pr[\party\mbox{ is successful during }\slot]\cdot\\
&\quad\quad\cdot\Pr[\mbox{all the other parties produce no blocks in } \slot]=\\
    =&\Big(1-(1-\delta)^{\mu_{\party} q}\Big)\cdot\prod_{\party'\neq\party}(1-\delta)^{\mu_{\party'} q}=\\
    =&\Big(1-(1-\delta)^{\mu_{\party} q}\Big)\cdot(1-\delta)^{(1-\mu_{\party}) q} = \\
    =&(1-\delta)^{(1-\mu_{\party}) q}-(1-\delta)^q\;.
\end{split}
\end{equation}

    The random variable $X_{\party,\execution_{\profile_\proto}}:=\sum_{\slot\in[N]}X_{\party,r}^{\profile_\proto}$ expresses the number of uniquely successful time slots for $\party$. By Eq.~\eqref{eq:pr_unique}, $X_{\party}^\Pi$ follows the binomial distribution with $N$ trials and probability of success $(1-\delta)^{(1-\mu_{\party}) q}-(1-\delta)^q$. Therefore:
    $E\big[X_{\party,\execution_{\profile_\proto}}\big]=N\Big((1-\delta)^{(1-\mu_{\party}) q}-(1-\delta)^q\Big)$.

Let $\msgOutputSet_{\party,\execution_{\profile_\proto}}$ be the number of blocks produced by $\party$ included in the chain output by the observer $\observer$. In a uniquely successful time slot $\slot$, $\party$ produces at least one block, and by Lemma~\ref{lem:unique_succ}, all the blocks that $\party$ produces during $\slot$ will be included in the chain output by the observer. Therefore, for all random coins $\msgOutputSet_{\party,\execution_{\profile_\proto}}\geq X_{\party,\execution_{\profile_\proto}}$ and so it holds that:
\begin{equation}\label{eq:lower_bound}
    E\big[\msgOutputSet_{\party,\execution_{\profile_\proto}}\big]\geq E\big[X_{\party,\execution_{\profile_\proto}}\big]=N\Big((1-\delta)^{(1-\mu_{\party}) q}-(1-\delta)^q\Big)\;.
\end{equation}

Now assume that $\party$ decides to unilaterally deviate from the protocol, following a strategy $S$. Let $\profile$ denote the respective strategy profile.
    Let $Z_{\party,\execution_\profile}$ be the number of blocks that $\party$ produces by following $S$ and $\msgOutputSet_{\party,\execution_\profile}$ be the number of blocks produced by $\party$ that will be included in the chain output by $\observer$. Clearly, for all random coins $\msgOutputSet_{\party,\execution_\profile}\leq Z_{\party,\execution_\profile}$. Without loss of generality, we may assume that $\party$ makes all of their $N\mu_{\party} q$ available oracle queries (indeed, if $\party$ made less than $N\mu_{\party} q$ queries then on average it would produce less blocks). Thus, we observe that $Z_{\party,\execution_\profile}$ follows the binomial distribution with $N\mu_{\party} q$ trials and probability of success $\delta$. Thus:
\begin{equation}\label{eq:upper_bound}
    E\big[\msgOutputSet_{\party,\execution_\profile}\big]\leq\big[Z_{\party,\execution_{\profile}}\big]=N\mu_{\party} q\delta\;.
\end{equation}

    By definition of $\utility_{\party}$ in Eq.~\eqref{eq:bitcoin_utility} and Eq.~\eqref{eq:lower_bound} and~\eqref{eq:upper_bound}, for fixed block rewards $R$ and for every strategy $S$ that $\party$ may follow, we have:
\begin{equation}\label{eq:upper-lower}
\begin{split}
    \utility_{\party}(\profile)&-\utility_{\party}(\profile_\proto)=E\big[\msgOutputSet_{\party,\execution_\profile}\big]\cdot R-E\big[\msgOutputSet_{\party,\execution_{\profile_\proto}}\big]\cdot R\leq\\
    &\leq N\mu_{\party} q\delta R- N\Big((1-\delta)^{(1-\mu_{\party}) q}-(1-\delta)^q\Big)R\;.
\end{split}
\end{equation}
By Bernoulli's inequality, we have that:
    $(1-\delta)^{(1-\mu_{\party}) q}\geq1-(1-\mu_{\party}) q\delta$.
Besides, by binomial expansion, and the assumption that $\delta\cdot q<1$ we have that:
$(1-\delta)^q\leq1-q\delta+\frac{q^2}{2}\delta^2$.
Thus, by applying the two above inequalities in Eq.~\eqref{eq:upper-lower}, we get that:
\begin{equation}\label{eq:btc_approx}
\begin{split}
    \utility_{\party}(\profile)-\utility_{\party}(\profile_\proto)& \leq N\mu_{\party} q\delta R- N\Big((1-\delta)^{(1-\mu_{\party}) q}-(1-\delta)^q\Big)R \leq \\
   & \leq N\mu_{\party} q\delta R- N\Big(1-(1-\mu_{\party}) q\delta-\Big(1-q\delta+\frac{q^2}{2}\delta^2\Big)\Big)R=\\
    &= N\mu_{\party} q\delta R-N\Big(\mu_{\party} q\delta-\frac{q^2}{2}\delta^2\Big)R = \\
   & = \frac{NRq^2}{2}\delta^2
\end{split}
\end{equation}

By Eq.~\eqref{eq:btc_approx}, Bitcoin is an $\epsilon$-Nash equilibrium for $\epsilon:=\frac{NRq^2}{2}\delta^2$.

\end{proof}

\paragraph{Remark.} The result of Theorem~\ref{thm:bitcoin_eq_approx} is in agreement with previous
works~\cite{KrollDaveyFeltenWEIS2013,kiayias2019coalitionsafe}, while similar results
exist \wrt \emph{profit}~\cite{EC:BGMTZ18}.
Nonetheless, there are a few remarks to be made. A well-known implication
from the selfish mining attack~\cite{FC:EyaSir14,FC:SapSomZoh16} is that
Bitcoin is not an equilibrium \wrt relative rewards. However,
selfish mining relies on withholding a block's publication, which is a
compliant behavior (\wrt $\infractionPredicate_{bc}$,
Definition~\ref{def:blockchain-infraction}). Second, our analysis assumes fixed
difficulty, while Bitcoin operates under variable difficulty, which is computed
on regular intervals depending on the active mining power. Various interesting
result exist for the variable difficulty setting:
\begin{inparaenum}[i)]
    \item \cite{grunspan2019profitability} showed that Bitcoin is not an
        equilibrium \wrt rewards, as selfish mining is more
        profitable;
    \item \cite{DBLP:conf/ec/FiatKKP19} showed that Bitcoin is not an
        equilibrium \wrt profit, as miners stop performing \emph{some}
        hashing queries to artificially reduce the difficulty;
    \item \cite{DBLP:conf/ec/GorenS19} showed that Bitcoin is not an
        equilibrium \wrt profit in some cases (depending on the cost
        mechanism) and seemingly also non-compliant \wrt
        $\infractionPredicate_\mathrm{abs}$, as miners take turns shutting down
        and resuming operations per difficulty adjustment epoch.
\end{inparaenum}
Therefore, Bitcoin's compliance under alternative utilities,
variable difficulty, and alternative infraction
predicates (\eg enabling some flexibility in the amount of
hashing queries) is a promising line of future research.

\subsection{Proof-of-Stake}\label{subsec:pos}

Proof-of-Stake (PoS) systems differ from Bitcoin in a few points. Typically, the
execution of such systems is organized in \emph{epochs}, each consisting of a
fixed number $\epochLength$ of time slots. On each slot, a specified set of
parties is eligible to participate in the protocol. Depending on the protocol,
the leader schedule of each epoch may or may not be a priori public.

The core difference with PoW concerns the power $\miningpower_{\party}$ of each
party. In PoS, $\miningpower_{\party}$ represents their \emph{stake} in the
system, \ie the number of coins that $\party$ owns. Stake is dynamic,
therefore the system's coins may change hands and the leader schedule of each
epoch depends on the stake distribution at the beginning of the
epoch.
\footnote{In reality, the snapshot of the stake distribution is retrieved at an earlier point of the previous epoch, but we can employ this simplified version without loss of generality.}
As in Bitcoin, each party participates
proportionately to their power, so the expected ratio of slots for which
$\party$ is leader over the total number of the epoch's slots is
$\miningpower_{\party}$.

Also, in PoS protocols, the oracle $\oracle_\proto$ does not perform hashing
as in Bitcoin. Instead, it is parameterized by the leader schedule and
typically performs signing. A signature output by $\oracle_\proto$ is valid
if and only if the input message is submitted by the slot leader in time.
This introduces two important
consequences: i) only the slot leader can produce valid messages during a
given slot; ii) the leader can produce as many valid messages as the number
of possible queries to $\oracle_\proto$.

In the next paragraphs, we use the following notation:

\begin{itemize}

    \item $\cost$: the cost of a single query to $\oracle_\proto$;
    \item $\reward$: the (fixed) reward per block;
    \item $\epoch$: the number of epochs in an execution;
    \item $\epochLength$: the number of slots per epoch;
    \item $\miningpower_{\party, j}$: the power of party $\party$ on epoch $j$.

\end{itemize}

\subsubsection{Single-Leader Proof-of-Stake}\label{subsec:single-leader-pos}

As before, we analyze a representative of a family of protocols; the family is
single-leader PoS (SL-PoS) and the representative is Ouroboros~\cite{C:KRDO17}.
The SL-PoS family includes systems like
EOS\footnote{\url{https://developers.eos.io/welcome/latest/protocol/consensus_protocol}}
and Ouroboros
BFT~\cite{EPRINT:KiaRus18}, while the analysis is also applicable (with minor
changes) on Gasper (cf. Appendix~\ref{subsec:gasper}).  We again utilize the
blockchain infraction predicates (cf.
Definition~\ref{def:blockchain-infraction}). Section~\ref{sec:universal} showed
that resource-proportional rewards do not necessarily guarantee compliance. Ouroboros, as a
consensus protocol, does not define rewards, but the authors proposed a complex
reward mechanism that approximates resource-proportional rewards. Nonetheless, the
Ouroboros implementation (in Cardano) employs block-proportional rewards, so we
will also consider fixed rewards per block (cf. Definition~\ref{def:proportional-rewards}).

On each slot, Ouroboros defines a single party, the ``slot leader'', as
eligible to create a valid message. Specifically, the protocol restricts that
a leader cannot extend the chain with multiple blocks for the same slot,
therefore all honest parties extend their chain by at most $1$ block per slot.
The leader schedule is public and is computed at the
beginning of each epoch via a secure, publicly verifiable Multi-Party
Computation (MPC) sub-protocol, which cannot be biased by any single party. To
prevent long-range attacks~\cite{buterin2014stake}, Ouroboros employs a form of
rolling checkpoints (``a bounded-depth longest-chain rule''~\cite{C:KRDO17}),
\ie a party ignores forks that stem from a block older than a
(protocol-specific) limit from the adopted chain's head
(it should be noted that subsequent versions of Ouroboros did not utilize the same logic, cf. Subsection~\ref{subsec:multi-leader-pos}).

Generalizing from the above, a SL-PoS protocol demonstrates the following properties:
\begin{itemize}
    \item the execution is organized in epochs;
    \item within each epoch, a single party (the \emph{leader}) is eligible to
        produce a message per index;
    \item a party which is online considers the blocks of each past epoch
        finalized (\ie does not remove them in favor of a competing, albeit
        possibly longer, chain);
    \item no single party with power less than $\frac{1}{2}$ can bias the
        epoch's leader schedule.
\end{itemize}

\paragraph{Synchronous network.}
First, we assume a diffuse functionality parameterized by a synchronous router (cf. Section~\ref{sec:preliminaries}).
Theorem~\ref{thm:compliant-ouroboros-synchronous} shows that SL-PoS with
block-proportional rewards is an $\epsilon$-Nash equilibrium for negligible $\epsilon$ (hence, by Theorem~\ref{thm:eq_comp}, it is also
$(\epsilon,\infractionPredicate)$-compliant, $\infractionPredicate$ being any
associated infraction predicate); this result is in line with the incentives'
analysis of Ouroboros~\cite{C:KRDO17}.
We remark that \cite{DBLP:conf/ec/Brown-CohenNPW19} explored the same setting
(synchronous, longest-chain PoS) and identified a selfish
mining-like attack against so-called ``predictable'' protocols, like
Ouroboros. This result was later refined
by~\cite{DBLP:conf/sigecom/FerreiraW21}, showing that such attacks are profitable
for participants controlling more than $32.5$\% of total stake.
Nonetheless, that line of research targeted relative rewards and relied on
withholding a block's publication for some time; our work considers absolute
rewards and does not consider block withholding as non-compliant behavior.

\begin{theorem}\label{thm:compliant-ouroboros-synchronous}
    Assume:
    \begin{inparaenum}[i)]
        \item a synchronous router $\adversary$;
        \item $\forall \party \in \partySet: \miningpower_{\party} <$~$\frac{1}{2}$.
    \end{inparaenum}
    SL-PoS with block-proportional rewards (Definition~\ref{def:proportional-rewards} for fixed block reward $\reward$) is an $\epsilon$-Nash equilibrium \wrt utility \emph{Reward} and,
    if $\reward > \cost$, is an $\epsilon$-Nash equilibrium
    \wrt utility \emph{Profit}, both for negligible $\epsilon$ and under $\adversary$.
\end{theorem}

\begin{proof}
    To prove the statement, it suffices to show that, if the assumptions hold,
    no party can increase its reward more than $\epsilon$ by unilaterally deviating from the protocol, where $\epsilon=\mathsf{negl}(\secparam)$.

    First, if all parties control a minority of staking power, no single
    party can bias the slot leader schedule for any epoch (unless with $\mathsf{negl}(\secparam)$ probability). Therefore,
    the (maximum) expected number of slots for which each party $\party$ is
    leader is $\sum_{j \in [1, \epoch]} \epochLength \cdot
    \miningpower_{\party,j}$, where $\miningpower_{\party,j}$ is the
    percentage of staking power of $\party$ during the $j$-th epoch.

    Second, if all parties follow $\proto$, then the total expected rewards for
    each party $\party$ are $\reward \cdot \sum_{j \in [1, \epoch]}
    \epochLength \cdot \miningpower_{\party,j}$. This is a direct consequence
    of the network synchronicity assumption. Specifically, on slot $\slot$ the
    (single) leader $\party$ creates exactly one block $\block$, which extends the longest
    chain (adopted by $\party$).
    At the beginning of slot $\slot + 1$, all other parties receive
    $\block$ and, since $\block$ is now part of the (unique) longest chain, all
    parties adopt it. Consequently, all following leaders will extend the chain
    that contains $\block$, so eventually $\block$ will be in the chain output
    by $\observer$.
    Therefore, if all parties follow the protocol and no party can bias the
    leader schedule, then no party can increase its expected rewards by
    deviating from the protocol.

    Regarding profit, a leader creates a block by performing a single
    query to $\oracle_\proto$. Additionally, cost depends only on the number of
    such queries. Therefore, if the cost of performing a single query is less
    than $\reward$, then the profit per slot is larger than $0$, so abstaining
    from even a single slot reduces the expected aggregate profit; therefore,
    all parties are incentivized to participate in all slots.
\end{proof}

\paragraph{Lossy network.}
Second, we assume a lossy, randomized router (cf.
Section~\ref{sec:preliminaries}).\footnote{The randomized router is an example,
which will be used to prove the negative result of
Theorem~\ref{thm:compliant-praos}. Other routers, which would model arguably
more realistic networks, could also be considered to, possibly, achieve
compliance results.}
Theorem~\ref{thm:compliant-ouroboros-lossy} shows that SL-PoS with block
proportional rewards is \emph{not} compliant \wrt the conflicting infraction
predicate $\infractionPredicate_\mathrm{conf}$; specifically, it shows that
$\epsilon$ is upper-bounded by a large value, which is typically
non-negligible.

\begin{theorem}\label{thm:compliant-ouroboros-lossy}
    Assume:
    \begin{inparaenum}[i)]
        \item a lossy, randomized router $\adversary$ with (non-negligible) parameter $\networkLossProb$ (cf. Section~\ref{sec:preliminaries});
        \item $\forall \party' \in \partySet: \miningpower_{\party'} < \frac{1}{2}$;
        \item $\party$ is the party with maximum power $\miningpower_\party$ across the execution and $s_{\party} = \sum_{j \in [1, \epoch]} \epochLength \cdot \miningpower_{\party, j}$ is the expected number of slots for which $\party$ is leader;
        \item $(1 - \networkLossProb) \cdot \reward \gg \cost$.
    \end{inparaenum}

    SL-PoS with block-proportional rewards (cf.
    Definition~\ref{def:proportional-rewards}) is
    \emph{not} $(\epsilon, \infractionPredicate_{conf})$-compliant (cf.
    Definition~\ref{def:compliant}) \wrt :
    \begin{inparaenum}[i)]
        \item utility \emph{Reward} for (non-negligible) $\epsilon < (\networkLossProb - \networkLossProb^t) \cdot \reward \cdot s_{\party}$;
        \item utility \emph{Profit} for (non-negligible) $\epsilon < ((\networkLossProb - \networkLossProb^t) \cdot \reward - ( t - 1) \cdot \cost) \cdot s_{\party}$, where $t = \lfloor\frac{\ln( \frac{\cost}{\reward \cdot \ln(1/\networkLossProb)} )}{\ln(\networkLossProb)}\rfloor$, in both cases under $\adversary$.
    \end{inparaenum}

\end{theorem}

\begin{proof}
    Since $\networkLossProb$ is an independent probability, $\party$ cannot
    bias it in any way. Therefore, $\party$ can only try to increase its rewards (beyond
    $\reward_H$) by performing multiple queries.
    To prove the statement, we will show that a party is indeed incentivized to
    produce multiple conflicting blocks, in order to increase the probability
    that at least one of them is delivered.

    Following the honest model, a party's expected rewards are
    $$\reward_H = (1 - \networkLossProb) \cdot \reward \cdot \sum_{j \in [1, \epoch]} \epochLength \cdot \miningpower_{\party, j}$$
    and its expected profit is
    $$P_H = ((1 - \networkLossProb) \cdot \reward - \cost) \cdot s_{\party}$$

    Regarding utility \emph{Reward}, assume a deviation
    $\strategy_{t}$, \st every time $\party$ is the leader of a slot, $\party$ creates $t$
    conflicting blocks. Observe that, if at least one of the $t$ blocks is not
    lost, then $\party$ will get the reward for the corresponding block
    (similar to the proof of Theorem~\ref{thm:compliant-ouroboros-synchronous}).

    The expected rewards under $\strategy_{t}$ are:
    $$\reward_{\strategy_{t}} = (1 - \networkLossProb^t) \cdot \reward \cdot s_{\party}$$

    As $t$ tends to $\infty$, $\networkLossProb^t$ tends to $0$, therefore $\reward_{\strategy}$ tends to $\reward \cdot s_{\party}$, which is the maximum possible amount of rewards that $\party$ can hope for. Therefore, the deviation $\strategy_{\infty}$ is the best response and the utility improvement is:
    $(\networkLossProb - \networkLossProb^t) \cdot \reward \cdot s_{\party}$.

    When $\party$ performs $t$ queries, their aggregate cost is $t \cdot \cost$.
    Therefore, the profit under $\strategy_{t}$ is:
    $$P_{\strategy_{t}} = ((1 - \networkLossProb^t) \cdot \reward - t \cdot \cost) \cdot s_{\party}$$

    Since $(1 - \networkLossProb) \cdot \reward \gg \cost$, performing multiple
    queries yields higher expected profit. However, as more queries are
    performed, the increase rate of the profit $P_{\strategy_{t}}$ decreases,
    until a point where the expected rewards become less than the cost of
    performing $t$ queries.
    To compute this threshold we consider the function of the difference between the profit of the honest protocol and $\strategy_t$:
    \begin{equation*}
    \begin{split}
        &f(t) = ( ((1 - \networkLossProb^t) \cdot \reward - t \cdot \cost) - ((1 - \networkLossProb) \cdot \reward - \cost) ) \cdot s_{\party} \nonumber \Rightarrow \\
       \Rightarrow & f(t) = ((\networkLossProb - \networkLossProb^t) \cdot \reward - ( t - 1) \cdot \cost) \cdot s_{\party} \nonumber
     \end{split}
    \end{equation*}
    To compute the $t$ for which $f(t)$ is maximized, we consider the point where the derivative of $f$ is $0$:
   \begin{equation*}
    \begin{split}
        &f'(t) = (-\ln(\networkLossProb) \cdot \reward \cdot \networkLossProb^t - \cost) \cdot \reward - \cost) ) \cdot s_{\party}= 0 \nonumber \Rightarrow \\
       \Rightarrow &  t = \frac{\ln( -\frac{\cost}{\reward \cdot \ln(\networkLossProb)})}{\ln(\networkLossProb)}\;.
        \end{split}
    \end{equation*}

\end{proof}

The lossy network analysis is particularly of interest as it was observed
in practice. On December 2019, Cardano released its Incentivized Testnet
(ITN)\footnote{\url{https://staking.cardano.org/}}, where
stakeholders, \ie users owning Cardano's currency, participated in PoS by
forming stake pools that produced blocks. The ITN used proportional rewards
and the SL-PoS execution model of epochs and slots.
Each pool was elected as a slot leader proportionally to its
stake and received its share of an epoch's rewards based on its
\emph{performance}, \ie the number of produced blocks compared to the
\emph{expected} blocks (based on its proportional stake). Thus, pool operators
were particularly incentivized to avoid abstaining, \ie failing to produce a block when
elected. However, the network was
unstable and lossy, so forks started to form. In turn, pools were incentivized\footnote{
\url{https://www.reddit.com/r/cardano/comments/ekncza}} to ``clone'' their
nodes, \ie run multiple parallel instances, to increase network
connectivity, reduce packet loss, and also extend all possible forks. To
make matters worse, this solution both perpetuated forks and created
new ones, as clones did not coordinate but produced different blocks, even
when extending the same chain.

We note that, although a lossy network may render a PoS protocol
non-compliant, the same does not hold for PoW. As described in the
proof of Theorem~\ref{thm:compliant-ouroboros-lossy}, a party produces multiple
blocks per slot to maximize the probability that one of them is eventually
output by $\observer$. Notably, since the PoS protocol restricts that at most
one block extends the longest chain per slot, these blocks are necessarily
conflicting. However, PoW ledgers do not enforce such restriction; therefore, a
party would instead create multiple consecutive (instead of parallel,
conflicting) blocks, as covered in the proof of
Theorem~\ref{thm:bitcoin_eq_approx}, which yields maximal expected rewards even
under a lossy network.

\subsubsection{Multi-Leader Proof-of-Stake}\label{subsec:multi-leader-pos}

We now turn to multi-leader PoS (ML-PoS) and Ouroboros Praos~\cite{EC:DGKR18}, a representative of a family alongside Ouroboros
Genesis~\cite{CCS:BGKRZ18}, Peercoin~\cite{king2012ppcoin}, and Tezos' baking
system~\cite{tezos-pos}. These protocols are similar SL-PoS, but with a core difference:
multiple parties may be chosen as leaders for the same slot.
As Theorem~\ref{thm:compliant-praos} shows, ML-PoS protocols are not compliant for block-proportional rewards.
The core idea is the same as with SL-PoS under a lossy network:
a party is incentivized to produce multiple blocks
to decrease the probability that a competing leader's competing block
is adopted over their own. We note that, although consensus doesn't
enforce a tie-breaking policy for competing messages,
parties typically opt for the message
that arrives first. The dependency on
randomized routing is also worth noting. Since alternative routers could yield
positive results, an interesting research direction is to explore the class of
routers under which compliance holds, possibly
avoiding infractions via specially-crafted peer-to-peer
message passing protocols.

\begin{theorem}\label{thm:compliant-praos}
    Assume:
    \begin{inparaenum}[i)]
        \item a synchronous, randomized router $\adversary$ (cf. Section~\ref{sec:preliminaries});
        \item $\forall \party' \in \partySet: \miningpower_{\party'} < \frac{1}{2}$;
        \item $\party$ is the party with maximum power $\miningpower_\party$ across the execution and $s_{\party} = \sum_{j \in [1, \epoch]} \epochLength \cdot \miningpower_{\party, j}$ is the expected number of slots \st $\party$ is leader;
        \item $(1 - \networkLossProb) \cdot \reward \gg \cost$.
    \end{inparaenum}
    Let $p_l$ be the (protocol-dependent) probability that multiple leaders are elected in the same slot.

    ML-PoS with block-proportional rewards (cf.
    Definition~\ref{def:proportional-rewards}) is \emph{not} $(\epsilon,
    \infractionPredicate_{conf})$-compliant (cf.
    Definition~\ref{def:compliant}) \wrt :
    \begin{inparaenum}[i)]
        \item utility \emph{reward} for (non-negligible) $\epsilon < \frac{p_l}{2} \cdot \reward \cdot s_{\party}$;
        \item utility \emph{profit} for (non-negligible) $\epsilon < (\frac{t - 1}{2 \cdot (t+1)} \cdot p_l \cdot \reward - ( t - 1) \cdot \cost) \cdot s_{\party}$, where $t = \lfloor\sqrt{\frac{p_l \cdot \reward}{\cost}}\rfloor - 1$, in both cases under $\adversary$.
    \end{inparaenum}

\end{theorem}

\begin{proof}
    First, we define a bad event $E$, during which the expected rewards of
    party $\party$ are less if following $\proto$, compared to a non-compliant
    strategy.

    Let $\slot$ be a slot during which $\party$ is leader, along with a
    different party $\party'$; this occurs with probability $p_l$. Also, on
    slot $\slot + 1$ a single leader $\party''$ exists. $E$ occurs if
    $\party''$ receives a block $\block'$ from $\party'$ before a block
    $\block$ from $\party$ and thus extends the former. If $E$ occurs, the
    reward for slot $\slot$ is credited to $\party'$.

    Let $p_E$ be the probability that $E$ occurs. First, $p_E$ depends on the
    probability $p_l$ that multiple leaders exist alongside $\party$. $p_l$
    depends on the protocol's leader schedule functionality, but is
    typically non-negligible. Second, it depends on the order delivery of
    $\block, \block'$; since the delivery is
    randomized, the probability $p_n$
    that $\block'$ is delivered before $\block$ is at most $p_n = \frac{1}{2}$.
    Therefore, it holds $p_E = p_l \cdot p_n = \frac{1}{2} p_l$, which is
    non-negligible.

    If all parties follow $\proto$, the expected reward of $\party$ is at most
    $\reward_H = (1 - \frac{1}{2} p_l) \cdot \reward \cdot s_{\party}$.

    The proof now follows the same reasoning as
    Theorem~\ref{thm:compliant-ouroboros-lossy}. Specifically, we will show a
    deviation \st $\party$ creates multiple blocks to increase the probability
    that at least one of them is delivered to the other parties first.

    Regarding utility \emph{Reward}, assume a deviation
    $\strategy_{t}$, \st every time $\party$ is the leader of a slot, $\party$ creates $t$
    conflicting blocks. The expected rewards under $\strategy_{t}$ are:
    $$\reward_{\strategy_{t}} = (1 - \frac{1}{t+1} p_l) \cdot \reward \cdot s_{\party}$$

    As $t$ tends to $\infty$, $\frac{1}{t+1}$ tends to $0$, therefore
    $\reward_{\strategy}$ tends to $\reward \cdot s_{\party}$, which is the
    maximum possible amount of rewards that $\party$ can hope for.  Therefore,
    the deviation $\strategy_{\infty}$ is the best response and the utility
    improvement is: $\frac{p_l}{2} \cdot \reward \cdot s_{\party}$.

    Regarding costs, the aggregate cost of $t$ queries is $t \cdot \cost$.
    Therefore, the profit under $\strategy_{t}$ is:
    $$P_{\strategy_{t}} = ( (1 - \frac{1}{t+1} p_l) \cdot \reward - t \cdot \cost ) \cdot s_{\party}$$

    As in Theorem~\ref{thm:compliant-ouroboros-lossy}, our goal is to find the
    maximum $t$ \st the utility increase of $P_{\strategy_{t}}$ is maximized.
    Now, the utility increase function is:
    \begin{equation*}
    \begin{split}
        &f(t) = ( ((1 - \frac{1}{t+1} p_l) \cdot \reward - t \cdot \cost) - ((1 - \frac{1}{2} p_l) \cdot \reward - \cost) ) \cdot s_{\party} \nonumber \Rightarrow \\
    \Rightarrow &    f(t) = (\frac{t - 1}{2 \cdot (t+1)} \cdot p_l \cdot \reward - ( t - 1) \cdot \cost) \cdot s_{\party} \nonumber
     \end{split}
    \end{equation*}

    To compute the $t$ for which $f(t)$ is maximized, we consider the point where the derivative of $f$ is $0$:
     \begin{equation*}
    \begin{split}
       & f'(t) =  (\frac{p_l \cdot \reward}{(t+1)^2} - \cost) ) \cdot s_{\party}= 0 \nonumber \Rightarrow \\
       \Rightarrow & t = \sqrt{\frac{p_l \cdot \reward}{\cost}} - 1\;.
    \end{split}
    \end{equation*}

\end{proof}

\section{Compliant Non-equilibria}\label{sec:compliant-non-equilibrium}

So far, our positive results \wrt compliance relied on showing that the
protocol is an equilibrium. In this section, we demonstrate the distinction
between the two notions via protocols that are compliant \wrt a non-trivial
infraction predicate, but not Nash equilibria.

\paragraph{Protocol specifications.}
To do this, we consider a simple, yet typical, Single-Leader PoS (SL-PoS)
protocol $\proto$, which features the following characteristics:
\begin{itemize}
    \item The slot leaders are randomly elected, directly proportional to their staking power.
    \item The staking power $\miningpower_{\party}$ of a party $\party$ remains fixed across the execution (this always holds when employing resource-proportional rewards).
    \item Only if elected as slot leader, $\party$ will make a single query to the signing oracle $\oracle_\proto$ and casts the received block at the specific time slot.
    \item The single query cost is $\cost$, a (typically small) polynomial on the security parameter $\secparam$.
\end{itemize}

We provide two results. First, we show that, under resource-proportional
rewards, $\proto$ is compliant \wrt $\infractionPredicate_{conf}$ but
non-compliant (hence not an equilibrium) \wrt $\infractionPredicate_{abs}$.
Next, we show that, under block-proportional rewards, $\proto$ is compliant \wrt
$\infractionPredicate_{conf}$ but is susceptible to \emph{selfish signing}, an
attack akin to selfish mining, when it is ``predictable'', i.e., the leader schedule is known in advance~\cite{DBLP:conf/ec/Brown-CohenNPW19}..

\subsection{Proof-of-Stake under Resource Proportional Rewards}\label{sec:PoS_comp}

First, we consider $\proto$ under resource-proportional rewards (cf. Definition~\ref{def:universal-rewards}) and utility Profit and investigate its compliance \wrt the two types of attacks captured by $\infractionPredicate_\mathrm{conf}$ and $\infractionPredicate_\mathrm{abs}$ (cf. Definition~\ref{def:blockchain-infraction}).
The goal of this study is to show  that, under a well-defined interval of approximation factor values, the protocol, although non $\infractionPredicate_\mathrm{abs}$-compliant (hence, also non approximate Nash equilibrium), operates in a $\infractionPredicate_\mathrm{conf}$-compliant manner. We note that non $\infractionPredicate_\mathrm{abs}$-compliance is consistent with Theorem~\ref{thm:universal-profit}; in particular, it applies that result and assigns concrete values to that theorem's generic parameters.

\begin{theorem}\label{thm:PoS_profit}
    Let:
    \begin{inparaenum}[i)]
        \item $\proto$ be the SL-PoS blockchain protocol specified in Section~\ref{sec:PoS_comp};
        \item $\adversary$ be a synchronous router;
        \item $\totalReward$ be the total rewards distributed by the protocol;
       \item $\universalRewardFunc: [0, 1] \rightarrow [0, 1]$ be the identity resource-proportional reward function, i.e., $\universalRewardFunc(\miningpower_{\party})=\miningpower_{\party}$;
        \item $N\geq\kappa^c$ be the number of time slots of the execution, where $\secparam$ is the security parameter and $c$ is a sufficiently large constant;
        \item $\party_\mathrm{max}$ is the party with the maximum staking power $ \miningpower_{\party_{\mathrm{max}}}$.
    \end{inparaenum}
    If $ \miningpower_{\party_{\mathrm{max}}}<\frac{1}{2}$, then the following hold:
    \begin{inparaenum}[i)]
        \item for every $\epsilon \geq 0$, $\proto$ is $(\epsilon,\infractionPredicate_\mathrm{conf})$-compliant \wrt utility \emph{Profit};
        \item for $\epsilon_\mathrm{max}:=\miningpower_{\party_{\mathrm{max}}}\cdot N\cdot\cost - \miningpower_{\party_{\mathrm{max}}}^{N+1}\cdot \totalReward$ and every $\epsilon<\epsilon_\mathrm{max}$, $\proto$ is \emph{not} $(\epsilon,\infractionPredicate_\mathrm{abs})$-compliant \wrt utility \emph{Profit}, in both cases under $\adversary$.
    \end{inparaenum}

\end{theorem}

\begin{proof} We begin by introducing a useful notion. We say that a strategy profile $\profile=\langle\strategy_1,\ldots,\strategy_\totalParties\rangle$ is \emph{$\infractionPredicate_\mathrm{conf}$-agnostic}, if for every $i\in[\totalParties]$, the strategy $\strategy_i$ of $\party_i$ does not depend on the $\infractionPredicate_\mathrm{conf}$-compliance of the other parties' strategies. For instance, $\strategy_i$ does not include checks such as ``if some party is creating conflicting blocks, then create conflicting blocks'', or ``if $\party_j$ does not create conflicting blocks, then abstain''. Clearly, $\profile_\proto$ is $\infractionPredicate_\mathrm{conf}$-agnostic. We note that an $\infractionPredicate_\mathrm{conf}$-agnostic strategy profile may still be non $\infractionPredicate_\mathrm{conf}$-compliant, i.e., the parties may create conflicting blocks independently of other parties' behavior.

In the following claim, we prove that it is not in the parties' interest to deviate by creating conflicting blocks, when they are behaving according to a $\infractionPredicate_\mathrm{conf}$-agnostic strategy profile.

\begin{claim}\label{claim:non_conf1}
Let $\profile=\langle\strategy_1,\ldots,\strategy_\totalParties\rangle\notin (\strategySet_{\infractionPredicate_{\mathrm{conf}}})^\totalParties$ be a strategy profile, i.e., for some party $\party_i\in\partySet$ and some trace $\executionTrace$ where $\party_i$ employs $\strategy_i$, it holds that $\infractionPredicate_{\mathrm{conf}}(\executionTrace,\party_i)=1$. If $\profile$ is also $\infractionPredicate_\mathrm{conf}$-agnostic, then there is a strategy profile $\profile'\in(\strategySet_{\infractionPredicate_{\mathrm{conf}}})^\totalParties$ where $\party_i$ unilaterally deviates from $\profile$, such that $\utility_{\party_i}(\profile')>\utility_{\party_i}(\profile)$.
\end{claim}

\noindent\emph{Proof of Claim~\ref{claim:non_conf1}}. We define $\profile'$ as follows: in any execution, $\party_i$ makes only one query to $\oracle_\proto$ in all time slots that $\party_i$ decided to produce conflicting blocks \wrt $\profile$. Whether $\party_i$ decides to honestly extend the longest chain or to create a fork (e.g., by performing selfish signing) remains unchanged in any corresponding executions \wrt $\profile$ or $\profile'$.

Assume that under $\profile$, $\party_i$ is the leader for a sequence of slots $\slot,\ldots,\slot+d$, $d\geq0$, and by creating conflicting blocks it produces a tree of blocks signed by $\party_i$, rooted at some block that was already in the chain when $\slot$ was reached. Observe that the height of the said tree is at most $d+1$ (if $\party_i$ does not abstain at any of slots $\slot,\ldots,\slot+d$).
This implies that when the longest chain rule is applied by the other parties and $\Omega$, $d'\leq d+1$ blocks of $\party_i$ for that period will be included by following a longest path of length $d'$ on the tree. On the other hand, if $\party_i$ behaves \wrt $\profile'$, then a single path of exactly $d'$ blocks of $\party_i$ will be included in the chain during $\slot,\ldots,\slot+d$.
Given that $\profile$ is $\infractionPredicate_\mathrm{conf}$-agnostic, the other parties' behavior in $\profile'$ remains the same as in $\profile$. So, it holds that $E[\reward_{\party_i, \execution_{\profile'}}]= E[\reward_{\party_i, \execution_\profile}]$ and for every other party $\hat{P}\neq\party_i$, it also holds that $E[\reward_{\hat{\party}, \execution_{\profile'}}]= E[\reward_{\hat{\party}, \execution_\profile}]$.

On the other hand, since $\profile\notin (\strategySet_{\infractionPredicate_{\mathrm{conf}}})^\totalParties$, and since creating conflicting blocks costs more than making only single queries, it holds that there are traces \wrt $\profile$ where the cost of $\party_i$ is strictly larger than the cost of $\party_i$ in the corresponding trace (same random coins) \wrt $\profile'$. Thus, it holds that $E[\cost_{\party_i, \execution_{\profile'}}]< E[\cost_{\party_i, \execution_\profile}]$.

Given the above, we conclude that
\begin{equation*}
\begin{split}
\utility_{\party_i}(\profile')&=E[\reward_{\party_i, \execution_{\profile'}}] - E[\cost_{\party_i, \execution_{\profile'}}]>E[\reward_{\party_i, \execution_{\profile}}] - E[\cost_{\party_i, \execution_{\profile}}]=\utility_{\party_i}(\profile)\;.
\end{split}
\end{equation*}
\hfill $\dashv$\\

Next, in the following claim, we show that by $\proto$'s description, we can focus on  $\infractionPredicate_\mathrm{conf}$-agnostic strategy  profiles.

\begin{claim}\label{claim:static1}
Let $\profile=\langle\strategy_1,\ldots,\strategy_\totalParties\rangle$ be a non  $\infractionPredicate_\mathrm{conf}$-agnostic strategy profile. Then, there exists a strategy profile $\profile'$ that is  $\infractionPredicate_\mathrm{conf}$-agnostic and for every $\party_i\in\partySet$ it holds that $\utility_{\party_i}(\profile')\geq\utility_{\party_i}(\profile)$.
\end{claim}

\noindent\emph{Proof of Claim~\ref{claim:static1}}.
We provide a constructive proof for creating $\profile'$; by expressing the strategy $\strategy_i$ as an algorithm, since $\profile$ is non  $\infractionPredicate_\mathrm{conf}$-agnostic, for $i\in[\totalParties]$, $\strategy_i$ potentially contains in each slot $\slot$ checks of the form ``$\mathbf{if}(A)\;\mathbf{then}\{\mathsf{cmd}_A\}\;\mathbf{else}\{\mathsf{cmd}_{\neg A}\}$'', where (i) $A$ is some condition related to the other parties' strategies regarding $\infractionPredicate_\mathrm{conf}$-compliance, an (ii) by $\proto$'s description, the commands $\mathsf{cmd}_A,\mathsf{cmd}_{\neg A}$ are selected from the following types of commands:
\begin{enumerate}
\item ``do nothing'';
\item ``abstain from querying $\oracle_\proto$'';
\item ``make one query to $\oracle_\proto$ and extend the longest chain''; (honest behavior)
\item ``make one query to $\oracle_\proto$ and create a fork''; (set of commands that includes selfish signing)
\item ``make multiple queries to $\oracle_\proto$ and extend the longest chain''; (set of commands that includes creation of conflicting blocks)
\item ``make multiple queries to $\oracle_\proto$ and create forks''; (set of commands)
\item ``make at least one query to $\oracle_\proto$ but publish no block''; (set of commands)
\item ``$\mathbf{if}(B)\;\mathbf{then}\{\mathsf{cmd}_B\}\;\mathbf{else}\{\mathsf{cmd}_{\neg B}\}$''; (nested $\mathbf{if}$)
\end{enumerate}

We show how we can ``remove'' checks as above, where in the case of nested $\mathbf{if}$ we proceed from the inner to the outer layer.

An inner layer check contains commands of type 1-7. The crucial observation is that the change in the utility by the behavior that derives from each of these seven types of commands is \emph{independent} from other parties' actions regarding $\infractionPredicate_\mathrm{conf}$-compliance. Namely, the change in rewards and cost by executing a command of type 1-7 is not affected by the case that some other subset of parties have created conflicting blocks or not so far (and when).

Thus, for the check ``$\mathbf{if}(A)\;\mathbf{then}\{\mathsf{cmd}_A\}\;\mathbf{else}\{\mathsf{cmd}_{\neg A}\}$'' in slot $\slot$, we consider two strategies $\strategy_{i,A}$, $\strategy_{i,\neg A}$ defined as follows; the party $\party_i$ behaves as in $\strategy_i$ with the following modification: at slot $\slot$, $\party_i$ always executes $\mathsf{cmd}_A$ (resp. $\mathsf{cmd}_{\neg A}$) when following $\strategy_{i,A}$ (resp. $\strategy_{i,\neg A}$).

Now, let $\profile_{i,A}$ (resp. $\profile_{i,\neg A}$) be the strategy profile where $\party_i$ follows $\strategy_{i,A}$ (resp. $\strategy_{i,\neg A}$) and all the other parties follow the same strategy as $\profile$. By the description of $\strategy_{i,A}$, $\strategy_{i,\neg A}$, we show that
\[\big(\utility_{\party_i}(\profile_{i,A})\geq\utility_{\party_i}(\profile)\big)\lor\big(\utility_{\party_i}(\profile_{i,\neg A})\geq\utility_{\party_i}(\profile)\big)\;.\]
Intuitively, the above holds because the utility of $\party_i$ depends on which of the two commands $\mathsf{cmd}_{i,A}$, $\mathsf{cmd}_{i,\neg A}$ will be executed, independently of the other parties' strategies regarding $\infractionPredicate_\mathrm{conf}$-compliance. Formally, we want to prove that it always holds that
\begin{equation}\label{eq:bounded_by_maximum}
\utility_{\party_i}(\profile)\leq\mathrm{max}\big\{\utility_{\party_i}(\profile_{i,A}),\utility_{\party_i}(\profile_{i,\neg A})\big\}\;.
\end{equation}
By the definition of  $\strategy_{i,A}$, $\strategy_{i,\neg A}$, we have that

\begin{align*}
\Pr[\reward_{\party_i, \execution_{\profile}}=x|A]=\Pr[\reward_{\party_i, \execution_{\profile_{i,A}}}=x]\quad&
\quad
\Pr[\reward_{\party_i, \execution_{\profile}}=x|\neg A]=\Pr[\reward_{\party_i, \execution_{\profile_{i,\neg A}}}=x]\\
Pr[\cost_{\party_i, \execution_{\profile}}=y|A]=\Pr[\cost_{\party_i, \execution_{\profile_{i,A}}}=y]\quad&
\quad
\Pr[\cost_{\party_i, \execution_{\profile}}=y|\neg A]=\Pr[\cost_{\party_i, \execution_{\profile_{i,\neg A}}}=y]
\end{align*}
By the above, we get that

\begin{equation*}\label{eq:profile_profit-split}
\begin{split}
\utility_{\party_i}(\profile)&=E[\reward_{\party_i, \execution_{\profile}}]-E[\cost_{\party_i, \execution_{\profile}}]=\sum_x\Pr[\reward_{\party_i, \execution_{\profile}}=x]\cdot x-\sum_y\Pr[\cost_{\party_i, \execution_{\profile}}=y]\cdot y=\\
&=\Pr[A]\cdot\sum_x\Pr[\reward_{\party_i, \execution_{\profile}}=x|A]\cdot x+\Pr[\neg A]\cdot\sum_x\Pr[\reward_{\party_i, \execution_{\profile}}=x|\neg A]\cdot x-\\
&\quad-\Pr[A]\cdot\sum_y\Pr[\cost_{\party_i, \execution_{\profile}}=y|A]\cdot y-\Pr[\neg A]\cdot\sum_y\Pr[\cost_{\party_i, \execution_{\profile}}=y|\neg A]\cdot y=\\
&=\Pr[A]\cdot\Big(\sum_x\Pr[\reward_{\party_i, \execution_{\profile}}=x|A]\cdot x-\sum_y\Pr[\cost_{\party_i, \execution_{\profile}}=y|A]\cdot y\Big)  +\\
&\quad+\Pr[\neg A]\cdot\Big(\sum_x\Pr[\reward_{\party_i, \execution_{\profile}}=x|\neg A]\cdot x-\sum_y\Pr[\cost_{\party_i, \execution_{\profile}}=y|\neg A]\cdot y\Big)=\\
&=\Pr[A]\cdot\Big(\sum_x\Pr[\reward_{\party_i, \execution_{\profile_{i,A}}}=x]\cdot x-\sum_y\Pr[\cost_{\party_i, \execution_{\profile_{i,A}}}=y]\cdot y\Big)  +\\
&\quad+\Pr[\neg A]\cdot\Big(\sum_x\Pr[\reward_{\party_i, \execution_{\profile_{i,\neg A}}}=x]\cdot x-\sum_y\Pr[\cost_{\party_i, \execution_{\profile_{i,\neg A}}}=y]\cdot y\Big)=\\
&=\Pr[A]\cdot\big( E[\reward_{\party_i, \execution_{\profile_{i,A}}}]-E[\cost_{\party_i, \execution_{\profile_{i,A}}}]\big)+\Pr[\neg A]\cdot \big(E[\reward_{\party_i, \execution_{\profile_{i,\neg A}}}]-E[\cost_{\party_i, \execution_{\profile_{i,\neg A}}}]\big)=\\
&=\Pr[A]\cdot\utility_{\party_i}(\profile_{i,A})+\Pr[\neg A]\cdot \utility_{\party_i}(\profile_{i,\neg A})\leq\mathrm{max}\big\{\utility_{\party_i}(\profile_{i,A}),\utility_{\party_i}(\profile_{i,\neg A})\big\}\;.
\end{split}
\end{equation*}
Therefore,  Eq.~\eqref{eq:bounded_by_maximum} holds.
Hence, if $\utility_{\party_i}(\profile_{i,A})\geq\utility_{\party_i}(\profile)$ (resp. $\utility_{\party_i}(\profile_{i,\neg A})\geq\utility_{\party_i}(\profile)$), then we can ``keep'' $\strategy_{i,A}$ (resp. $\strategy_{i,\neg A}$) that contains one less check than $\profile$ and ``discard'' the latter. By continuing this process iteratively for all parties, we end up in a strategy profile $\profile'$ that contains no checks of the aforementioned form, i.e. $\profile'$ is  $\infractionPredicate_\mathrm{conf}$-agnostic. By construction, it holds that $\utility_{\party_i}(\profile')\geq\utility_{\party_i}(\profile)$. \hfill  $\dashv$\\

Given Claims~\ref{claim:non_conf1} and~\ref{claim:static1}, we prove the $\infractionPredicate_\mathrm{conf}$-compliance of $\proto$. In particular, we show that for every $\epsilon\geq0$, if a strategy profile $\profile$ is $\epsilon$-reachable from $\profile_\proto$, then it is $\infractionPredicate_\mathrm{conf}$-compliant.

For the sake of contradiction, assume that  $\profile$ is not $\infractionPredicate_\mathrm{conf}$-compliant. The proof is by induction on the length $\ell$ of the shortest path from $\profile_\proto$ to $\profile$.
\begin{itemize}
\item \emph{Basis}: $\ell=1$. If $\profile$ is $\infractionPredicate_\mathrm{conf}$-agnostic, then by Claim~\ref{claim:non_conf1}, $\profile$ cannot set a best response for the party $\party_i$ that unilaterally deviates from $\profile_\proto$. Besides, if $\profile$ is not $\infractionPredicate_\mathrm{conf}$-agnostic, then by Claim~\ref{claim:static1}, there is a strategy profile $\profile'$ that is $\infractionPredicate_\mathrm{conf}$-agnostic and for which it holds that $\utility_{\party_i}(\profile')\geq\utility_{\party_i}(\profile)$. Since by Claim~\ref{claim:non_conf1} $\profile'$ cannot set a best response for the party $\party_i$, neither $\profile$ can set a best response for the party $\party_i$. Therefore, $\profile$ is not $\epsilon$-directly reachable from $\profile_\proto$.
\item \emph{Induction step}: Assume that the statement holds for every $\epsilon$-reachable strategy profile with shortest path of length $\ell$. Let $\profile_\proto\rightarrow\profile_1\rightarrow\cdots\rightarrow\profile_\ell\rightarrow\profile$ be the shortest path of length $\ell+1$ from $\profile_\proto$ to $\profile$. Then, $\profile_1,\ldots,\profile_\ell$ are all $\infractionPredicate_\mathrm{conf}$-compliant. By following the same steps as in Basis, we conclude that, either $\infractionPredicate_\mathrm{conf}$-agnostic or not, $\profile$ cannot set a best response for the party that unilaterally deviates from $\profile_\ell$, thus $\profile$ is not $\epsilon$-directly reachable from $\profile_\proto$.
\end{itemize}
To prove the non $\infractionPredicate_\mathrm{abs}$-compliance of $\proto$, we apply Theorem~\ref{thm:universal-profit} by setting the probabilities $\alpha$ and $\beta_i,C_i$, $i\in[\totalParties]$ according to $\proto$'s specification. Namely,
\begin{itemize}
\item $\alpha=0$, as when all parties participate in a round then certainly a slot leader will be elected and create a block.

\item $\beta_i=\miningpower_{\party_i}^{N}$, as if only $\party_i$ abstains throughout the entire execution while the other parties remain honest, the only case that a block will not be produced is if $\party_i$ is always elected.

\item $\cost_i=\miningpower_{\party_i}\cdot E[\cost_{\party_i, \execution_\profile}]= N\cdot\cost$, by Assumption~\ref{ass:zero-cost}.

\item The assumption $\cost>\miningpower_{\party_i}^{N}\cdot\miningpower_{\party_i}\cdot\totalReward\cdot1$ holds, as $\miningpower_{\party_i}^{N}$ is typically a very small value ($\mathsf{negl}(\secparam)$).
\end{itemize}
Therefore, we set $\epsilon_\mathrm{max}:=\underset{i \in [\totalParties]}{\mathsf{max}}\{\miningpower_{\party_i}\cdot N\cdot\cost - \miningpower_{\party_i}^{N}\cdot \miningpower_{\party_i} \cdot \totalReward \}=\miningpower_{\party_\mathrm{max}}\cdot N\cdot\cost - \miningpower_{\party_\mathrm{max}}^{N+1}\cdot \totalReward$, where we make the reasonable assumption that $\totalReward$ is strictly upper bounded by $\dfrac{N\cdot\cost}{\miningpower_{\party_\mathrm{max}}^N}$, as $\dfrac{N\cdot\cost}{\miningpower_{\party_\mathrm{max}}^N}$ is typically  a very large value.

By Theorem~\ref{thm:universal-profit}, we have that for every $\epsilon<\epsilon_\mathrm{max}$, $\proto$ is \emph{not} $(\epsilon,\infractionPredicate_\mathrm{abs})$-compliant \wrt $\bar{\utility}$.

\end{proof}

\subsection{Proof-of-Stake under Relative Utilities}\label{sec:relative}

We now continue our study of the SL-PoS protocol $\proto$ above. In particular, we also assume that the protocol is \emph{predictable}, i.e., the slot leader schedule for the entire execution is globally known to the parties in advance~\cite{DBLP:conf/ec/Brown-CohenNPW19}. However, we can get similar results by studying (sufficiently large) fragments of the execution when the protocol is predictable.  We now consider block-proportional rewards (cf. Definition~\ref{def:proportional-rewards}) and a different utility function that we call \emph{Relative Profit}.
This utility is defined as the fraction of the party's expected profit over the aggregate expected rewards of all parties when the denominator is not $0$, and $0$ otherwise
\footnote{A seemingly plausible alternative approach would be to consider the fraction of the party's expected profit over the aggregate expected profit of all parties, i.e., $\utility_{\party}(\profile)= \frac{E[\reward_{\party, \execution_{\env,\adversary,\profile}} - \cost_{\party, \execution_{\env,\adversary,\profile}}]}{E\big[\sum_{\hat{\party}\in\partySet}\reward_{\hat{\party}, \execution_{\env,\adversary,\profile}}-\sum_{\hat{\party}\in\partySet}\cost_{\hat{\party}, \execution_{\env,\adversary,\profile}}\big]}$. However, in this approach, there are corner cases where the denominator becomes negative and then, the utility would not provide intuition on the parties' payoffs. By considering only the (always non-negative) aggregate expected rewards in the fraction, we avoid such problematic cases while maintaining relativity in our definition.}.
Formally, for a party $\party$ and strategy profile $\profile$:

\begin{equation}\label{eq:relative_profit}
\begin{split}
\utility_{\party}(\profile) &=
\left\{\begin{array}{ll}
\dfrac{E[\reward_{\party, \execution_{\env,\adversary,\profile}} - \cost_{\party, \execution_{\env,\adversary,\profile}}]}{E\big[\sum_{\hat{\party}\in\partySet}\reward_{\hat{\party}, \execution_{\env,\adversary,\profile}}\big]}, &\mbox{if }E\big[\sum_{\hat{\party}\in\partySet}\reward_{\hat{\party}, \execution_{\env,\adversary,\profile}}\big]>0\\
0, &\mbox{if }E\big[\sum_{\hat{\party}\in\partySet}\reward_{\hat{\party}, \execution_{\env,\adversary,\profile}}\big]=0\\
\end{array}\right.=\\
&= \left\{\begin{array}{ll}
\dfrac{E[\reward_{\party, \execution_{\env,\adversary,\profile}} ]- E[\cost_{\party, \execution_{\env,\adversary,\profile}}]}{\sum_{\hat{\party}\in\partySet}E\big[\reward_{\hat{\party}, \execution_{\env,\adversary,\profile}}\big]}, &\mbox{if }\sum_{\hat{\party}\in\partySet}E\big[\reward_{\hat{\party}, \execution_{\env,\adversary,\profile}}\big]>0\\
0, &\mbox{if }\sum_{\hat{\party}\in\partySet}E\big[\reward_{\hat{\party}, \execution_{\env,\adversary,\profile}}\big]=0\\
\end{array}\right.
\;.
\end{split}
\end{equation}

It is easy to see that Relative Profit is an extension of  Relative Rewards~\cite{FC:EyaSir14,kiayias16EC,DBLP:conf/ec/Brown-CohenNPW19,kiayias2019coalitionsafe}, where the utility is the fraction of the party's rewards over the total rewards, by now taking non-zero costs into account.  As we will shortly prove, $\proto$ is $\infractionPredicate_\mathrm{conf}$-compliant \wrt to Relative Profit, but it is not compliant under a type of deviant behavior that we call \emph{selfish singing}.
In selfish signing, a party $\party$ that knows she is going to be elected for $d+1$ consecutive time slots, $\slot,\slot+1,\ldots,\slot+d$, where $d$ is the \emph{depth} of the specific selfish signing event, can create a fork of $d+1$ consecutive blocks pointing to a block created $d+1$ steps earlier.
This results in a new longest chain such that the last $d$ blocks of the old chain will get discarded. Given that the discarded blocks belonged to other parties, then selfish signing strictly improves the Relative Profit of $\party$. The selfish signing behavior is described in Algorithm~\ref{alg:selfish-signing} and illustrated in Figure~\ref{fig:selfish-signing}.

\begin{algorithm}[h]
	\scriptsize

	\KwIn{A sequence of $d+1$ blocks $\block_{\slot-d-1}\leftarrow\cdots\leftarrow\block_{\slot-1}$.}
	\KwOut{A new fork of $d+2$ blocks $\block_{\slot-d-1}\leftarrow\block'_\slot\leftarrow\cdots\leftarrow\block'_{\slot+d}$.}\vspace{3pt}
	\For{$j\gets0$ \KwTo $d$}{
	\If{$j==0$}{
	As leader of time slot $\slot$, create a new block $\block'_\slot$ that points to $\block_{\slot-d-1}$;}
	\Else{
	As leader of time slot $\slot+j$, create a new block $\block'_{\slot+j}$ that points to $\block'_{\slot+j-1}$;
	}
	}
	\caption{Selfish signing of depth $d\geq 1$ during time slots $r,\ldots,r+d$.}
	\label{alg:selfish-signing}
\end{algorithm}

\begin{figure}[h]
\centering
\includegraphics[scale=1.0]{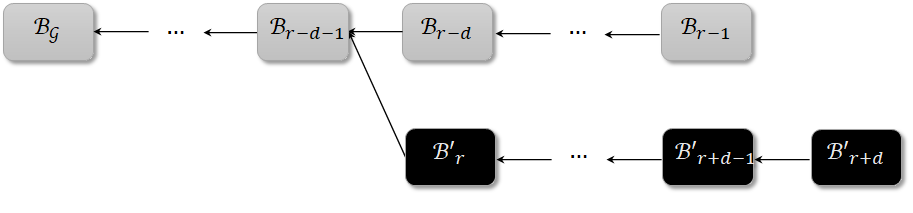}
\caption{Selfish signing of depth $d$ during time slots $\slot,\ldots,\slot+d$. The blocks $\block_{\slot-d},\ldots,\block_{\slot-1}$ will get discarded when the longest chain rule is applied.}
\label{fig:selfish-signing}
\end{figure}
Next, we describe a strategy denoted by $\strategy_\mathrm{self}$, where the party $\party$ takes advantage of $\proto$'s predictability and executes selfish signing at maximum depth whenever possible, under the condition that she never abstains or allows selfish singing to discard her own existing blocks. The strategy $\strategy_\mathrm{self}$ is described in detail in Algorithm~\ref{alg:strategy_self}. For an execution of $N$ time slots, the input is a string $\mathsf{schedule}_\party\in\{0,1\}^N$ defined as follows: for $\slot\in[N]$, $\mathsf{schedule}_\party[\slot]$ is  $1$, if $\party$ is the leader of slot $\slot$, and $0$ otherwise.

\begin{algorithm}[h]
	\scriptsize

	\KwIn{A string $\mathsf{schedule}_\party\in\{0,1\}^N$, where $N$ is the time length of the execution.}
	\KwOut{A sequence of pairs $\big((r_1,d_1),\ldots,(r_w,d_w)\big)$ indicating the time slots that selfish signing will take place at the respective depth.}\vspace{3pt}
	Initialize a list $\mathsf{strategy}_\party\leftarrow()$;\\
	Set $\slot\leftarrow 1$;\\
	\While{$\slot\leq N$}{
	\If{$\mathsf{schedule}_\party[\slot]==0$}{
	Set $\slot\leftarrow\slot+1$\tcc*{As no leader of time slot $\slot$, $\party$ takes no action}}
	\Else{
	Set $k^*\leftarrow\underset{k}{\mathrm{max}}\big\{k\big| \bigwedge_{j=0}^{k}(\mathsf{schedule}_\party[\slot+j]==1)\big\}$; \\
	Set $\ell^*\leftarrow\underset{\ell}{\mathrm{max}}\big\{\ell\big| \bigwedge_{j=1}^{\ell}(\mathsf{schedule}_\party[\slot-j]==0)\big\}$;\\
	\If{$(k^*==0)\lor(\ell^*==0)$}{
	Set $\slot\leftarrow\slot+1$\tcc*{If no selfish signing is possible, $\party$ acts as an honest party}}
	\Else{
  	Set $d\leftarrow\mathsf{min}\{k^*,\ell^*\}$;\\
  	Add $(\slot,d)$ to $\mathsf{strategy}_\party$\tcc*{$\party$ will execute Algorithm~\ref{alg:selfish-signing} in slot $\slot$ at depth $d$}
  	Set $\slot\leftarrow\slot+(d+1)$;
	}
	}
	}
	\Return $\mathsf{strategy}_\party$;
	\caption{The strategy $\strategy_\mathrm{self}$ for party $\party$.}
	\label{alg:strategy_self}
\end{algorithm}

Note that the output $\big((\slot_1,d_1),\ldots,(\slot_w,d_w)\big)$ of Algorithm~\ref{alg:strategy_self} fully determines the behavior of $\party$ throughout the execution. Namely, $\mc{P}$ acts honestly until time slot $\slot_1$ when it performs selfish signing that lasts until $\slot_1+d_1$, then it acts honestly at slots $\slot_1+(d_1+1),\ldots,\slot_2-1$ and at slot $\slot_2$ it performs selfish signing that lasts until $\slot_2+d_2$, etc.

Next, we define the infraction predicate $\infractionPredicate_\mathrm{self}$:\footnote{Observe that $\infractionPredicate_\mathrm{self}$ is a special case of the generic family of long-range attacks, \ie when a party creates a fork by extending an ancestor block, instead of the longest chain's head.}
\begin{equation*}
\infractionPredicate_\mathsf{self}(\executionTrace,\party):=
 \left\{\begin{array}{ll}
  0,&\mbox{if $\party$ never performs selfish signing in }\executionTrace\\
  1,&\mbox{otherwise}
\end{array}
\right.\;.
\end{equation*}

Having introduced $\strategy_\mathrm{self}$ and $\infractionPredicate_\mathrm{self}$, we prove the following theorem. In the theorem statement, we deploy the function $\delta(\miningpower)=5\cdot(1-\miningpower)\cdot\miningpower^2+6\cdot(1-\miningpower)^2\cdot\miningpower^3+3\cdot(1-\miningpower)^2\cdot\miningpower^4+3\cdot(1-\miningpower)^3\cdot\miningpower^4$, for $\miningpower\in(0,1)$. The function $\delta(\miningpower)$  sets a lower bound on the expected number of blocks that get discarded every $7$ consecutive time slots when a party with staking power $\miningpower$ unilaterally deviates from $\profile_\proto$ by following $\strategy_\mathrm{self}$.
It has a maximum at $\miningpower\approx0.64469$ and $\delta(0.64469)\approx1.03001$. In Table~\ref{tab:delta}, we provide some indicative values of $\delta(\miningpower)$, for $\miningpower\leq\frac{1}{2}$. By the data in Table~\ref{tab:delta}, we conclude that the block discarding rate remains significant even when the party's staking power is not particularly high, e.g., $\miningpower=\frac{1}{5}$. This result is a significant enhancement of~\cite{DBLP:conf/ec/Brown-CohenNPW19}, as it provides quantitative evidence on the effectiveness of the selfish signing strategy against predictable PoS protocols.

\begin{table}[h]
\centering
\begin{tabular}{|c|c||c|c|}
\hline
$\miningpower$&$\delta(\miningpower)$&$\miningpower$&$\delta(\miningpower)$\\
\hline
\hline
0.05&0.01258&0.3&0.41462\\
\hline
0.1&0.05032&0.35&0.53819\\
\hline
0.15&0.11228&0.4&0.66247\\
\hline
0.2&0.19624&0.45&0.77994\\
\hline
0.25&0.29864&0.5&0.88281\\
\hline
\end{tabular}
\caption{Evaluation of $\delta(\miningpower)=5\cdot(1-\miningpower)\cdot\miningpower^2+6\cdot(1-\miningpower)^2\cdot\miningpower^3+3\cdot(1-\miningpower)^2\cdot\miningpower^4+3\cdot(1-\miningpower)^3\cdot\miningpower^4$ for various values of $\miningpower$.}
\label{tab:delta}
\end{table}

\begin{theorem}\label{thm:PoS_relative_profit}
    Let:
    \begin{inparaenum}[i)]
        \item $\proto$ be the SL-PoS blockchain protocol specified in Section~\ref{sec:PoS_comp} with block-proportional rewards (Definition~\ref{def:proportional-rewards} for fixed block reward $\reward$), and assume that $\proto$ is also predictable;
        \item a synchronous router $\adversary$;
        \item $\party_\mathrm{max}$ is the party with the maximum staking power $ \miningpower_{\party_{\mathrm{max}}}$.
    \end{inparaenum}

   If $\reward>\cost$ and $ \miningpower_{\party_{\mathrm{max}}}<\frac{1}{2}$, then the following hold:
    \begin{inparaenum}[i)]
        \item for every $\epsilon \geq 0$, $\proto$ is $(\epsilon,\infractionPredicate_\mathrm{conf})$-compliant \wrt utility Relative Profit;
        \item for $\epsilon_\mathrm{max}:=\frac{\miningpower_{\party_\mathrm{max}}}{\frac{7}{\delta(\miningpower_{\party_\mathrm{max}})}-1}\cdot\frac{\reward-\cost}{\reward}$, where $\delta(\miningpower)=5\cdot(1-\miningpower)\cdot\miningpower^2+6\cdot(1-\miningpower)^2\cdot\miningpower^3+3\cdot(1-\miningpower)^2\cdot\miningpower^4+3\cdot(1-\miningpower)^3\cdot\miningpower^4$, and every $\epsilon\leq\epsilon_\mathrm{max}$, $\proto$ is \emph{not} $(\epsilon,\infractionPredicate_\mathrm{self})$-compliant \wrt utility Relative Profit, in both cases under $\adversary$.
    \end{inparaenum}

\end{theorem}

\begin{proof} First, we express Relative Profit according to the specifications of $\proto$. Let $N\geq\kappa^c$ be the number of time slots of the execution, where $\secparam$ is the security parameter and $c$ is a sufficiently large constant. For any party $\party$ and strategy profile $\profile$, let $\msgOutputSet_{\party,\executionTrace}$ be the number of
blocks in the chain output by $\observer$ which are produced by $\party$, and $Q_{\party,\executionTrace}$ be the number of
queries that $\party$ makes to $\oracle_\proto$ during the whole execution. Let $\msgOutputSet_{\party, \execution_\profile},Q_{\party, \execution_\profile}$ be the corresponding random variables. By Eq.~\eqref{eq:relative_profit}, we have that if $E\big[\msgOutputSet_{\hat{\party}, \execution_\profile}\big]>0$, then

\begin{equation}\label{eq:PoS_relative_profit}
\utility_{\party}(\profile) = \dfrac{E[\msgOutputSet_{\party, \execution_\profile}]\cdot\reward - E[Q_{\party, \execution_\profile}]\cdot\cost}{\sum_{\hat{\party}\in\partySet}E\big[\msgOutputSet_{\hat{\party}, \execution_\profile}\big]\cdot\reward}\;.
\end{equation}

\noindent By Eq.~\eqref{thm:PoS_relative_profit}, we have that for the all honest strategy profile $\sigma_\proto$ it holds that

\begin{equation}\label{eq:honest_relative_profit}
\utility_{\party}(\profile_\proto) = \dfrac{\miningpower_{\party}\cdot N\cdot\reward - \miningpower_{\party}\cdot N\cdot\cost}{N\cdot\reward}=\miningpower_{\party}
\cdot\dfrac{\reward-\cost}{\reward}\;.
\end{equation}

 To show the  $\infractionPredicate_\mathrm{conf}$-compliance of $\proto$, we follow similar steps as in the proof of Theorem~\ref{thm:PoS_profit}, yet we adapt our analysis to the Relative Profit utility.  In particular, recall that a strategy profile $\profile=\langle\strategy_1,\ldots,\strategy_\totalParties\rangle$ is \emph{$\infractionPredicate_\mathrm{conf}$-agnostic}, if for every $i\in[\totalParties]$, the strategy $\strategy_i$ of $\party_i$ does not depend on the $\infractionPredicate_\mathrm{conf}$-compliance of the other parties' strategies.

In the following claim, we prove that it is not in the parties' interest to deviate by creating conflicting blocks, when they are behaving according to a $\infractionPredicate_\mathrm{conf}$-agnostic strategy profile.

\begin{claim}\label{claim:non_conf2}
Let $\profile=\langle\strategy_1,\ldots,\strategy_\totalParties\rangle\notin (\strategySet_{\infractionPredicate_{\mathrm{conf}}})^\totalParties$ be a strategy profile, i.e., for some party $\party_i\in\partySet$ and some trace $\executionTrace$ where $\party_i$ employs $\strategy_i$, it holds that $\infractionPredicate_{\mathrm{conf}}(\executionTrace,\party_i)=1$. If $\profile$ is also $\infractionPredicate_\mathrm{conf}$-agnostic, then there is a strategy profile $\profile'\in(\strategySet_{\infractionPredicate_{\mathrm{conf}}})^\totalParties$ where $\party_i$ unilaterally deviates from $\profile$, such that $\utility_{\party_i}(\profile')>\utility_{\party_i}(\profile)$.
\end{claim}

\noindent\emph{Proof of Claim~\ref{claim:non_conf2}}. We define $\profile'$ as follows: in any execution, $\party_i$ makes only one query to $\oracle_\proto$ in all time slots that $\party_i$ decided to produce conflicting blocks \wrt $\profile$. Whether $\party_i$ decides to honestly extend the longest chain or to create a fork (e.g., by performing selfish signing) remains unchanged in any corresponding executions \wrt $\profile$ or $\profile'$.

Assume that under $\profile$, $\party_i$ is the leader for a sequence of slots $\slot,\ldots,\slot+d$, $d\geq0$, and by creating conflicting blocks it produces a tree of blocks signed by $\party_i$, rooted at some block that was already in the chain when $\slot$ was reached. Observe that the height of the said tree is at most $d+1$ (if $\party_i$ does not abstain at any of slots $\slot,\ldots,\slot+d$).
This implies that when the longest chain rule is applied by the other parties and $\Omega$, $d'\leq d+1$ blocks of $\party_i$ for that period will be included by following a longest path of length $d'$ on the tree. On the other hand, if $\party_i$ behaves \wrt $\profile'$, then a single path of exactly $d'$ blocks of $\party_i$ will be included in the chain during $\slot,\ldots,\slot+d$.
Given that $\profile$ is $\infractionPredicate_\mathrm{conf}$-agnostic, the other parties' behavior in $\profile'$ remains the same as in $\profile$. So, it holds that $E[\reward_{\party_i, \execution_{\profile'}}]= E[\reward_{\party_i, \execution_\profile}]$ and for every other party $\hat{P}\neq\party_i$, it also holds that $E[\reward_{\hat{\party}, \execution_{\profile'}}]= E[\reward_{\hat{\party}, \execution_\profile}]$.

On the other hand, since $\profile\notin (\strategySet_{\infractionPredicate_{\mathrm{conf}}})^\totalParties$, and since creating conflicting blocks costs more than making only single queries, it holds that there are traces \wrt $\profile$ where the cost of $\party_i$ is strictly larger than the cost of $\party_i$ in the corresponding trace (same random coins) \wrt $\profile'$. Thus, it holds that $E[\cost_{\party_i, \execution_{\profile'}}]< E[\cost_{\party_i, \execution_\profile}]$.

Next, we show that the non $\infractionPredicate_\mathrm{conf}$-compliance of $\profile$ implies an important fact; namely, it holds that $\sum_{\hat{\party}\in\partySet}E\big[\reward_{\hat{\party}, \execution_{\profile}}\big]>0$, so we are not in the corner case of the branching definition of $\utility_{\party_i}(\profile)$ in Eq.~\eqref{eq:relative_profit}. Indeed, recall that there are traces (at least one) where $\party_i$ creates conflicting blocks, and one of them will be part of the longest chain. We observe that even if it gets discarded by another party $\party_j$ that performs selfish signing, this block of $\party_i$ is always replaced by blocks of $\party_j$. Even in the extreme case where back-to-back selfish signing happens, the blocks of the last party that performed selfish signing will be included in the chain. More generally, when a party adds a block in the chain, then we are certain that the chain will contain at least one block in the end of the execution, so the total rewards in the traces where $\party_i$ creates conflicting blocks are a positive value.

Thus, $\sum_{\hat{\party}\in\partySet}E\big[\reward_{\hat{\party}, \execution_{\profile'}}\big]=\sum_{\hat{\party}\in\partySet}E\big[\reward_{\hat{\party}, \execution_{\profile}}\big]>0$, we conclude that
\begin{equation*}
\begin{split}
\utility_{\party_i}(\profile')&=\dfrac{E[\reward_{\party_i, \execution_{\profile'}}] - E[\cost_{\party_i, \execution_{\profile'}}]}{\sum_{\hat{\party}\in\partySet}E\big[\reward_{\hat{\party}, \execution_{\profile'}}\big]}>
\dfrac{E[\reward_{\party_i, \execution_{\profile}}] - E[\cost_{\party_i, \execution_{\profile}}]}{\sum_{\hat{\party}\in\partySet}E\big[\reward_{\hat{\party}, \execution_{\profile}}\big]}=\utility_{\party_i}(\profile)\;.
\end{split}
\end{equation*}

\hfill $\dashv$\\

Next, in the following claim, we show that by $\proto$'s description, we can focus on  $\infractionPredicate_\mathrm{conf}$-agnostic strategy  profiles.

\begin{claim}\label{claim:static2}
Let $\profile=\langle\strategy_1,\ldots,\strategy_\totalParties\rangle$ be a non  $\infractionPredicate_\mathrm{conf}$-agnostic strategy profile. Then, there exists a strategy profile $\profile'$ that is  $\infractionPredicate_\mathrm{conf}$-agnostic and for every $\party_i\in\partySet$ it holds that $\utility_{\party_i}(\profile')\geq\utility_{\party_i}(\profile)$.
\end{claim}

\noindent\emph{Proof of Claim~\ref{claim:static2}}.
We provide a constructive proof for creating $\profile'$; by expressing the strategy $\strategy_i$ as an algorithm, since $\profile$ is non  $\infractionPredicate_\mathrm{conf}$-agnostic, for $i\in[\totalParties]$, $\strategy_i$ potentially contains in each slot $\slot$ checks of the form ``$\mathbf{if}(A)\;\mathbf{then}\{\mathsf{cmd}_A\}\;\mathbf{else}\{\mathsf{cmd}_{\neg A}\}$'', where (i) $A$ is some condition related to the other parties' strategies regarding $\infractionPredicate_\mathrm{conf}$-compliance, an (ii) by $\proto$'s description, the commands $\mathsf{cmd}_A,\mathsf{cmd}_{\neg A}$ are selected from the following types of commands:
\begin{enumerate}
\item ``do nothing'';
\item ``abstain from querying $\oracle_\proto$'';
\item ``make one query to $\oracle_\proto$ and extend the longest chain''; (honest behavior)
\item ``make one query to $\oracle_\proto$ and create a fork''; (set of commands that includes selfish signing)
\item ``make multiple queries to $\oracle_\proto$ and extend the longest chain''; (set of commands that includes creation of conflicting blocks)
\item ``make multiple queries to $\oracle_\proto$ and create forks''; (set of commands)
\item ``make at least one query to $\oracle_\proto$ but publish no block''; (set of commands)
\item ``$\mathbf{if}(B)\;\mathbf{then}\{\mathsf{cmd}_B\}\;\mathbf{else}\{\mathsf{cmd}_{\neg B}\}$''; (nested $\mathbf{if}$)
\end{enumerate}

We show how we can ``remove'' checks as above, where in the case of nested $\mathbf{if}$ we proceed from the inner to the outer layer.

An inner layer check contains commands of type 1-7. The crucial observation is that the change in the utility by the behavior that derives from each of these six types of commands is \emph{independent} from other parties' actions regarding $\infractionPredicate_\mathrm{conf}$-compliance. Namely, the change in rewards and cost by executing a command of type 1-7 is not affected by the case that some other subset of parties have created conflicting blocks or not so far (and when).

Thus, for the check ``$\mathbf{if}(A)\;\mathbf{then}\{\mathsf{cmd}_A\}\;\mathbf{else}\{\mathsf{cmd}_{\neg A}\}$'' in slot $\slot$, we consider two strategies $\strategy_{i,A}$, $\strategy_{i,\neg A}$ defined as follows; the party $\party_i$ behaves as in $\strategy_i$ with the following modification: at slot $\slot$, $\party_i$ always executes $\mathsf{cmd}_A$ (resp. $\mathsf{cmd}_{\neg A}$) when following $\strategy_{i,A}$ (resp. $\strategy_{i,\neg A}$).

Now, let $\profile_{i,A}$ (resp. $\profile_{i,\neg A}$) be the strategy profile where $\party_i$ follows $\strategy_{i,A}$ (resp. $\strategy_{i,\neg A}$) and all the other parties follow the same strategy as $\profile$. By the description of $\strategy_{i,A}$, $\strategy_{i,\neg A}$, we show that
\[\big(\utility_{\party_i}(\profile_{i,A})\geq\utility_{\party_i}(\profile)\big)\lor\big(\utility_{\party_i}(\profile_{i,\neg A})\geq\utility_{\party_i}(\profile)\big)\;.\]
Intuitively, the above holds because the utility of $\party_i$ depends on which of the two commands $\mathsf{cmd}_{i,A}$, $\mathsf{cmd}_{i,\neg A}$ will be executed, independently of the other parties' strategies regarding $\infractionPredicate_\mathrm{conf}$-compliance. Formally, we want to prove that it always holds that
\begin{equation}\label{eq:bounded_by_max}
\utility_{\party_i}(\profile)\leq\mathrm{max}\big\{\utility_{\party_i}(\profile_{i,A}),\utility_{\party_i}(\profile_{i,\neg A})\big\}\;.
\end{equation}
By the definition of  $\strategy_{i,A}$, $\strategy_{i,\neg A}$, we have that

\begin{align*}
\Pr[\reward_{\party_i, \execution_{\profile}}-\cost_{\party_i, \execution_{\profile}}=x|A]&=\Pr[\reward_{\party_i, \execution_{\profile_{i,A}}}-\cost_{\party_i, \execution_{\profile_{i,A}}}=x]\\
\Pr[\reward_{\party_i, \execution_{\profile}}-\cost_{\party_i, \execution_{\profile}}=x|\neg A]&=\Pr[\reward_{\party_i, \execution_{\profile_{i,\neg A}}}-\cost_{\party_i, \execution_{\profile_{i,\neg A}}}=x]\\
\Pr\big[\sum_{\hat{\party}\in\partySet}\reward_{\hat{\party} ,\execution_{\profile}}=y|A\big]&=\Pr\big[\sum_{\hat{\party}\in\partySet}\reward_{\hat{\party}, \execution_{\profile_{i,A}}}=y\big]\\
\Pr\big[\sum_{\hat{\party}\in\partySet}\reward_{\hat{\party}, \execution_{\profile}}=y|\neg A\big]&=\Pr\big[\sum_{\hat{\party}\in\partySet}\reward_{\hat{\party}, \execution_{\profile_{i,\neg A}}}=y\big]
\end{align*}
By the above, we get that

\begin{equation}\label{eq:profile_profit-split_enumerator}
\begin{split}
&E[\reward_{\party_i, \execution_{\profile}}-\cost_{\party_i, \execution_{\profile}}]=\sum_x\Pr[\reward_{\party_i, \execution_{\profile}}-\cost_{\party_i, \execution_{\profile}}=x]\cdot x=\\
&=\Pr[A]\cdot\sum_x\Pr[\reward_{\party_i, \execution_{\profile}}-\cost_{\party_i, \execution_{\profile}}=x|A]\cdot x+\Pr[\neg A]\cdot\sum_x\Pr[\reward_{\party_i, \execution_{\profile}}-\cost_{\party_i, \execution_{\profile}}=x|\neg A]\cdot x=\\
&=\Pr[A]\cdot\sum_x\Pr[\reward_{\party_i, \execution_{\profile_{i,A}}}-\cost_{\party_i, \execution_{\profile_{i,A}}}=x]\cdot x+\Pr[\neg A]\cdot\sum_x\Pr[\reward_{\party_i, \execution_{\profile_{i,
\neg A}}}-\cost_{\party_i, \execution_{\profile_{i,
\neg A}}}=x]\cdot x=\\
&=\Pr[A]\cdot E[\reward_{\party_i, \execution_{\profile_{i,A}}}-\cost_{\party_i, \execution_{\profile_{i,A}}}]+(1-\Pr[A])\cdot E[\reward_{\party_i, \execution_{\profile_{i,\neg A}}}-\cost_{\party_i, \execution_{\profile_{i,\neg A}}}].
\end{split}
\end{equation}
and
\begin{equation}\label{eq:profile_profit-split_denominator}
\begin{split}
&E\big[\sum_{\hat{\party}\in\partySet}\reward_{\hat{\party}, \execution_{\profile}}\big]=\sum_y\Pr\big[\sum_{\hat{\party}\in\partySet}\reward_{\hat{\party}, \execution_{\profile}}=y\big]\cdot y=\\
&=\Pr[A]\cdot\sum_y\Pr\big[\sum_{\hat{\party}\in\partySet}\reward_{\hat{\party}, \execution_{\profile}}=y|A\big]\cdot y+\Pr[\neg A]\cdot\sum_y\Pr\big[\sum_{\hat{\party}\in\partySet}\reward_{\hat{\party}, \execution_{\profile}}=y|\neg A\big]\cdot y=\\
&=\Pr[A]\cdot\sum_y\Pr\big[\sum_{\hat{\party}\in\partySet}\reward_{\hat{\party}, \execution_{\profile_{i,A}}}=y\big]\cdot y+\Pr[\neg A]\cdot\sum_y\Pr\big[\sum_{\hat{\party}\in\partySet}\reward_{\hat{\party}, \execution_{\profile_{i,
\neg A}}}=y\big]\cdot y=\\
&=Pr[A]\cdot E\big[\sum_{\hat{\party}\in\partySet}\reward_{\hat{\party}, \execution_{\profile_{i,A}}}\big]+(1-\Pr[ A])\cdot E\big[\sum_{\hat{\party}\in\partySet}\reward_{\hat{\party}, \execution_{\profile_{i,\neg A}}}\big]
\end{split}
\end{equation}
For simplicity, we set $x:=\Pr[A]$, $a:=E[\reward_{\party_i, \execution_{\profile_{i,A}}}-\cost_{\party_i, \execution_{\profile_{i,A}}}]$, $b:=E\big[\sum_{\hat{\party}\in\partySet}\reward_{\hat{\party}, \execution_{\profile_{i,A}}}\big]$, $c:=E[\reward_{\party_i, \execution_{\profile_{i,\neg A}}}-\cost_{\party_i, \execution_{\profile_{i,\neg A}}}]$, and $d:=E\big[\sum_{\hat{\party}\in\partySet}\reward_{\hat{\party}, \execution_{\profile_{i,\neg A}}}\big]$. By Eq.~\eqref{eq:relative_profit},~\eqref{eq:profile_profit-split_enumerator}, \eqref{eq:profile_profit-split_denominator}, we have that

\begin{align*}
\utility_{\party_i}(\profile)&=\left\{
\begin{array}{ll}
\frac{x\cdot a+(1-x)\cdot c}{x\cdot b+(1-x)\cdot d},&\mbox{if }x\cdot b+(1-x)\cdot d>0\\
0,&\mbox{if }x\cdot b+(1-x)\cdot d=0
\end{array}\right.\\
\utility_{\party_i}(\profile_{i,A})&=\left\{
\begin{array}{ll}
\frac{a}{b},&\mbox{if }b>0\\
0,&\mbox{if } b=0
\end{array}\right.\\
\utility_{\party_i}(\profile_{i,\neg A})&=\left\{
\begin{array}{ll}
\frac{c}{d},&\mbox{if }d>0\\
0,&\mbox{if } d=0
\end{array}\right.
\end{align*}
We study all the following  possible cases:
\begin{itemize}
\item If $x\cdot b+(1-x)\cdot d=0$, then
\begin{itemize}
\item If $b=0$, then $0=\utility_{\party_i}(\profile)=\utility_{\party_i}(\profile_{i,A})$.
\item If $d=0$, then $0=\utility_{\party_i}(\profile)=\utility_{\party_i}(\profile_{i,\neg A})$.
\end{itemize}
\item If $x\cdot b+(1-x)\cdot d>0$ ,then
\begin{itemize}
\item If $b=0$, then $a\leq0$ and $(1-x)\cdot d>0$. Therefore,
\[\utility_{\party_i}(\profile)=\dfrac{x\cdot a+(1-x)\cdot c}{x\cdot b+(1-x)\cdot d}\leq\dfrac{(1-x)\cdot c}{(1-x)\cdot d}=\dfrac{c}{d}=\utility_{\party_i}(\profile_{i,\neg A}).\]
\item If $d=0$, then $c\leq0$ and $x\cdot b>0$. Therefore,
\[\utility_{\party_i}(\profile)=\dfrac{x\cdot a+(1-x)\cdot c}{x\cdot b+(1-x)\cdot d}\leq\dfrac{x\cdot a}{x\cdot b}=\dfrac{a}{b}=\utility_{\party_i}(\profile_{i, A}).\]
\item If $b>0$ and $d>0$,
\begin{itemize}
\item If $x=0$, then $\utility_{\party_i}(\profile)=\frac{c}{d}=\utility_{\party_i}(\profile_{i,\neg A})$.
\item If $x=1$, then $\utility_{\party_i}(\profile)=\frac{a}{b}=\utility_{\party_i}(\profile_{i, A})$.
\item If $x\in(0,1)$ and $\frac{a}{b}\geq\frac{c}{d}$, then
\begin{equation*}
\begin{split}
&(1-x)\cdot c\cdot b\leq(1-x)\cdot a\cdot d\Leftrightarrow x\cdot a\cdot b+(1-x)\cdot c\cdot b\leq x\cdot a\cdot b+(1-x)\cdot a\cdot d\Leftrightarrow\\
\Leftrightarrow&\dfrac{x\cdot a+(1-x)\cdot c}{x\cdot b+(1-x)\cdot d}\leq\dfrac{a}{b}\Leftrightarrow\utility_{\party_i}(\profile)\leq\utility_{\party_i}(\profile_{i,A}).
\end{split}
\end{equation*}
\item If $x\in(0,1)$ and $\frac{a}{b}<\frac{c}{d}$, then
\begin{equation*}
\begin{split}
&x\cdot c\cdot b>x\cdot a\cdot d\Leftrightarrow x\cdot c\cdot b+(1-x)\cdot c\cdot d>x\cdot a\cdot d+(1-x)\cdot c\cdot d\Leftrightarrow\\
\Leftrightarrow&\dfrac{x\cdot a+(1-x)\cdot c}{x\cdot b+(1-x)\cdot d}<\dfrac{c}{d}\Leftrightarrow\utility_{\party_i}(\profile)<\utility_{\party_i}(\profile_{i,\neg A}).
\end{split}
\end{equation*}
\end{itemize}
\end{itemize}
\end{itemize}

Therefore, for all possible cases Eq.~\eqref{eq:bounded_by_max} holds.
Hence, if $\utility_{\party_i}(\profile_{i,A})\geq\utility_{\party_i}(\profile)$ (resp. $\utility_{\party_i}(\profile_{i,\neg A})\geq\utility_{\party_i}(\profile)$), then we can ``keep'' $\strategy_{i,A}$ (resp. $\strategy_{i,\neg A}$) that contains one less check than $\profile$ and ``discard'' the latter. By continuing this process iteratively for all parties, we end up in a strategy profile $\profile'$ that contains no checks of the aforementioned form, i.e. $\profile'$ is  $\infractionPredicate_\mathrm{conf}$-agnostic. By construction, it holds that $\utility_{\party_i}(\profile')\geq\utility_{\party_i}(\profile)$. \hfill  $\dashv$\\

Given Claims~\ref{claim:non_conf2} and~\ref{claim:static2}, we prove the $\infractionPredicate_\mathrm{conf}$-compliance of $\proto$. In particular, we show that for every $\epsilon\geq0$, if a strategy profile $\profile$ is $\epsilon$-reachable from $\profile_\proto$, then it is $\infractionPredicate_\mathrm{conf}$-compliant.

For the sake of contradiction, assume that  $\profile$ is not $\infractionPredicate_\mathrm{conf}$-compliant. The proof is by induction on the length $\ell$ of the shortest path from $\profile_\proto$ to $\profile$.
\begin{itemize}
\item \emph{Basis}: $\ell=1$. If $\profile$ is $\infractionPredicate_\mathrm{conf}$-agnostic, then by Claim~\ref{claim:non_conf2}, $\profile$ cannot set a best response for the party $\party_i$ that unilaterally deviates from $\profile_\proto$. Besides, if $\profile$ is not $\infractionPredicate_\mathrm{conf}$-agnostic, then by Claim~\ref{claim:static2}, there is a strategy profile $\profile'$ that is $\infractionPredicate_\mathrm{conf}$-agnostic and for which it holds that $\utility_{\party_i}(\profile')\geq\utility_{\party_i}(\profile)$. Since by Claim~\ref{claim:non_conf2} $\profile'$ cannot set a best response for the party $\party_i$, neither $\profile$ can set a best response for the party $\party_i$. Therefore, $\profile$ is not $\epsilon$-directly reachable from $\profile_\proto$.
\item \emph{Induction step}: Assume that the statement holds for every $\epsilon$-reachable strategy profile with shortest path of length $\ell$. Let $\profile_\proto\rightarrow\profile_1\rightarrow\cdots\rightarrow\profile_\ell\rightarrow\profile$ be the shortest path of length $\ell+1$ from $\profile_\proto$ to $\profile$. Then, $\profile_1,\ldots,\profile_\ell$ are all $\infractionPredicate_\mathrm{conf}$-compliant. By following the same steps as in Basis, we conclude that, either $\infractionPredicate_\mathrm{conf}$-agnostic or not, $\profile$ cannot set a best response for the party that unilaterally deviates from $\profile_\ell$, thus $\profile$ is not $\epsilon$-directly reachable from $\profile_\proto$.
\end{itemize}
%
%
Subsequently, we will prove the non $\infractionPredicate_\mathrm{self}$-compliance of $\proto$. We denote by $\profile_i$ the strategy profile where $\party_i$ unilaterally deviates from $\profile_\proto$ by following $\strategy_\mathrm{self}$. Let $D_{\party_i,\execution_{\profile_i}}$ be the number of blocks that are discarded due to the selfish signing of $\party_i$ \wrt $\profile_i$. By the description of $\strategy_\mathrm{self}$ (cf. Algorithm~\ref{alg:strategy_self}), we have that
(i) $E[\msgOutputSet_{\party_i, \execution_{\profile_i}}]=E[\msgOutputSet_{\party, \execution_{\profile_\proto}}]$, (ii) $E[Q_{\party_i, \execution_{\profile_i}}]=E[Q_{\party_i, \execution_{\profile_\proto}}]$,  and
(iii) $E\big[\sum_{\hat{\party}\in\partySet}\msgOutputSet_{\hat{\party}, \execution_{\profile_i}}\big]=E\big[\sum_{\hat{\party}\in\partySet}\msgOutputSet_{\hat{\party}, \execution_{\profile_\proto}}-D_{\party_i,\execution_{\profile_i}}\big]$. Thus, by Eq.~\eqref{eq:PoS_relative_profit} and~\eqref{eq:honest_relative_profit}, we get that

\begin{equation}\label{eq:profile_self}
\begin{split}
\utility_{\party_i}(\profile_i)&=\dfrac{\miningpower_{\party_i}\cdot N\cdot\reward - \miningpower_{\party_i}\cdot N\cdot\cost}{N\cdot\reward-E[D_{\party_i,\execution_{\profile_i}}]\cdot\reward}=\miningpower_{\party_i}\cdot\dfrac{\reward-\cost}{\reward}\cdot\dfrac{1}{1-\frac{E[D_{\party_i,\execution_{\profile_i}}]}{N}}=\\
&=\utility_{\party_i}(\profile_\proto)\cdot\dfrac{1}{1-\frac{E[D_{\party_i,\execution_{\profile_i}}]}{N}}.
\end{split}
\end{equation}

We will lower bound $\utility_{\party_i}(\profile_i)$ by showing a lower bound for $E[D_{\party_i,\execution_{\profile_i}}]$. To achieve this, we devise a selfish signing strategy, $\strategy_{d\leq3}$, that although less effective than $\strategy_\mathrm{self}$, it captures a significant part of all possible selfish signing attempts in an execution.

The effectiveness of $\strategy_{d\leq3}$ relies on the fact that the probability the conditions allow for performing selfish signing of depth $d$ drops exponentially in $d$. Namely, if we consider the string $\mathsf{schedule}_{\party_i}\in\{0,1\}^N$ as in Algorithm~\ref{alg:strategy_self}, the probability that for round $\slot$, a substring $\underbrace{0\cdots0}_{\text{$d$ times}}1\underbrace{1\cdots1}_{\text{$d$ times}}$ appears (indicating conditions that allow for performing selfish signing of depth $d$), is $(1-\miningpower_{\party_i})^d\cdot\miningpower_{\party_i}^{d+1}$.
Therefore, by searching $\mathsf{schedule}_{\party_i}$ for substrings $011,00111,0001111$ that correspond to selfish signing at depth up to $3$, we can capture a good portion of all available cases that selfish signing is possible according to $\mathsf{schedule}_{\party_i}$.

Given the above, $\party_i$ follows $\strategy_{d\leq3}$ by executing the following steps:
\begin{enumerate}
\item She fragments $\mathsf{schedule}_{\party_i}$ into $7$-bit substrings, where for simplicity we assume that the length, $N$, of $\mathsf{schedule}_{\party_i}$ is a multiple of $7$.
\item For each $7$-bit substring denoted by $b_1b_2b_3b_4b_5b_6b_7\in\{0,1\}^7$:
\begin{enumerate}
\item She sequentially checks if $b_1b_2b_3=011$, $b_2b_3b_4=011$, $b_3b_4b_5=011$, $b_4b_5b_6=011$, $b_5b_6b_7=011$. When some of these checks is successful, she performs selfish signing at depth $1$ during the corresponding slot.
\item If all the above five checks fail, she sequentially checks if $b_1b_2b_3b_4b_5=00111$, $b_2b_3b_4b_5b_6=00111$, $b_3b_4b_5b_6b_7=00111$. When some of these checks is successful, she performs selfish signing at depth $2$ during the corresponding slot.
\item If all the above three checks fail, she sequentially checks if $b_1b_2b_3b_4b_5b_6b_7=0001111$. If the check is successful, she performs selfish signing at depth $3$ during the corresponding slot.
\item After all checks are completed, she proceeds similarly with the next $7$-bit substring.
\end{enumerate}
\end{enumerate}

Next, we compute the probabilities of the corresponding events. We get that the following hold:

\begin{itemize}
\item $\Pr[b_1b_2b_3=011]=\Pr[b_2b_3b_4=011]=\Pr[b_3b_4b_5=011]=\Pr[b_4b_5b_6=011]=\Pr[b_5b_6b_7=011]=(1-\miningpower_{\party_i})\cdot\miningpower_{\party_i}^2$.
\item $\Pr[b_1b_2b_3b_4b_5=00111]=\Pr[b_2b_3b_4b_5b_6=00111]=\Pr[b_3b_4b_5b_6b_7=00111]=(1-\miningpower_{\party_i})^2\cdot\miningpower_{\party_i}^3$.
\item $\Pr[b_1b_2b_3b_4b_5b_6b_7=0001111]=(1-\miningpower_{\party_i})^3\cdot\miningpower_{\party_i}^4$.
\item $\Pr[(b_1b_2b_3=011)\land(b_4b_5b_6=011)]=\Pr[(b_1b_2b_3=011)\land(b_5b_6b_7=011)]=\Pr[(b_2b_3b_4=011)\land(b_5b_6b_7=011)]=(1-\miningpower_{\party_i})^2\cdot\miningpower_{\party_i}^4$, while the probability of any other conjuction of events is $0$.
\end{itemize}

Let $\profile_{i,d\leq3}$ be the strategy profile where $\party_i$ unilaterally deviates from $\profile_\proto$ by following $\strategy_{d\leq3}$.

Let $D_{\party_i,b_{1-7}}$ be the number of blocks that are discarded during the $7$ consecutive slots that correspond to a single $7$-bit substring due to the selfish signing of $\party_i$ \wrt $\profile_{i,d\leq3}$. By the above, we have that

\begin{equation*}
\begin{split}
E[D_{\party_i,b_{1-7}}]=&\sum_{j=1}^3j\cdot\Pr[\mbox{exactly }j \mbox{ blocks are discarded in }b_1b_2b_3b_4b_5b_6b_7]=\\
=&1\cdot\big(5\cdot(1-\miningpower_{\party_i})\cdot\miningpower_{\party_i}^2-3\cdot(1-\miningpower_{\party_i})^2\cdot\miningpower_{\party_i}^4\big)+\\
&2\cdot\big(3\cdot(1-\miningpower_{\party_i})^2\cdot\miningpower_{\party_i}^3+3\cdot(1-\miningpower_{\party_i})^2\cdot\miningpower_{\party_i}^4\big)+\\
&3\cdot(1-\miningpower_{\party_i})^3\cdot\miningpower_{\party_i}^4=\\
=&5\cdot(1-\miningpower_{\party_i})\cdot\miningpower_{\party_i}^2+6\cdot(1-\miningpower_{\party_i})^2\cdot\miningpower_{\party_i}^3+3\cdot(1-\miningpower_{\party_i})^2\cdot\miningpower_{\party_i}^4+3\cdot(1-\miningpower_{\party_i})^3\cdot\miningpower_{\party_i}^4\;.
\end{split}
\end{equation*}

Let $\delta(\miningpower_{\party_i}):=E[D_{\party_i,b_{1-7}}]$
Let $D_{\party_i,\execution_{\profile_{i,d\leq3}}}$ be the total number of blocks that are discarded due to the selfish signing of $\party_i$ \wrt $\profile_{i,d\leq3}$ during the execution. Since the examination of these events is independent across the $7$-bit substrings, we have that

\begin{equation}\label{eq:d_leq_3_total}
\begin{split}
E[D_{\party_i,\execution_{\profile_{i,d\leq3}}}]=\dfrac{N}{7}\cdot\delta(\miningpower_{\party_i})\;.
\end{split}
\end{equation}

By the description of $\strategy_\mathrm{self}$ and $\strategy_{d\leq3}$, it is straightforward that $E[D_{\party_i,\execution_{\profile_i}}]>E[D_{\party_i,\execution_{\profile_{i,d\leq3}}}]$. Thus, by Eq.~\eqref{eq:honest_relative_profit},~\eqref{eq:profile_self} and~\eqref{eq:d_leq_3_total}, we have that

\begin{equation*}
\begin{split}
\utility_{\party_i}(\profile_i)&>\utility_{\party_i}(\profile_\proto)\cdot\dfrac{1}{1-\frac{\frac{N}{7}\cdot\delta(\miningpower_{\party_i})}{N}}=\utility_{\party_i}(\profile_\proto)\cdot\Big(1+\dfrac{1}{\frac{7}{\delta(\miningpower_{\party_i})}-1}\Big)=\\
&=\utility_{\party_i}(\profile_\proto)+ \dfrac{\miningpower_{\party_i}}{\frac{7}{\delta(\miningpower_{\party_i})-1}}\cdot\dfrac{\reward-\cost}{\reward}\;.
\end{split}
\end{equation*}

The function $\delta(\miningpower)=5\cdot(1-\miningpower)\cdot\miningpower^2+6\cdot(1-\miningpower)^2\cdot\miningpower^3+3\cdot(1-\miningpower)^2\cdot\miningpower^4+3\cdot(1-\miningpower)^3\cdot\miningpower^4$ is increasing on the interval $(0,0.64469]$, therefore the function
$\frac{\miningpower}{\frac{7}{\delta(\miningpower)}-1}\cdot\frac{\reward-\cost}{\reward}$ is also increasing on the same interval.
So, by setting $\epsilon_\mathrm{max}:=\frac{\miningpower_{\party_\mathrm{max}}}{\frac{7}{\delta(\miningpower_{\party_\mathrm{max}})}-1}\cdot\frac{\reward-\cost}{\reward}$, we have that for every $\epsilon\leq\epsilon_\mathrm{max}$, it holds that
\[\utility_{\party_\mathrm{max}}(\profile_\mathrm{max})>\utility_{\party_\mathrm{max}}(\profile_\proto)+\epsilon\;,\]
where in $\profile_\mathrm{max}$ is the strategy profile that the party $\party_\mathrm{max}$ unilaterally deviates from $\profile_\proto$ by following $\strategy_\mathrm{self}$.

It remains to show that the best response for some party $\party_i$ that unilaterally deviates from $\profile_\proto$ is not $\infractionPredicate_\mathrm{self}$-compliant. In fact, we will show something stronger; for every strategy profile $\profile$ that is a $\infractionPredicate_\mathrm{self}$-compliant unilateral deviation of $\party_i$, it holds that $\utility_{\party_i}(\profile)\leq\utility_{\party_i}(\profile_\proto)$.

By the $(0,\infractionPredicate_\mathrm{conf})$-compliance of $\profile_\proto$, we have that it is in $\party_i$'s interest to make only one query to $\oracle_\proto$, if she decides to participate at a given slot. Thus, it suffices to focus on strategy profiles that are not $\infractionPredicate_\mathrm{abs}$-compliant. So, assume that there are execution traces \wrt $\profile$ where $\party_i$ abstains at certain slots. However, each time $\party_i$ decides to abstain, she loses an amount $\reward-\cost$ on her overall profit while, since we assumed that $\profile$ is     $\infractionPredicate_\mathrm{self}$-compliant, the other (honest) parties' rewards remain unaffected. Therefore, by abstaining, $\party_i$ only reduces the  contribution to the her own  utility.\footnote{Note that this does not imply the $\infractionPredicate_\mathrm{abs}$-compliance of $\profile_\proto$, only that it is not in the party's interest to unilaterally deviate from $\profile_\proto$ by just abstaining.}. Formally, let $A_{\party_i,\execution_{\profile}}$ be the number of times that $\party_i$ decides to abstain. Since $E[A_{\party_i,\execution_{\profile}}]>0$, we have that
\begin{equation*}
\begin{split}
\utility_{\party_i}(\profile)&=\dfrac{(\miningpower_{\party_i}\cdot N-E[A_{\party_i,\execution_{\profile}}])\cdot\reward - (\miningpower_{\party_i}\cdot N-E[A_{\party_i,\execution_{\profile}}])\cdot\cost}{(N-E[A_{\party_i,\execution_{\profile}}])\cdot\reward}=\\
&=\dfrac{\miningpower_{\party_i}\cdot N-E[A_{\party_i,\execution_{\profile}}]}{N-E[A_{\party_i,\execution_{\profile}}]}\cdot\dfrac{\reward-\cost}{\reward}<\dfrac{\miningpower_{\party_i}\cdot (N-E[A_{\party_i,\execution_{\profile}}])}{N-E[A_{\party_i,\execution_{\profile}}]}\cdot\dfrac{\reward-\cost}{\reward}=\\
&=\miningpower_{\party_i}\cdot\dfrac{\reward-\cost}{\reward}=\utility_{\party_i}(\profile_\proto)\;.
\end{split}
\end{equation*}

Finally, given that (i) for every $\epsilon\leq\epsilon_\mathrm{max}$, it holds that $\utility_{\party_\mathrm{max}}(\profile_\mathrm{max})>\utility_{\party_\mathrm{max}}(\profile_\proto)+\epsilon$ and (ii) the best response for any party is not $\infractionPredicate_\mathrm{self}$-compliant, we conclude that for every $\epsilon\leq\epsilon_\mathrm{max}$, the protocol $\proto$ is not $(\epsilon,\infractionPredicate_\mathrm{self})$-compliant \wrt Relative Profit.

\end{proof}

\paragraph{Remark.}
$\strategy_\mathrm{self}$, although profitable, as shown in
Theorem~\ref{thm:PoS_relative_profit}, is not optimal. Indeed, one may
consider incremental improvements by combining the selfish signing behavior
with other infractions, such as strategically abstaining during some
slots. On that account, obtaining an optimal attack that exploits selfish
signing and identifying for which infractions such attack is non-compliant
is an interesting future research direction.

\section{Externalities}\label{sec:externalities}

In practice, blockchains coexist with other systems, which may affect the
participants' behavior.
This section enhances
our analysis with parameters external to
the ledger. We introduce an exchange rate,
to account rewards in the same unit as costs, and
analyze how it should behave to ensure compliance, assuming infractions yield
an external utility, and finally take penalties into account.

\subsection{Utility}

In distributed ledger systems, rewards are denominated in the ledger's native
currency, but cost is typically denominated in fiat. Therefore, we introduce an
\emph{exchange rate}, between the ledger's native currency and USD, to
denominate the rewards and cost in the same unit of account and precisely
estimate a party's utility.
The exchange rate $\exchangeRate_{\execution}$ is a random variable,
parameterized by a strategy profile $\profile$. For a trace
$\executionTrace_{\profile}$ under $\profile$, the exchange rate takes a
non-negative real value. The exchange rate is applied once, at the end of the
execution. Intuitively, this implies that a party eventually sells their
rewards at the end of the execution. Therefore, its utility depends on the
accumulated rewards, during the execution, and the exchange rate at the end.

The infraction predicate expresses a deviant behavior that parties may exhibit.
So far, we considered distributed protocols in a standalone fashion, analyzing
whether they incentivize parties to avoid infractions. In reality, a ledger
exists alongside other systems, and a party's utility may depend on parameters
external to the distributed ledger. For instance, double spending against
Bitcoin is a common hazard, which does not increase an attacker's \emph{Bitcoin
rewards}, but awards them external rewards, \eg goods that are purchased with
the double-spent coins.

The external -- to the ledger -- reward is modeled as a random variable
$\utilityBoost_{\party, \execution_{\profile}}$, which takes non-negative integer
values. Similarly to the rewards' random variable, it is
parameterized by a
party $\party$ and a strategy profile $\profile$.  The infraction utility is
applied once when computing a party's utility and has the property that, for
every trace $\executionTrace$ during which a party $\party$ performs no
infraction, it holds that $\utilityBoost_{\party, \executionTrace} = 0$, \ie
a party receives these external rewards only by performing an
infraction.

We define a new utility function $\utility$, which also takes two forms,
\emph{Reward} and \emph{Profit}. For the former, $\utility$
applies the exchange rate on the protocol rewards and adds the external. For \emph{Profit}, it also subtracts the cost. Definition~\ref{def:utility-external} defines the utility under
externalities. For ease of notation, we set the following:
\begin{itemize}
    \item $\rewardVal_{\party, \profile} = E[\reward_{\party, \execution_{\env, \adversary, \profile}}]$;
    \item $\exchangeRateVal_{\profile} = E[\exchangeRate_{\execution_{\env, \adversary, \profile}}]$;
    \item $\utilityBoostVal_{\party, \profile} = E[\utilityBoost_{\party, \execution_{\env, \adversary, \profile}}]$;
    \item $\costVal_{\party, \profile} = E[\cost_{\party, \execution_{\env, \adversary, \profile}}]$.
\end{itemize}
As in Section~\ref{sec:blockchain-utility}, when computing the utility, the
environment and the router are fixed.

\begin{definition}\label{def:utility-external}
    Let:
    \begin{inparaenum}[i)]
        \item $\profile$ be a strategy profile;
        \item $\execution_{\profile}$ be an execution under $\profile$;
        \item $\exchangeRateVal_{\profile}$ be the (expected) exchange rate of $\execution_{\profile}$;
        \item $\utilityBoostVal_{\party, \profile}$ be the (expected) external rewards of $\party$ under $\profile$.
    \end{inparaenum}
    We define two types of utility $\utility_{\party}$ of a party $\party$ for
    $\profile$ under externalities:
    \begin{enumerate}
        \item \emph{Reward}: $\utility_{\party}(\profile) = \rewardVal_{\party, \profile} \cdot \exchangeRateVal_{\profile} + \utilityBoostVal_{\party, \profile}$;
        \item \emph{Profit}: $\utility_{\party}(\profile) = \rewardVal_{\party, \profile} \cdot \exchangeRateVal_{\profile} + \utilityBoostVal_{\party, \profile} - \costVal_{\party, \profile}$.
    \end{enumerate}
\end{definition}

\subsection{Compliance}

To evaluate compliance under externalities, we will
find a relation between the assets' price and external, infraction-based
rewards, \st the former counters the latter, hence parties are
incentivized to remain compliant.
Specifically, it suffices to show that
the exchange rate reduction counterbalances the external rewards.
More formally,
Theorem~\ref{thm:external-ouroboros} analyzes SL-PoS under a synchronous
network and externalities; similar statements can be made for the positive
results of Sections~\ref{sec:universal} and~\ref{subsec:bitcoin}.

\begin{theorem}\label{thm:external-ouroboros}
    Assume
    \begin{inparaenum}[i)]
        \item a synchronous router $\adversary$ (cf. Section~\ref{sec:preliminaries}),
        \item the conflicting predicate $\infractionPredicate_\mathrm{conf}$, and
        \item that $\forall \party \in \partySet: \miningpower_{\party} <$~$\frac{1}{2}$.
    \end{inparaenum}
    Also let:
    \begin{inparaenum}[i)]
        \item $\strategySet_{-\infractionPredicate_\mathrm{conf}}$: the set of all non $\infractionPredicate_\mathrm{conf}$-compliant strategies;
        \item $\exchangeRateVal_{\profile_\proto}$: the (expected) exchange rate under $\execution_{\profile_\proto}$;
        \item $\exchangeRateVal_{\profile_{\strategy_{\party}}}$: the (expected) exchange rate when only $\party$ employs some non $\infractionPredicate_\mathrm{conf}$-compliant strategy $\strategy_{\party}$;
        \item $\utilityBoostVal_{\party, \profile_{\strategy_{\party}}}$: the external utility that $\strategy_{\party}$ yields for $\party$.
    \end{inparaenum}

    SL-PoS with block-proportional rewards (cf.
    Definition~\ref{def:proportional-rewards}, for fixed block reward
    $\reward$) under the aforementioned externalities is not $(\epsilon,
    \infractionPredicate_\mathrm{conf})$-compliant (cf.
    Definition~\ref{def:compliant}) \wrt utility \emph{Reward} (cf.
    Definition~\ref{def:utility}) and, if $\reward > \cost$, it is also not
    $(\epsilon, \infractionPredicate_\mathrm{conf})$-compliant \wrt utility
    \emph{Profit}, in both cases under $\adversary$ and if and only if
    $\epsilon < \mathsf{max}\{  \underset{\party \in \partySet}{\mathsf{max}}\{ \underset{\strategy_{\party} \in \strategySet_{-\infractionPredicate_\mathrm{conf}}}{\mathsf{max}}\{ \rewardVal_{\party, \profile_{\proto}} \cdot (\exchangeRateVal_{\profile_{\strategy_{\party}}} - \exchangeRateVal_{\profile_\proto}) + \utilityBoostVal_{\party, \profile_{\strategy_{\party}}} \} \}, 0 \}$.
\end{theorem}

    \begin{proof}
    Following the same reasoning as
    Theorem~\ref{thm:compliant-ouroboros-synchronous}, if a party $\party$
    deviates by only producing conflicting messages, but does not abstain, its
    expected rewards are the same as following the protocol; specifically, due
    to network synchronicity, after every round when $\party$ is leader, every
    other party adopts one of the blocks produced by $\party$ (although
    possibly not everybody adopts the same block), and, since all these blocks
    are part of the (equally-long) longest chain (at that point), eventually
    one of these blocks will be output in the chain of the observer.
    Consequently, it holds that $\rewardVal_{\party,
    \profile_{\strategy_{\party}}} = \rewardVal_{\party, \profile_{\proto}}$.

    Second, the maximum additional utility that a party $\party$ may receive by
    deviating from the honest protocol via producing conflicting blocks is:
    $\underset{\strategy_{\party} \in \strategySet_{-\infractionPredicate_{conf}}}{\mathsf{max}}\{ \rewardVal_{\party, \profile_{\proto}} \cdot (\exchangeRateVal_{\profile_{\strategy_{\party}}} - \exchangeRateVal_{\profile_\proto}) + \utilityBoostVal_{\party, \profile_{\strategy_{\party}}} \}$.
    Therefore, if for at least one party this value is non-negligible,
    $\epsilon$ is not small enough and so the protocol is not compliant.
\end{proof}

The previous sections offer non-compliance, negative results in PoS systems where
\begin{inparaenum}[(a)]
    \item resource-proportional rewards are employed and
    \item a party is incentivized to produce multiple
        conflicting messages, \ie under a lossy network or multiple leaders per
        slot.
\end{inparaenum}

Regarding (a), Section~\ref{sec:universal} shows that resource-proportional rewards ensure
compliance under utility \emph{Reward}, but non-compliance regarding profit.
Specifically, assuming a minimal participation cost $\cost^\bot_{\party}$, we
showed that, if $\party$ abstains, they incur zero cost without any reward
reduction. To explore compliance of resource-proportional rewards under externalities, we
consider two strategy profiles $\profile_\proto,
\profile_{\strategy_{\party}}$, as before. Notably, $\strategy_{\party}$ is the
abstaining strategy which, as shown in Section~\ref{sec:universal}, maximizes
utility in the standalone setting. For the two profiles, the profit for
$\party$ becomes
$\rewardVal_{\party, \profile_\proto} \cdot \exchangeRateVal_{\profile_\proto} - \costVal_{\party, \profile_\proto}$ and
$\rewardVal_{\party, \profile_{\strategy_{\party}}} \cdot \exchangeRateVal_{\profile_{\strategy_{\party}}} + \utilityBoostVal_{\party, \profile_{\strategy_{\party}}}$ respectively.
Again, in both cases the party's rewards are equal.
Therefore, since it
holds that $\cost^\bot_{\party} \leq \costVal_{\party, \profile_\proto}$,
$\party$ is incentivized to be $(\epsilon, \infractionPredicate_\mathrm{conf})$-compliant (for some $\epsilon$) if:
\begin{align}
    \rewardVal_{\party, \profile_{\strategy_{\party}}} \cdot \exchangeRateVal_{\profile_{\strategy_{\party}}} + \utilityBoostVal_{\party, \profile_{\strategy_{\party}}} \leq \rewardVal_{\party, \profile} \cdot \exchangeRateVal_{\profile_\proto} - \costVal_{\party, \profile_\proto} + \epsilon \Rightarrow 
    \cost^\bot_{\party} + \utilityBoostVal_{\party, \profile_{\strategy_{\party}}} \leq \rewardVal_{\party, \profile_\proto} \cdot (\exchangeRateVal_{\profile_\proto} - \exchangeRateVal_{\profile_{\strategy_{\party}}}) + \epsilon \nonumber
\end{align}
If the abstaining strategy yields no external rewards, as is typically the
case, $\utilityBoostVal_{\party, \profile_{\strategy_{\party}}} = 0$, so the
exchange rate needs to only counterbalance the minimal participation cost.

Regarding (b), we consider single-leader PoS under a lossy network, since the
analysis is similar for multi-leader PoS. We again consider two strategy
profiles $\profile_\proto, \profile_{\strategy_{\party}}$ as above. Now, under
$\profile_{\strategy_{\party}}$, $\party$ produces $k$ blocks during each slot for which
it is leader, to increase the probability that at least one of them is output
in the observer's final chain. Also, for simplicity, we set
$\utilityBoostVal_{\party, \profile_{\strategy_{\party}}} = 0$. These PoS
systems become $(\epsilon', \infractionPredicate_\mathrm{conf})$-compliant (for
$\epsilon' = \frac{\epsilon}{(1 - \networkLossProb^k \cdot (1 - \networkLossProb)^2) \cdot \reward_{max}}$) if
    $\exchangeRateVal_{\profile_{\strategy_{\party}}} \leq \frac{1 - \networkLossProb \cdot (1 - \networkLossProb)^2}{1 - \networkLossProb^k \cdot (1 - \networkLossProb)^2} \cdot \exchangeRateVal_{\profile_\proto} + \epsilon'$
where $\reward_{max} = \reward \cdot \sum_{i \in [1, \epoch]} \epochLength \cdot \miningpower_{\party, i}$ and $\networkLossProb, \reward, \epochLength$ are as in Subsection~\ref{subsec:multi-leader-pos}.

\subsection{Attacks and Market Response}\label{sec:attacks}

To estimate the exchange rate's behavior vis-à-vis external infraction rewards,
we turn to historical data from the cryptocurrency market. Although no
infractions of the type considered in this work have been observed in
deployed PoS systems, we extrapolate data from similar attacks against PoW
cryptocurrencies (Table~\ref{tab:attacks}).

\begin{table*}[ht]
    \centering \def\arraystretch{1.5}

    \begin{center}
      \footnotesize
        \begin{tabular}{|c|c|c|c|c|c|}
            \hline

              System
            & Date
            & \begin{tabular}[c]{@{}c@{}} External Utility \end{tabular}
            & \begin{tabular}[c]{@{}c@{}} Rewards \end{tabular}
            & \begin{tabular}[c]{@{}c@{}} Reward Difference \end{tabular}
            & \begin{tabular}[c]{@{}c@{}} Attack Hash Rate \% \end{tabular} \\
            \hline

            \multirow{4}{*}{\begin{tabular}[c]{@{}c@{}} Ethereum \\ Classic \end{tabular}}   & 5/1/19 \cite{etc-attack-2}        & \$$1.1$M     & \$$12.410$       & \$$-2,646$       & $0.48026$  \\
                                                                                             & 1/8/20 \cite{etc-attack}          & \$$5.6$M     & \$$84,059$      & \$$-11,806$     & $0.4913$  \\
                                                                                             & 6/8/20 \cite{etc-attack-3}        & \$$1.68$M    & \$$91,715$      & \$$-5,761$      & $0.4913$  \\
            \hline

            Horizen                             & 8/6/18 \cite{zencash-attack}      & \$$550,000$   & \$$5,756$      & \$$-752$  & $0.461373$ \\
            \hline

            \multirow{1}{*}{Vertcoin}           & 2/12/18 \cite{vertcoin-attack}     & \$$100,000$   & \$$3,978$       & \$$-879$  & $0.487124$ \\
            \hline

            \multirow{2}{*}{\begin{tabular}[c]{@{}c@{}} Bitcoin Gold \end{tabular}}       & 16/5/18 \cite{btg-attack-2}       & \$$17.5$M     & \$$11,447$     & \$$-1,404$    & $0.441631$ \\
                                                & 23/1/20 \cite{btg-attack}         & \$$72,000$    & \$$4,247$      & \$$814$   & $0.43991$ \\
            \hline

            Feathercoin                         & 1/6/13 \cite{ftc-attack}         & \$$63,800$    & \$$1,203$      & \$$-95.73$   & $0.48283$ \\
            \hline

        \end{tabular}
      \normalsize
    \end{center}
    \caption{
        Double spending attacks and the market's response to them.
        External utility is estimated as the reward from double-spent
        transactions. To compute the reward difference, we multiply the rewards
        from reorganized blocks with the exchange rate difference, \ie the
        asset's price $5$ days after the attack minus the expected price, if an
        attack had not occurred (following Bitcoin's price in the same period).
    }
    \label{tab:attacks}
\end{table*}

In the considered attacks, the perpetrator $\attacker$ performed double
spending. Specifically, $\attacker$ created a fork and two conflicting
transactions, each published on the two chains of the fork, the
main and the adversarial chain. The main chain's transaction is redeemed for
external rewards, \eg a payment in USD, while the adversarial chain's
transaction transfers the assets between two accounts of $\attacker$.
Therefore, $\attacker$ both receives external rewards and retains its
cryptocurrency rewards.

The adversarial chain contains a number of blocks created by $\attacker$.
After this chain becomes longest and is adopted by the network,
$\attacker$ sells its block rewards for USD.  To evaluate the exchange rate at
this point, we set the period between the launch of the attack and the
(presumable) selling of the block rewards to $5$ days. This value depends on
various parameters. For instance, in Bitcoin, the rewards for a block $\block$
can be redeemed after a ``coinbase maturity'' period of $100$ confirmations,
\ie after at least $100$ blocks have been mined on top of $\block$ (equiv. $17$
hours).\footnote{A Bitcoin
block is created on expectation every $10$ minutes.}
Furthermore, transactions are typically not finalized immediately; for
instance, most parties finalize a Bitcoin transaction after $6$ confirmations
and an Ethereum transaction after $240$ confirmations (equiv. approximately $1$
hour).
Usually this restraint is tightened~\cite{etc-attack-4} after an
attack is revealed.

To estimate the difference in rewards that an infraction effects, we use
cryptocurrency prices from
Coinmarketcap.\footnote{\url{https://coinmarketcap.com/}}  First, we obtain the
price $P_\mathbf{C}$ of each cryptocurrency $\mathbf{C}$ $5$ days after the attack. Second, we
compute the percentage difference $p_{BTC}$ of Bitcoin's price, between the end
and the beginning of the $5$ day period. The value $P_\mathbf{C} \cdot p_{BTC}$
expresses the \emph{expected} price of the cryptocurrency, assuming no attack
had occurred.\footnote{Historically, the prices of Bitcoin and alternative
cryptocurrencies are strongly correlated~\cite{btc-price-correlation}.} Next,
we find the number $b$ of blocks created in the attack and the reward
$\reward$ per block. Thus, the reward difference is $P_\mathbf{C}
\cdot p_{BTC} \cdot b \cdot \reward$.

As shown in~\cite{CCS:GazKiaRus20,CCS:DKTTVWZ20}, this attack is
optimal. Therefore, using the computations in~\cite{nakamoto2008bitcoin} and
the reorganized blocks during each attack, we approximate the
percentage of power needed so that the attack's success
probability is at least $0.5$.

\subsection{Penalties}

Historically (cf. Subsection~\ref{sec:attacks}), attacks
are profitable, so the market's response is typically insufficient to
incentivize compliance. Interestingly, in many occasions the external utility
was so high that, even if the exchange rate became $0$, it would exceed the
amount of lost rewards. Therefore, an additional form of utility reduction is
necessary to prevent any specific infraction that is essential to mount attacks
similar to those presented in Subsection~\ref{sec:attacks}.
In many PoS systems, like
Casper~\cite{buterin2017casper,casper-incentives},
Gasper~\cite{buterin2020combining}, and Tezos~\cite{tezos-pos}, a form of utility
reduction has been implemented in  the
form of penalties. In effect, each party $\party$ is required to deposit an
amount of assets $\deposit_{\party}$, which it forfeits if it violates a
well-defined condition.

Under penalties, $\party$'s reward is as follows. Consider profiles
$\profile_\proto, \profile_{\strategy_{\party}}$ as before. With
$\profile_\proto$, $\party$ receives $\rewardVal_{\party, \profile_\proto}$ and
retains its deposit $\deposit_{\party}$, both exchanged at rate
$\exchangeRateVal_{\profile_\proto}$. With $\profile_{\strategy_{\party}}$,
$\party$ forfeits its rewards and deposit, but receives external utility
$\utilityBoostVal_{\party, \profile_{\strategy_{\party}}}$.  Thus, under
penalties a party is incentivized to be compliant if the deposit and
rewards are larger than the external utility.
Note that Theorem~\ref{thm:penalties-ouroboros}'s $\epsilon$ bound is tighter
than that of Theorem~\ref{thm:external-ouroboros}, so penalties can make
infractions less appealing.

\begin{theorem}\label{thm:penalties-ouroboros}
    Assume
    \begin{inparaenum}[i)]
        \item a synchronous router $\adversary$ (cf. Section~\ref{sec:preliminaries}),
        \item the conflicting predicate $\infractionPredicate_\mathrm{conf}$, and
        \item that $\forall \party \in \partySet: \miningpower_{\party} <$~$\frac{1}{2}$.
    \end{inparaenum}
    Also let:
    \begin{inparaenum}[i)]
        \item $\strategySet_{-\infractionPredicate_\mathrm{conf}}$: the set of all non-compliant strategies;
        \item $\exchangeRateVal_{\profile_\proto}$: the (expected) exchange rate under $\profile_\proto$;
        \item $\exchangeRateVal_{\profile_{\strategy_{\party}}}$: the (expected) exchange rate when only $\party$ employs some non-compliant conflicting strategy $\strategy_{\party}$;
        \item $\utilityBoostVal_{\party, \profile_{\strategy_{\party}}}$: the external utility that $\strategy_{\party}$ yields for $\party$;
        \item and $\rewardVal_{\party, \profile} = E[\reward_{\party, \execution_{\profile}}]$, \ie the expected rewards of $\party$ under
    profile $\profile$.
    \end{inparaenum}
    Finally, assume the block-proportional rewards (cf.
    Definition~\ref{def:proportional-rewards} for fixed block reward $\reward$)
    for which it also holds:
    $$
    \forall \executionTrace \; \forall \party \in \partySet: \reward_{\party, \executionTrace} =
    \left\{
    \begin{array}{ll}
        \proportionalRewardFunc(\chain_{\observer,\executionTrace}, \party) \cdot \totalReward_{\observer,\executionTrace} + \deposit_{\party},&\party \mbox{ produces no conflicting blocks during } \executionTrace\\
        0,&\mbox{otherwise}
    \end{array}
    \right.$$
    where  $\deposit_{\party}$ is a protocol-specific deposit value.
    SL-PoS with the above rewards and under the aforementioned externalities
    is not $(\epsilon, \infractionPredicate_\mathrm{conf})$-compliant (cf.
    Definition~\ref{def:compliant}) \wrt utility \emph{Reward} (cf.
    Definition~\ref{def:utility}) and, if $\reward > \cost$, it is also not
    $(\epsilon, \infractionPredicate_\mathrm{conf})$-compliant \wrt utility \emph{Profit},
    in both cases under $\adversary$ and if and only if
    $\epsilon < \mathsf{max}\{ \underset{\party \in \partySet}{\mathsf{max}}\{ \underset{\strategy_{\party} \in \strategySet_{-\infractionPredicate_\mathrm{conf}}}{\mathsf{max}}\{ \utilityBoostVal_{\party, \profile_{\strategy_{\party}}} \} \} - \rewardVal_{\party, \profile_{\proto}} \cdot \exchangeRateVal_{\profile_\proto}, 0 \} - \mathsf{negl}(\secparam)$.
\end{theorem}

\begin{proof}
    When a party $\party$ employs a non-compliant strategy $\strategy_{\party}$,
    it receives an external utility $\utilityBoostVal_{\party,
    \profile_{\strategy_{\party}}}$. Also, in that case, it produces
    conflicting blocks (due to the non-compliance property of
    $\strategy_{\party}$). Therefore, by definition of the above
    block-proportional rewards, $\reward_{\party, \profile_{\strategy_{\party}}}
    = 0$; in other words, when $\party$ employs $\strategy_{\party}$ and produces
    conflicting blocks, $\party$ forfeits the protocol's rewards (which include
    the original rewards plus the deposit $\deposit_{\party}$). Therefore, when
    $\party$ employs $\strategy_{\party}$ and all other parties employ $\proto$,
    $\party$'s utility is
    $\utility_{\party}(\profile_{\strategy_{\party}}) = \utilityBoostVal_{\party, \profile_{\strategy_{\party}}}$
    (cf. Definition~\ref{def:utility-external}).

    We also remind that (from the proof of
    Theorem~\ref{thm:compliant-ouroboros-synchronous}), $\party$ can bias the
    leader schedule with some negligible probability $negl(\secparam)$, if it
    controls a minority of power.

    Therefore, for any party $\party$ and strategy $\strategy_{\party}$,
    $\profile_{\strategy_{\party}}$ is directly $\epsilon$-reachable from
    $\profile_{\proto}$ if:
    \begin{equation*}
    \begin{split}
        \rewardVal_{\party, \profile_{\proto}} \cdot \exchangeRateVal_{\profile_\proto} + \mathsf{negl}(\secparam) + \epsilon &< \utilityBoostVal_{\party, \profile_{\strategy_{\party}}} \Leftrightarrow  \\
     \Leftrightarrow   \epsilon &< \utilityBoostVal_{\party, \profile_{\strategy_{\party}}} - \rewardVal_{\party, \profile_{\proto}} \cdot \exchangeRateVal_{\profile_\proto} - \mathsf{negl}(\secparam) \nonumber
    \end{split}
    \end{equation*}

    Across all parties and all non-compliant strategies, the maximum such $\epsilon$ is:
    \begin{align}
        \epsilon < \underset{\party \in \partySet}{\mathsf{max}}\{ \underset{\strategy_{\party} \in \strategySet_{-\infractionPredicate_{conf}}}{\mathsf{max}}\{ \utilityBoostVal_{\party, \profile_{\strategy_{\party}}} \} \} - \rewardVal_{\party, \profile_{\proto}} \cdot \exchangeRateVal_{\profile_\proto} - \mathsf{negl}(\secparam) \nonumber
    \end{align}

    We note that, as shown in
    Theorem~\ref{thm:compliant-ouroboros-synchronous}, SL-PoS with block
    proportional rewards under a synchronous router is an equilibrium, \ie the
    honest protocol yields the maximum rewards for each party compared to all
    other strategies. In the present setting, the honest protocol again yields
    the maximum utility, compared to all other \emph{compliant} strategies. To
    prove this it suffices to observe that, if a party does not produce
    conflicting blocks, its rewards are a linear function of the rewards of the
    setting of Theorem~\ref{thm:compliant-ouroboros-synchronous}; therefore,
    between compliant strategies, the honest protocol yields the maximum
    rewards (as shown in Theorem~\ref{thm:compliant-ouroboros-synchronous}).

    Therefore, the given bound of $\epsilon$ is bounded from below and, for
    $\epsilon$ less than this bound, there exists a party $\party$ that is
    incentivized to employ a (non-compliant) strategy $\strategy_{\party}$ and
    produce conflicting blocks, rendering $\proto$ not $(\epsilon,
    \infractionPredicate_{conf})$-compliant.
\end{proof}

The $\epsilon$ bounds in Theorems~\ref{thm:external-ouroboros}
and~\ref{thm:penalties-ouroboros} depend on the external utility boost
$\utilityBoostVal_{\party, \profile_{\strategy_{\party}}}$. This highlights the
inherent limitations of such systems' designers, since the bound depends on
external (to the protocol) parameters. Intuitively, these bounds show that
attacks which utilize the $\infractionPredicate_\mathrm{conf}$ infraction can
be prevented in two ways.
First, larger deposits increase the threshold that makes some attacks
profitable.
However, they also shut off small parties, with inadequate assets. Therefore, a
tradeoff exists in preventing such attacks and enabling participation.
Second, the longer an attack's duration, the more blocks an adversary
needs to produce, hence the larger the rewards that it forfeits. Typically, the
attack's duration depends on the required number of confirmations for a
transaction to be finalized. Therefore, different confirmation limits, \eg
based on a transaction's value, could satisfy the tradeoff between fast
settlement and security.

Considering the latter observation, we now briefly review users' behavior in
SL-PoS (cf. Section~\ref{subsec:single-leader-pos}) under deposits and
penalties. In an SL-PoS execution, the percentage of parties that actively
participate during each epoch is identifiable via the block density and the
number of empty slots (when no block is diffused). Therefore, it is
possible to estimate the level of double-signing that a party
needs to perform to mount a double-spending attack, and then enforce a
transaction finalization rule to dis-incentivize such attacks.

Let $\party$ be a user of an SL-PoS ledger. $\party$ requires
$k$ confirmations, \ie finalizes a transaction after it is
``buried'' under $k$ blocks. Let $\tau$ be a transaction, published
on slot $\slot$, with value $v_\tau$. After $l$ slots,
$\tau$ is buried under $b$ blocks, with $b = x \cdot l$ for some $x \in
(0, 1)$. In case we have full participation in the protocol and the adversary
is bounded by $\frac{1}{2}$, it holds $x > \frac{1}{2}$; in the rest of the
section, we will focus on this setting.
Observe that $(1 - x) \cdot 100$\% of slots will be -- seemingly -- empty.
$\party$ will (on expectation) confirm $\tau$ after $\frac{1}{x} \cdot k$
slots, \ie when $k$ blocks are produced; of these, $\frac{1 - x}{x} \cdot k$
are empty.

Let $\attacker$ be a party that wants to double-spend $\tau$.
$\attacker$ should produce a private chain with at least $k$ blocks. Of these,
at most $\frac{1 - x}{x} \cdot k$ correspond to the respective empty slots,
while $k - \frac{1 - x}{x} \cdot k = \frac{2 \cdot x - 1}{x} \cdot k$ conflict
with existing blocks, \ie are evidence of infraction.
Let $d$ be a deposit amount, which corresponds to a single slot. Thus, for a
period of $t$ slots, the total deposited assets $D = t \cdot d$ are distributed evenly
across all slots.
$\attacker$ can be penalized only for infraction blocks, \ie for
slots which showcase conflicting blocks. In a range of
$\frac{1}{x} \cdot k$ slots, infraction slots are $\frac{2 \cdot x - 1}{x} \cdot k$.
Therefore, $\attacker$ forfeits at most $\frac{2 \cdot x - 1}{x} \cdot k \cdot
d$ in deposit and $\frac{2 \cdot x - 1}{x} \cdot k \cdot \reward$ in rewards
that correspond to infraction blocks.
Thus, if $v_\tau > \frac{2 \cdot x - 1}{x} \cdot k \cdot (d + \reward)$,
$\attacker$ can profitably double-spend $\tau$.
Consequently, depending on the amount $d$ of deposit per slot, the block reward
$\reward$, and the rate $(1 - x)$ of empty slots, for a transaction $\tau$ with
value $v_\tau$, $\party$ should set the confirmation window's size to:

\begin{align}
    k_\tau > \frac{v_\tau}{\frac{2 \cdot x - 1}{x} \cdot (d + \reward)}
\end{align}

Finally, the system should allow each participant to withdraw their deposit
after some time. However, it should also enforce \emph{some} time limit, such
deposits are adequate to enforce (possible) penalties.
Intuitively, a party $\party$ should be able to withdraw a deposit amount that
corresponds to a slot $\slot$, only if no transaction exists, such that $\slot$
is part of the window of size $k$ (computed as above). In other words,
$\party$'s deposit should be enough to cover all slots, which $\party$ has led
and which are part of the confirmations' window of at least one non-finalized
transaction.

\paragraph{Remark.}
    The correctness of the penalty enforcement mechanism should be verifiable
    by only parsing the ledger, \ie in a non-interactive manner. Interestingly,
    it might be impossible to create such (non-interactive) proofs of
    misbehavior for some predicates. For example, in a semi-synchronous network,
    block withholding (\eg selfish mining) is indistinguishable from
    behavior that occurs due to benign network delays.
    In those cases, the protocol's designer
    cannot rely on reward distribution and penalization to enforce compliance;
    instead, the utility should reflect resistance to such infractions, \eg via
    extra costs.
    Additionally, penalties are not applicable if the block's creators
    cannot be identified, \eg as in PoW, where block producers
    are decoupled from the users, and anonymous
    protocols~\cite{SP:MGGR13,SP:BCGGMT14,EC:GanOrlTsc19,SP:KKKZ19}.

\section{Conclusion}\label{sec:conclusion}

Our work explores the ability of blockchain designs to disincentivize
infractions. Inspired by Nash dynamics, we present a model for
strategic compliance, \ie rational
participant behavior that, while it does not
violate well-defined
properties, it may potentially exhibit ``compliant'' protocol deviations.
We focus primarily on two blockchain-related infractions,
abstaining and same-origin conflicting blocks. Given two
types of utility, ``rewards''  and ``profit'', we
analyze compliance of various protocols, offering both
positive and negative results:
\begin{inparaenum}[i)]
    \item PoS blockchains that employ resource-proportional rewards, \ie which
        depend \emph{solely} on a party's power, are compliant \wrt rewards,
        but non-compliant \wrt profit, as they incentivize abstaining;
    \item PoW systems with block-proportional rewards are compliant;
    \item compliance of PoS systems, which enforce a single participant per
        slot, depends on the network's lossiness;
    \item PoS systems where multiple participants per slot may act are
        non-compliant, under certain network routing conditions, as they
        incentivize producing conflicting blocks;
    \item PoS systems that  are compliant \wrt producing conflicting blocks, but
        non-compliant \wrt abstaining or performing a selfish signing attack.
\end{inparaenum}
Finally, we consider externalities, namely the assets' market price
and external proceedings of a successful attack. We show that, if the market
does not respond decisively when an attack occurs, infractions that yield such
proceedings may render protocols non-compliant. We also suggest a mitigation that
combines two approaches in the context of longest chain protocols:
i) deposits, which are potentially used for penalizing misbehavior;
ii) an adaptive transaction finality rule, applied by users, which increases the
number of blocks an attacker has to create to discard a
finalized high-value transaction.

Our work opens various lines for future work. First, alternative infraction
predicates and utilities could be explored, \eg to capture attacks such as
selfish mining and double spending. For instance, given that Bitcoin is not an
equilibrium \wrt relative rewards due to selfish mining, proving compliance for
this utility would be a rather promising result. Also we only consider rewards
that originate from the system; however, in real world systems, a party may
also receive assets via \emph{reward transfers} from another party. Introducing
transfers would enrich the strategy set, \eg allowing a party to
temporarily reduce its rewards for
long term profits, possibly affecting compliance in terms of coalition
forming. Finally, Section~\ref{sec:externalities} assumes that the exchange
rate and protocol rewards are independent random variables; future work could
explore possible correlations between the two and produce an
analysis that can be (arguably) closer to the real world. In addition, our
analysis treated changes in the market price as an effect of infractions;
an interesting question is whether price volatility (\eg as examined
in~\cite{DBLP:conf/sigecom/NodaOH20}) can be the cause of non-compliance.

\def\doi#1{\url{https://doi.org/#1}}
\bibliographystyle{plain}
\bibliography{additional,abbrev0,crypto_crossref}

\appendix

\section{Hybrid PoS and Finality Gadgets}\label{subsec:gasper}

Gasper~\cite{buterin2020combining} is a proof-of-stake blockchain protocol that
combines ideas from Casper FFG~\cite{buterin2017casper,casper-incentives} and
the (last message driven) GHOST fork-choice rule~\cite{FC:SomZoh15}. It is
of special interest, since it is the protocol of choice for Ethereum's updated
version, 2.0. Interestingly, Eth 2.0 aims at preventing malicious behavior via
penalties on incentives, which take the form of confiscating the deposit of a
(provably) misbehaving party. However, the high level of required
deposit\footnote{The required deposit has been set to
$32$ETH~\cite{gasper-rewards}, which as of April 2021 is equivalent to
\$$75000$. [\url{https://coinmarketcap.com}]} prevents small investors from
participating in the protocol. Therefore, compliance is a helpful property in
identifying whether infraction is viable under the plain protocol (\ie without
penalties) and, if so, how high the penalty should be set, \st it both
incentivizes compliance and allows as many parties as possible to participate.
Although Gasper is a consensus protocol in itself, our analysis applies to most
systems that employ finality gadgets like Winkle~\cite{EPRINT:AzoDanNik19},
Afgjort~\cite{SCN:DMMNT20}, and Fantomette~\cite{azouvi2018betting}.

As with the other PoS protocols (Section~\ref{subsec:pos}),
in Gasper time is measured in \emph{slots} and a constant number $\epochLength$
of slots defines an \emph{epoch}. In Gasper, a slot $\slot$ which has number
$\epoch \cdot \epochLength + k$, $k \in [0, \epochLength-1]$ belongs to epoch
$\epoch$, denoted as $\mathsf{ep}(\slot) = \epoch$. The \emph{view} of a block
$\block$ is the view consisting of $\block$ and all its ancestors, where an
ancestor of $\block$ is a block that is reachable from $\block$ following
parent-child edges on the chain.

For a block $\block$ and an epoch $\epoch$ there is a well-defined
\emph{$\epoch$-th epoch boundary block} of $\block$, denoted by
$\mathsf{EBB}(\block, \epoch)$. In particular, $\mathsf{EBB}(\block, \epoch)$
is the block with the highest slot number that is $\leq \epoch \cdot
\epochLength$ in the chain determined by $\block$. The \emph{last epoch
boundary block} of $\block$ is denoted by $\mathsf{LEBB}(\block)$. Since a
block $\block'$ can be an epoch boundary block of $\block$ for multiple epochs,
an \emph{epoch boundary pair} $(\block', \epoch')$ is used. These pairs play
the role of checkpoints, \ie a party that follows the protocol never discards a
justified block (and, consequently, its ancestor block back to the genesis).

In each slot, a subset of protocol participants, called \emph{validators},
forms a committee. Each validator belongs to a single committee per epoch. The
first committee member is the \emph{proposer}, \ie is responsible for creating
a new block; similar to the protocols of
Section~\ref{subsec:single-leader-pos}, a single proposer per slot is assigned.
Then, each member of the committee attests to the block that is the head of its
chain, using a variation of GHOST. Each attestation $\attestation$ is defined
as:
\begin{align*}
    \attestation = \langle \slot_{\attestation}, \block_{\attestation}, \mathsf{LJ}(\attestation)\overset{\validator}{\rightarrow}\mathsf{LE}(\attestation) \rangle
\end{align*}
where:
\begin{inparaenum}[(i)]
    \item $\slot_{\attestation} = \epoch \cdot \epochLength + k$ is the slot of
        epoch $\epoch$ during which the validator $\validator$ makes the
        attestation;
    \item the $\block_{\attestation}$ block that $\attestation$ attests to;
    \item
        $\mathsf{LJ}(\attestation)\overset{\validator}{\rightarrow}\mathsf{LE}(\attestation)$
        is a \emph{checkpoint edge}, where $\mathsf{LJ}(\attestation)$ and
        $\mathsf{LE}(\attestation)$ are epoch boundary pairs that will be
        defined shortly.
\end{inparaenum}
Therefore, $\attestation$ acts as a ``GHOST vote'' for block
$\block_{\attestation}$ and as a ``Casper FFG vote'' for the transition from
the epoch boundary pair $\mathsf{LJ}(\attestation)$ to
$\mathsf{LE}(\attestation)$.  $\mathsf{slot}(\attestation)$ denotes the slot of
$\attestation$ and two attestations are equal if their hash is equal (for some
hash function $\hash$).

If a set of validators with weight more than $\frac{2}{3}$ of the total
validating stake attests to an edge $(\block',
\epoch')\overset{\validator}{\rightarrow}(\block, \epoch)$, then a
\emph{supermajority link} is formed, denoted by $(\block',
\epoch')\overset{J}{\rightarrow}(\block, \epoch)$.

To define the pairs $\mathsf{LJ}(\attestation)$ and
$\mathsf{LE}(\attestation)$, the concept of \emph{justified pairs} is
introduced. Given a party $\party$'s view of the execution trace
$\executionTrace^{\party}$, the set of justified pairs
$J(\executionTrace^{\party})$ is defined recursively as follows:
\begin{itemize}
    \item the pair $(\genesis, 0)$ is in $J(\executionTrace^{\party})$, where
        $\genesis$ is the genesis block;
    \item if $(\block', \epoch') \in J(\executionTrace^{\party})$ and $(\block',
        \epoch') \overset{J}{\rightarrow} (\block, \epoch)$, then $(\block,
        \epoch) \in J(\executionTrace^{\party})$.
\end{itemize}
Now given an attestation $\attestation = \langle \slot_{\attestation},
\block_{\attestation},
\mathsf{LJ}(\attestation)\overset{\validator}{\rightarrow}\mathsf{LE}(\attestation)
\rangle$, the following holds:
\begin{itemize}
    \item $\mathsf{LJ}(\attestation)$ is the \emph{last justified pair of
        $\attestation$} (in terms of epoch number) \wrt the view of
        $\mathsf{LEBB}(\block_{\attestation})$;
    \item $\mathsf{LE}(\attestation)$ is the \emph{last epoch boundary pair of
        $\attestation$}, \ie $\mathsf{LE}(\attestation) =
        (\mathsf{LEBB}(\block_{\attestation}), \epoch_{\attestation})$.
\end{itemize}

Finally, with $(\block_k, \epoch + k)$ being the last justified pair in
$J(\executionTrace^{\party})$ (in terms of epoch number), a pair $(\block,
\epoch)$ is \emph{$k$-finalized} in $\executionTrace^{\party}$ if:
\begin{enumerate}
    \item $(\block, \epoch)=(\genesis, 0)$, or
    \item there exist $k+1$ (justified) adjacent epoch boundary pairs $(\block, \epoch),
        (\block_1, \epoch+1), \ldots, (\block_k, \epoch + k) \in J(\executionTrace^{\party})$
        and there also exists $(\block, \epoch)\overset{J}{\rightarrow}(\block_k, \epoch +
        k)$.
\end{enumerate}

In Gasper, each protocol message $\mesg$ is an attestation $\attestation$.
As with PoS protocols, the oracle $\oracle_\proto$ produces a signature for a
submitted attestation and a party can produce as many (valid) signatures per
round as queries to $\oracle_\proto$.

With the introduction of attestations, we need to slightly adapt the
standard blockchain infraction predicate (cf.
Section~\ref{sec:blockchain-infraction-predicate}) for Gasper. The deviant behavior for
validators is expressed by Gasper's infraction predicate
(Definition~\ref{def:gasper-infraction-validator}), which defines three
clauses, based on Gasper's ``slashing conditions''~\cite{buterin2020combining}.
Intuitively, an infraction occurs if a validator:
\begin{inparaenum}[i)]
    \item produces two conflicting attestations for the same (attestation) epoch;
    \item attempts a ``long-range'' attestation, \ie produces two attestations \st
        the first's edge source is older than the second's and its edge target
        is newer than the second's;
    \item does not produce an attestation for some epoch.
\end{inparaenum}

\begin{definition}[Gasper Infraction Predicate (Validators)]\label{def:gasper-infraction-validator}
    Given a validator party $\validator$ and an execution trace
    $\executionTrace$, $\infractionPredicate_{Gasper}(\executionTrace_{\env, \adversary,
    \slot}, \validator) = 1$ at slot $\slot$ if one of the following conditions
    holds:
    \begin{enumerate}
        \item $\exists \attestation, \attestation':
            (\messageValidityPredicate(\executionTrace_{\env, \sigma, \slot}^\validator, \attestation) = \messageValidityPredicate(\executionTrace_{\env, \sigma, \slot}^\validator, \attestation') = 1) \land
            (\mathsf{creator}(\attestation) = \mathsf{creator}(\attestation') =
            \validator) \land (\mathsf{slot}(\attestation) =
            \mathsf{slot}(\attestation')) \land (\attestation \neq
            \attestation')$;
        \item $\exists
            \attestation = \langle \cdot, \cdot, (\block_s, \epoch_s) \overset{\validator}{\rightarrow} (\block_t, \epoch_t) \rangle,
            \attestation' = \langle \cdot, \cdot, (\block'_s, \epoch'_s) \overset{\validator}{\rightarrow} (\block'_t, \epoch'_t) \rangle:
            (\mathsf{creator}(\attestation) = \mathsf{creator}(\attestation') =
            \validator) \land (\epoch_s < \epoch'_s) \land (\epoch_t > \epoch'_t)$;
        \item $\validator$ is a committee member for slot $\slot$ and makes
            \emph{no} queries to the oracle $\oracle_\proto$ during $\slot$.
    \end{enumerate}
\end{definition}

Regarding rewards, we assume those of Eth $2.0$, which implements Gasper. Given the chain
$\chain_{\observer}$ output by the observer $\observer$, parties receive
rewards as validators and/or as proposers.
A validator receives a fixed reward\footnote{In
practice, both the validation and proposer rewards are not fixed, but depend on
the total amount of stake in the system and the number of included attestations
respectively. Our simplification though eases analysis and expresses the
setting where the total stake and the rate of produced attestations per round
are fixed throughout the execution.} $\reward_v$ for each attestation
$\attestation$ it produces, which satisfies the following conditions:
\begin{inparaenum}[i)]
    \item the block that $\attestation$ attests to is finalized in
        $\chain_{\observer}$;
    \item $\attestation$ is published in $\chain_{\observer}$.
\end{inparaenum}
A proposer receives a fixed reward $\reward_p$ for each block $\block$ it
produces which is finalized in $\chain_{\observer}$.

Consequently, the analysis of Section~\ref{subsec:single-leader-pos} applies
directly on Gasper, for the updated infraction predicate
(Definition~\ref{def:gasper-infraction-validator}). Specifically, when a party
is chosen to participate, either as proposer or validator, no other party may
also be elected for the same role; therefore, Gasper's participation schedule
is equivalent to single-leader PoS protocols. Additionally, Gasper rewards are
akin to block proportional rewards, as parties are rewarded for producing
attestations which are eventually accepted by $\observer$. Therefore, applying
the analysis of Section~\ref{subsec:single-leader-pos}, Gasper is compliant
(\wrt the Gasper infraction predicate) under a synchronous network, but it is
not compliant under a lossy network.

\end{document}